\documentclass[12pt]{article}
\usepackage{amssymb, amsmath, amsthm}
\usepackage{natbib}
\usepackage{setspace}
\usepackage{multirow}
\usepackage{graphicx}
\usepackage{caption}
\usepackage{titling}
\usepackage{subcaption}
\usepackage{algorithm}
\usepackage{algpseudocode}
\usepackage{url}
\usepackage{enumitem}

\usepackage{natbib}

\usepackage{color}

\usepackage[colorlinks,citecolor=blue,urlcolor=blue]{hyperref}
\RequirePackage{hypernat}

\usepackage{imakeidx}
\makeindex

\usepackage{booktabs}

\newcommand{\argmin}{\operatorname*{arg \ min}}

\usepackage[usenames,dvipsnames]{xcolor}

\newcommand{\EE}{{\mathbb E}}

\newcommand{\RR}{\mathbb{R}}

\theoremstyle{plain}

\newtheorem{prop}{Proposition}

\newtheorem{remark}{Remark}
\newtheorem{lemma}{Lemma}

\makeatletter
\newcommand{\myitem}[1]{%
\item[#1]\protected@edef\@currentlabel{#1}%
}
\makeatother

\newcounter{alphasect}
\def\alphainsection{0}

\let\oldsection=\section
\def\section{%
  \ifnum\alphainsection=1%
    \addtocounter{alphasect}{1}
  \fi%
\oldsection}%

\renewcommand\thesection{%
 \ifnum\alphainsection=1%
   \Alph{alphasect}%
 \else
   \arabic{section}%
 \fi%
}%

\newenvironment{alphasection}{%
  \ifnum\alphainsection=1%
    \errhelp={Let other blocks end at the beginning of the next block.}
    \errmessage{Nested Alpha section not allowed}
  \fi%
  \setcounter{alphasect}{0}
  \def\alphainsection{1}
}{%
  \setcounter{alphasect}{0}
  \def\alphainsection{0}
}%

\title{Fitted value shrinkage}
\author{Daeyoung Ham \& Adam J. Rothman\\
School of Statistics, University of Minnesota\\
\texttt{ham00024@umn.edu}, \texttt{arothman@umn.edu}}
\date{}

\DeclareUnicodeCharacter{2009}{FIXME}
\def\spacingset#1{\renewcommand{\baselinestretch}%
{#1}\small\normalsize} \spacingset{1}
\begin{document}

\maketitle

\begin{abstract}
We propose a penalized least-squares method to fit the linear regression model with fitted values that are invariant to invertible linear transformations of the design matrix. This invariance is important, for example, when practitioners have categorical predictors and interactions. Our method has the same computational cost as ridge-penalized least squares, which lacks this invariance. We
derive the expected squared distance between 
the vector of population fitted values and its shrinkage estimator as well as the tuning parameter value that minimizes this expectation.  In addition to using cross validation, we construct two estimators of this optimal tuning parameter value and study their asymptotic properties.
Our numerical experiments and data examples show that our method performs similarly to ridge-penalized least-squares.
\end{abstract}
\noindent%
\vfill

\spacingset{1.9}

\section{Introduction}
We will introduce a new shrinkage strategy for fitting linear regression models, which assume that the measured response
for $n$ subjects is a realization of the random vector
\begin{equation} \label{gencasesvec}
Y = X \beta +   \varepsilon,
\end{equation}
where $X\in\mathbb{R}^{n\times p}$ is the nonrandom known design matrix with ones in 
its first column and with values of the predictors in its remaining columns; $  \beta\in\mathbb{R}^p$ is an unknown vector of regression coefficients; 
and $  \varepsilon$ has iid entries with mean zero 
and unknown variance $\sigma^2\in(0,\infty)$. 

We will consider fitting \eqref{gencasesvec}
in both low and high-dimensional settings, where the second
scenario typically has ${\rm rank}(X) < p$.
If ${\rm rank}(X) < p$, then it is well known that $  \beta$ is not identifiable
in \eqref{gencasesvec}, i.e. there exists a 
$\tilde{  \beta} \neq   \beta$ such that $X  \beta=X\tilde{  \beta}$.  Similarly, if ${\rm rank}(X)< p$, then
there are infinitely many solutions to the least-squares problem:
$\argmin_{  b\in\mathbb{R}^p} \|Y-X  b\|^2$.
Given this issue (which is unavoidable in high dimensions), 
our inferential target is $X  \beta$, 
which is the expected value of the response for the 
$n$ subjects.  

To describe least squares estimators whether
${\rm rank}(X)< p$ or ${\rm rank}(X) = p$, we will use the reduced
singular value decomposition of $X$.  Let $q={\rm rank}(X)$.
Then $X = UDV'$, where $U\in\mathbb{R}^{n\times q}$ with $U'U=I_q$; 
$V\in\mathbb{R}^{p\times q}$ with $V'V=I_q$; and
$D\in\mathbb{R}^{q\times q}$ is diagonal with positive diagonal entries.
The Moore--Penrose generalized inverse of $X$ is $X^{-} = V D^{-1} U'$
and a least-squares estimator of $  \beta$ is $\hat{  \beta}=X^{-}Y$.
The vector of fitted values is $X\hat{ \beta} = XX^{-}Y = P_X Y$, where
$P_X = XX^{-}=UU'$.  If ${\rm rank}(X)=p$, then $P_X =  X( X' X)^{-1} X'$.

A nice property of this least-squares method is that its fitted
values are invariant
to invertible linear transformations of the design matrix.
Suppose that we replace $X$ by $X_{\bullet} = XT$, where $T\in\mathbb{R}^{p\times p}$
is invertible.  Then $X = X_{\bullet} T^{-1}$. So \eqref{gencasesvec} is
$$
Y = X  \beta +   \varepsilon 
= X_{\bullet} T^{-1} \beta +   \varepsilon = X_{\bullet}   \beta_{\bullet} +   \varepsilon, 
$$
where $  \beta_{\bullet} = T^{-1}  \beta$.
We estimate $X  \beta = X_{\bullet}   \beta_{\bullet}$ 
with $P_X Y = P_{X_{\bullet}} Y$, so the fitted values did not change
by changing $X$ to $X_{{\bullet}}$.

Fitting \eqref{gencasesvec} by penalized least-squares has been studied by many scholars.  Well-studied penalties include the ridge penalty \citep{hoerl70}, 
the bridge/lasso penalty \citep{frank93, tibs96}, the adaptive lasso penalty \citep{zou06_adapt}, the SCAD penalty \citep{fan01}, 
and the MCP penalty \citep{zhang10}. Unfortunately, these methods' fitted values are not invariant to invertible linear transformations of $X$.  
The lack of invariance is particularly problematic
when a practitioner wants to find a set of linear combinations of a subset of the predictors to reduce multicollinearity, because a rotational change in the model matrix can substantially change the predicted values of the penalized methods that lack invariance.
In addition, invariance to invertible linear transformations is also important in any application that includes categorical predictors (with three or more categories) and their interactions: the model fit should not change by changing the coding scheme used to represent these variables in the design matrix.  The Group Lasso \citep{Grouplasso.YuanandLin} and its variants with different standardization \citep{choi2012lasso,simon2012standardization} could also be used when there are categorical predictors, but these methods still lack invariance and our simulation and data examples illustrate their instability.
This lack of invariance is also present in principal components regression \citep{hotelling33}, and partial least squares \citep{wold66}.

We illustrate that a change in the reference level coding of categorical predictors (see Section \ref{factorsimul}) and that a rotational transformation of the design matrix 
that removes zeros from the regression coefficient 
vector (see Section \ref{fullrankmtx} and \ref{fullrankmtx.high}) leads the regression shrinkage methods that lack this invariance to perform worse.  
In contrast, our method performed the same before and after these transformations.
We also propose a method to estimate the regression's error variance in high dimensions that may be of independent interest. 

\section{A new shrinkage method for linear regression with
invariance} \label{sec:main.method}
\subsection{Method description}\label{method}
To preserve the invariance to invertible linear transformations of the design matrix
discussed in the previous section, 
we will use penalties that can be expressed as a function of the $n$-dimensional vector
$Xb$, where $b$ is the optimization variable
corresponding to $\beta$.  To construct
our estimator, we start with the following
penalized least squares optimization:
\begin{align} \label{mainoptlam}
  \argmin_{ b \in \mathbb{R}^p}  \left\{ \|  Y - Xb\|^2 + \lambda \|Xb - \bar {  Y} 1_n\|^2 \right\},  
\end{align}
where $\bar {  Y} = 1_n'Y/n$; $1_n'=(1,\dots,1)\in\mathbb{R}^n$; and $\lambda \in [0,\infty)$ is a tuning parameter. 
As $\lambda$ increases, the fitted values are shrunk towards
the intercept-only model's fitted values $\bar {  Y} {  1_n}$.  
Let $\gamma = 1/(1+\lambda)$.  We can express the optimization in \eqref{mainoptlam} as
\begin{equation} \label{mainopt}
\hat{  \beta}^{(\gamma)} \in 
\argmin_{{  b} \in \mathbb{R}^p}  \left\{ \gamma\|{  Y} - Xb\|^2 + (1-\gamma) \|Xb - \bar {  Y} {  1_n}\|^2 \right\},
\end{equation}
where $\gamma \in [0,1]$. The $\in$ is used because
there are infinitely many global minimizers for the optimization in \eqref{mainopt}
when ${\rm rank}(X) < p$.  To preserve 
invariance and obtain the minimum 2-norm solution, 
we define our estimator as
$$
\hat{ \beta}^{(\gamma)} = {  X}^{-}\left\{ \gamma {  Y} + (1-\gamma) \bar {  Y} {  1_n} \right\},
$$
which is a global minimizer of \eqref{mainopt}
that uses the Moore-Penrose generalized inverse of $X$.
Let ${  P_1} = {  1_n} ({  1_n}'{  1_n})^{-1} {  1_n}'$.
The estimator of ${  X}{  \beta}$ is
\begin{equation} \label{mainestimator}
X\hat{  \beta}^{(\gamma)} = \gamma P_X {  Y}  + (1-\gamma) {  P_{1}} {  Y},
\end{equation}
which is simply a convex combination of the least-squares fitted values ${P_X} {  Y}$
and the intercept-only model's fitted values ${  P_1} {  Y} = \bar {  Y} {  1_n}$. 
Since $P_X= P_{X_{\bullet}}$, where $X_{\bullet} = XT$ with $T\in\mathbb{R}^{p\times p}$ invertible, $X\hat{  \beta}^{(\gamma)}$ is invariant to invertible linear transformations
of $X$.

Due to its computational simplicity, $\hat{  \beta}^{(\gamma)}$ is
a natural competitor to ridge-penalized least squares, which lacks this 
invariance property.  Both methods can be computed efficiently 
when $p$ is much larger than $n$ by using the reduced singular value
decomposition of $X$ \citep{Hastie2004RidgeSVD}.  Specifically,
they both cost $O(n \ {\rm rank}^2(X))$ 
floating-point operations.
If $\gamma=0$ for our method
and $\lambda\rightarrow \infty$ 
for ridge-penalized least squares (without intercept
penalization), then both procedures fit the intercept-only model.

We will derive
an optimal value of $\gamma$ that minimizes 
$\mathbb{E}\|X\hat{  \beta}^{(\gamma)}-X \beta\|^2$
and propose two estimators of it: 
one for low dimensions 
and one for high dimensions.  We also explore
using cross validation to select $\gamma$ when the response and predictor measurement pairs are drawn from a joint
distribution.  Conveniently, our results 
generalize to shrinkage towards a submodel's
fitted values $P_{X_0}Y$, where
$X_0$ is a matrix with a proper subset of the 
columns of $X$ (see Section \ref{submodelshrink} of the Supplementary material).

\subsection{Related work}
\citet{Copas1997} proposed to predict a future value of the response
for the $i$th subject with a convex combination of its fitted value (from ordinary least squares) and $\bar Y$.   Although our methods are related,
\citet{Copas1997} used a future-response-value prediction paradigm and
did not establish a theoretical analysis of his approach. 

\citet{azriel.linear.projection} study the estimation of $a'\beta$
when $p > n$ and $\beta$ is not sparse.  They establish a condition 
for which $a'\beta$ is identifiable and show that using least-squares
(with a pseudoinverse) has an optimal property when the errors are Gaussian.
We prove that our shrinkage method outperforms least squares in same-$X$ prediction and we illustrate it in our numerical examples.  
\citet{zhao.linearfunctional} also study the estimation of $a'\beta$, but they use procedures
that are not invariant to invertible linear
transformations of $X$.

\section{Theoretical properties of the method}\label{sec:theory.main}
Given the lack of identifiability of $  \beta$ in high dimensions, we investigate the estimation of the $n$-dimensional vector $X  \beta$ with $X\hat{  \beta}^{(\gamma)}$.  This
is an example of same-X prediction \citep{rosset20}. 
It is related to predicting near $X$ when $p>n$
\citep*[Proposition 3.4]{cfr13}.
In contrast to a random-design analysis, our fixed-design analysis allows a practitioner to control a subset of the columns of $X$, which would be
essential for designed experiments.

Suppose that the linear regression model specified in \eqref{gencasesvec} is 
true (this model did not specify an error distribution, just that they
are iid mean 0 and variance $\sigma^2\in(0,\infty)$).
We define $\mu=X \beta$ and 
we assume that ${\rm rank}(X)\ge 2$
throughout the paper, so $X$ cannot
correspond to an intercept-only model.
Then we have the following result: 
\begin{prop}\label{emsefvs2}
For all $(n, p) \in \{1,2,\ldots\}\times \{1,2,\ldots\}$,
$$
\mathbb{E}\|X\hat{ \beta}^{(\gamma)}-X{ \beta}\|^2=\sigma^2(\gamma^2 r+1-\gamma^2)+(1-\gamma)^2\|  \mu-P_1  \mu\|^2.
$$
\end{prop}
\noindent
The proof of Proposition \ref{emsefvs2} is in Section \ref{appendthm1} of the Supplementary material. When $\gamma=1$, which is least squares, 
$\mathbb{E}\|X\hat{  \beta}^{(1)}- X \beta\|^2 = \sigma^2 {\rm rank}(X)$.

The right side of the equality in Proposition \ref{emsefvs2} is minimized when $\gamma = \gamma_{\rm opt}$,
where
\begin{align}\label{gamopt}
  \gamma_{\rm opt}=\frac{\| \mu - P_1  \mu\|^2}
{\sigma^2 ({\rm rank}(X) -1) + \|  \mu - P_1  \mu\|^2}. 
\end{align}
So the best our procedure could do is when $\|  \mu - P_1  \mu\|^2=0$ (that is, the intercept-only model is correct), in which case
$\gamma_{\rm opt}=0$ and $\mathbb{E}\|X\hat{  \beta}^{(0)}- X  \beta\|^2 = \sigma^2$. The expression for $\gamma_{\rm opt}$ enables us to construct a sample-based one-step estimator of it,  which is more computationally efficiency than cross validation.

\section{Selection of $\gamma$} \label{tuningselection}
\subsection{Low-dimensional case}

Let $\hat\sigma^2  = \|Y - P_X Y\|^2/(n-{\rm rank}(X))$,
which is an unbiased estimator of $\sigma^2$.
To construct an estimator of $\gamma_{\rm opt}$, we 
use the ratio of $\|P_X Y - P_1 Y\|^2 -\hat\sigma^2 ({\rm rank}(X)-1)$, which is an unbiased estimator of $\gamma_{\rm opt}$'s numerator, to $\|P_X Y - P_1 Y\|^2$, which is an unbiased estimator of $\gamma_{\rm opt}$'s denominator.
This ratio estimator can be expressed as 
\begin{align*}
\frac{\|P_XY - P_1Y\|^2 -\hat\sigma^2 ({\rm rank}(X)-1)}{\|P_XY - P_1Y\|^2} 
& = 1-\frac{\hat\sigma^2({\rm rank}(X)-1)}{\|P_XY - P_1Y\|^2}\\
& = 1-1/F,
\end{align*}
where $F$ is the F statistic that compares the intercept-only model to the full model:
\begin{align}
F & = \frac{(\|Y - P_1Y\|^2 - \|Y - P_XY\|^2)/({\rm rank}(X)-1)}{\|Y - P_XY\|^2/(n-{\rm rank}(X))}\nonumber\\
 & = \frac{\|P_XY - P_1Y\|^2}{\hat\sigma^2({\rm rank}(X)-1)}.\label{fstat}
\end{align}

Since $F$ could be realized less than one (which 
corresponds to a fail-to-reject the intercept-only model situation), we define our 
estimator of $\gamma_{\rm opt}$ to be 
\begin{align}\label{gamhat}
  \hat\gamma = (1-1/F) \cdot 1(F > 1)  
\end{align}
\noindent
If the regression errors in \eqref{gencasesvec} are Normal,
$n > {\rm rank}(X)$, and ${\rm rank}(X) > 1$,
then $F$ has a non-central F-distribution with 
degrees of freedom parameters ${\rm rank}(X)-1$ and $n-{\rm rank}(X)$; and noncentrality parameter $\|  \mu - P_1   \mu\|^2/\sigma^2$.
Larger realizations of $F$ correspond to worse intercept-only model fits compared to the
full model, which makes $\hat\gamma$ closer to 1.

We also explore two additional estimators of $\gamma_{\rm opt}$:
\begin{align}
    &\hat\gamma_{90}=(1-1/F) \cdot 1(F\geq f_{0.9}),\label{fq90}\\
    &\hat\gamma_{95}=(1-1/F) \cdot 1(F\geq f_{0.95}),\label{fq95}
\end{align}
where $f_{90}$ and $f_{95}$ are the $0.9$ and $0.95$ quantiles of the central F-distribution with degrees of freedom ${\rm rank}(X)-1$ and $n-{\rm rank}(X)$.
These estimators may perform better when $\gamma_{\rm opt}$ is near zero because they have a greater probability of estimating $\gamma_{\rm opt}$ as zero than $\hat\gamma$ has.

Interestingly, \citet{Copas1997} proposed to predict a 
future response value for the $i$th subject
with $(1-\rho)\hat\beta'x_i  + \rho\bar Y$, where $\rho\in[0,1]$ is estimated and $\hat\beta$ is the ordinary least-squares estimator.  
They derived $1/F$ as an estimator of $\rho$ from the normal equations for the regression of 
$Y_{{\rm new}, i}$ on $\hat\beta'x_i$, $(i=1,\ldots, n)$, 
where $Y_{{\rm new}, i}$ is an independent copy of $Y_i$.
They also discussed using truncation to ensure their estimator
of $\rho$ is in $[0,1]$.

\subsection{Consistency and the 
convergence rate of $\hat\gamma$}
We analyze the asymptotic performance 
of $\hat\gamma$ when the
data are generated from 
\eqref{gencasesvec} and $n$ and $p$ grow together.  
Define $r = {\rm rank}(X)$ and $\delta^2=\|  \mu-P_1 \mu\|^2$.  The optimal tuning parameter 
value is a function of $r$ and $\delta^2$, 
so its value in the limit 
will depend on these sequences.

\begin{prop}\label{consistency}
Assume that the data-generating
model in \eqref{gencasesvec} is correct, that the errors
have a finite fourth moment, and that $r \geq 2$. 
If $p/n\rightarrow\tau\in [0,1)$ and 
either $r\rightarrow \infty$
or $\delta^2\rightarrow \infty$,
then $\hat\gamma -  \gamma_{\rm opt} \rightarrow_P 0$ as $n\rightarrow \infty$.
\end{prop}
\noindent
The proof of Proposition \ref{consistency} is in Section \ref{appendthm2} of the Supplementary material. 
We see that consistency is possible whether
the design matrix rank $r$ grows.  If $r$ is 
bounded, then consistency requires 
$\delta^2=\| \mu-P_1  \mu\|^2\rightarrow \infty$, 
which is reasonable even when the intercept-only model is a good approximation because $n$ is growing.
One can also show consistency of $\hat\gamma_{90}$ and $\hat\gamma_{95}$ with $\delta^2=o(r)$ added to the assumptions for Proposition \ref{consistency}.

Next, we establish a bound on the rate of convergence of $\hat\gamma$ with further assumptions on the design matrix $X$ and the error $  \varepsilon$.

\begin{prop}\label{convergencerate}
Suppose that the assumptions of Proposition \ref{consistency} hold, that the errors
in \eqref{gencasesvec} are Gaussian, 
and that $r \ge 6$ is nondecreasing as $n\rightarrow\infty$.  Then
\begin{equation}\label{equaitonrate}
  \hat\gamma-\gamma_{\rm opt} =
    \begin{cases}
    O_P(r^{-1/2}) & \text{if $\delta^2=O(r)$}\\
O_P((\delta^2/r)^{-3/4})+O_P(n^{-1/2}(\delta^2/r)^{-1/2}) & \text{if $r \rightarrow \infty$ and $r=o(\delta^2)$}\\
     O_P((\delta^2)^{-3/4})+O_P(n^{-1/2}(\delta^2)^{-1/2}) & \text{if $r=O(1)$.}\\
    \end{cases}       
\end{equation}
\end{prop}
\noindent
The proof of Proposition \ref{convergencerate} is in Section \ref{appendthm3} of the Supplementary material. From the definition of $\gamma_{\rm opt}$, we know that $\gamma_{\rm opt}\rightarrow 0$ when $\delta^2 =o(r)$, in 
which case $\hat\gamma-\gamma_{\rm opt}=O_P(n^{-1/2})$
provided that $r\asymp n$.  On the other hand,
$\gamma_{\rm opt} \rightarrow 1$ when $r = o(\delta^2)$.  For example,
$\hat\gamma-\gamma_{\rm opt}=O_P(n^{-1/2})$
provided that $\delta^2 \asymp n^{2/3}$
and $r$ is bounded.
When $\delta^2=o(r)$, 
$\hat\gamma_{90},\hat\gamma_{95}$ 
and $\hat\gamma$ all have the same
convergence rate bound.

\subsection{Tuning parameter selection in high dimensions}\label{highdimsection}
Estimating the unknown parameters in 
$\gamma_{\rm opt}$ is challenging when $p > n$
and ${\rm rank}(X) = n$.  For example, it is impossible to estimate the regression's error variance $\sigma^2$ without assuming something extra about $\mu=X  \beta$. This is because the
data-generating model in \eqref{gencasesvec} reduces to 
$$
  Y =   \mu +  \varepsilon,
$$
where $  \mu$ has $n$ unknown free parameters and $\varepsilon$ has iid entries
with mean zero and variance $\sigma^2$. So we have a sample size
of 1 to estimate each $\mu_i$, which is not enough if we also want
to estimate $\sigma^2$.  
An anonymous referee mentioned that if there are replicated rows in $X$ (which implies ${\rm rank}(X)<n$), then one could estimate $\sigma^2$ using the measured responses corresponding to these the repeated rows.

We explore using cross-validation to choose a
value of $\gamma$ that minimizes the total
validation squared error in our numerical experiments.  This cross-validation procedure implicitly assumes that 
the response and predictor measurement pairs for 
each subject are drawn from a joint distribution.
As an alternative, we derive a high-dimensional estimator of $\gamma_{\rm opt}$ that estimates $\sigma^2$ with an assumption about $  \mu$.  The following paragraphs introduce this estimator, which is not invariant to invertible
linear transformations of $X$.

Recall that $\gamma_{\rm opt}={\delta^2}/({\sigma^2({\rm rank}(X)-1)+\delta^2})$, where $\delta^2=\| \mu-P_1  \mu\|^2$.  
Since $P_X$ is an identity operator when ${\rm rank}(X)=n$,
\begin{align*}
    \mathbb{E}\|Y-P_1Y\|^2=\mathbb{E}\|P_XY-P_1Y\|^2=\sigma^2({\rm rank}(X)-1)+\delta^2.
\end{align*}
So given an estimator $\check\sigma^2$ of $\sigma^2$, 
we study the following plug-in estimator of $\gamma_{\rm opt}$:
\begin{align}\label{highdimplugin}
    \Bar\gamma(\check\sigma^2)={\rm max}\Bigg(0,\frac{\|Y-P_1Y\|^2-\check\sigma^2({\rm rank}(X)-1)}{\|Y-P_1Y\|^2}\Bigg),
\end{align}
where the truncation at $0$ is necessary to ensure that $\Bar{\gamma}\in[0,1]$.  We continue by describing the
estimator of $\sigma^2$ that we will use in \eqref{highdimplugin}.

If one ignores invariance and assumes Gaussian errors, then one could simultaneously estimate $  \beta$ and $\sigma^2$
by penalized likelihood with the same penalty used in \eqref{mainoptlam}. However, this joint optimization is not convex. We avoid this nonconvexity by modifying a reparametrized penalized Gaussian likelihood optimization problem proposed by \citet{zhuconvex2020}.
Let $\eta=\sigma^{-1}$ and $ \beta^*= \beta\eta$.
We estimate these parameters with
\begin{align}\label{highdimopt}
    (\hat{ \beta}^*,\hat\eta)=\argmin_{( \beta^*, \eta) \in\mathbb{R}^p \times (0, \infty)} \left\{ \frac{1}{2n}\|Y\eta-X \beta^*\|^2-{\rm log}(\eta)+\alpha\| \beta^*_{-1}\|^2 \right\},
\end{align}
where $ \beta^*=(\beta^*_{1},\dots, \beta^*_{p})=(\beta^*_{1}, \beta^*_{-1})$; and 
$\alpha\ge 0$ be the tuning parameter for the Ridge penalty.  
This choice of $\alpha$ was motivated by \citet{Liu2020EstimationOE}, who verified that ridge regression (with tuning parameter $\alpha$) can
be used to consistently estimate $\sigma^2$
provided that $\alpha\| \beta\|^2=o(1)$.  However,
\citet{Liu2020EstimationOE} use 
a different estimator of $\sigma^2$ than 
the transformed solution to \eqref{highdimopt}.
We also examine other choices for $\alpha$ in the simulations (see section \ref{highsimulsection}).

The reparametrized optimization problem in \eqref{highdimopt} is strongly convex with the following global minimizer:
\begin{align}
\hat\eta&=\Big(n^{-1}Y'(I-X(X'X+2n\alpha M)^{-1}X')Y\Big)^{-1/2},\label{formeta}\\
\hat{ \beta}^*&=\hat\eta (X'X+2n\alpha M)^{-1}X'Y\nonumber,
\end{align}
where $M={\rm diag}(0,1,1,\dots,1)\in\mathbb{R}^{p\times p}$.
Since $\eta=\sigma^{-1}$, which is estimated using \eqref{formeta}, the corresponding estimator of $\sigma^2$ is
\begin{align}\label{checksigma}
\check\sigma^2=n^{-1}Y'(I-K)Y,   
\end{align}
where $K=X(X'X+2n\alpha M)^{-1}X'$.
Using $\check\sigma^2$ in \eqref{highdimplugin}, we propose a high-dimensional estimator of $\gamma_{\rm pot}$ which is given by
\begin{align}\label{highdimour}
\Bar{\gamma} = \frac{Y'(I-P_1-\frac{{\rm rank}(X)-1}{{\rm rank}(X)}(I-K))Y}{Y'(I-P_1)Y},    
\end{align}
where truncation at $0$ is not required since the quadratic term on the numerator of \eqref{highdimour} is non-negative definite.
\citet{Liu2020EstimationOE} proposed a bias corrected version of $\check\sigma^2$, since the uncorrected
estimator does not converge to $\sigma^2$ under the assumptions for Theorem 1 of \citet{Liu2020EstimationOE}. 
Their corrected estimator is $\check\sigma^2_c=C^{-1}\check\sigma^2$, where $C=1-{\rm tr}(K)/{\rm rank}(X)$. Using $\check\sigma^2_c$ in \eqref{highdimplugin}, we also propose alternative 
high-dimensional estimator of $\gamma_{\rm opt}$
by 
\begin{align}\label{highdimcorrected}
    \Bar{\gamma}_c = {\rm max}\Bigg(0, \frac{Y'(I-P_1-\frac{{\rm rank}(X)-1}{{\rm rank}(X)-{\rm tr}(K)}(I-K))Y}{Y'(I-P_1)Y}\Bigg).
\end{align}

First, we have the following
consistency result for the corrected estimator $\Bar{\gamma}_c$.
\begin{prop}\label{highdimconsistency}
Assume that the data-generating
model in \eqref{gencasesvec} is correct, that error
distribution has a finite fourth moment, 
$\delta^2=o(n)$, and $d_2 (n\alpha)^{-1}=o_P(1)$, where $d_2$ is
the second-largest eigenvalue of $XX'$. 
Then $\Bar{\gamma}_c - \gamma_{\rm opt}\rightarrow_P 0$ as $n\rightarrow \infty$.
\end{prop}
\noindent
The proof is in Section \ref{appendthm4} of the Supplementary material. %
We see that $\Bar{\gamma}_c$ converges to $\gamma_{\rm opt}$
when $\gamma_{\rm opt}\rightarrow 0$.
The assumption that $d_2 (n\alpha)^{-1}=o_P(1)$ is met 
when $\alpha=n^{3/2} \left(2\|Y-P_1Y\|^2\right)^{-1}$, $\delta^2=o(n)$,
and $d_2=O(n)$ because 
\begin{align*}
    d_2 (n\alpha)^{-1}&=2n^{-5/2}d_2\|Y-P_1Y\|^2\\
    &=2n^{-5/2}d_2(\delta^2+2( \mu-P_1 \mu)' \varepsilon+ \varepsilon'(I-P_1) \varepsilon)\rightarrow_P 0,
\end{align*}
since $\|\varepsilon\|^2=O_P(n)$.

We also established the following result for
the uncorrected estimator $\Bar\gamma$.

\begin{prop}\label{highdimconsistency:2}
Assume that the data-generating
model in \eqref{gencasesvec} is correct, that the error distribution has a finite fourth moment, and at least one of the followings holds:
\begin{itemize}
    \item $\delta^2=o(n)$ and $d_2(n\alpha)^{-1} =o_P(1)$,
    \item $n=o(\delta^2) $ and $n\alpha(d_n)^{-1}=o_P(1)$,
\end{itemize} 
where $d_2$ (resp. $d_n$) is
the secondly largest (smallest) eigenvalue of $XX'$. 
Then $\Bar{\gamma} - \gamma_{\rm opt}\rightarrow_P 0$ as $n\rightarrow \infty$.
\end{prop}
The proof is in Section \ref{appendthm:highdim2} of the Supplementary material. 
Using the uncorrected error variance estimator  $\check\sigma^2$
enables us to prove consistency of $\Bar\gamma$ when either $\gamma_{\rm opt}\rightarrow 0$ or  $\gamma_{\rm opt}\rightarrow 1$. When, $\gamma_{\rm opt}\rightarrow 0$, we have consistency provided that $d_2=O(n)$ and $\alpha=n^{t} \left(2\|Y-P_1Y\|^2\right)^{-1}$ with any $t\in(1,\infty)$.
We test $t=2,3$ in Section \ref{highsimulsection}. On the other hand, when $\gamma_{\rm opt}\rightarrow 1$, we have consistency provided that 
$d_n\asymp n$ and $\alpha=n \left(2\|Y-P_1Y\|^2\right)^{-1}$.
\begin{remark}
\citet{Azriel2019} proved that ``general" consistent estimation of the error variance $\sigma^2$ in high-dimensional linear regression is impossible without further assumptions in fixed-X settings. Their proof is based on a Bayes risk approach that exploits a prior that requires 
$\delta^2 \sim n$, which is not what
we assume in Proposition \ref{highdimconsistency:2} ($\delta^2=o(n)$ or $n=o(\delta^2)$). 
\end{remark}

\section{Simulation studies}\label{sec:simuls}
\subsection{Low-dimensional experiments}\label{lowsimulsection}
We conducted a lower-dimensional simulation study in which 
the data were generated from the linear regression 
subjects model \eqref{gencasesvec}
with $n=300$ and $p\in\{75, 150, 250\}$.
Also,  $\epsilon_1,\ldots, \epsilon_n$ are iid $N(0,1)$.
The design matrix $X$ has ones in its first column and independent
draws from $N_{p-1}(0,  \Sigma)$ in the remaining entries on each row,
where $\Sigma_{jk}=0.5^{|j-k|}$.
We randomly generated the regression coefficient vector with the following equation: 
$$
 \beta={  X}^{-}(1_p+\tau   Z),
$$ 
where $\tau\in\{0,10^{-4},10^{-2},10^{-1},10^{-0.5},1,10^{0.5},10^1,10^{1.5},10^2\}$; 
and $  Z$ is $N_p(0, I)$. Then
\begin{align*}
     \mu-P_1 \mu=(P_X-P_1)(1_p+\tau    Z)=\tau (P_X-P_1)   Z.
\end{align*}
So $\tau$ controls the size of $\delta^2=\|  \mu-P_1 \mu\|^2$. 

We used 50 independent replications in each setting.
In each replication, we measured the performance of each estimator $\hat\beta_{\rm est}$ 
using the same-X loss: 
\begin{align}\label{losscrit}
n^{-1}\|X \beta-X\hat{ \beta}_{\rm est}\|^2.    
\end{align}

The candidate estimators $\hat{\beta}_{\rm est}$ that we considered
were the following:
\begin{itemize}
    \item \textbf{2n-G}: $L_2$-squared penalty with 10-fold cross validation for $\gamma$ \eqref{mainopt}.
    \item \textbf{2n-Or}: $L_2$-squared penalty using the oracle $\gamma_{\rm opt}$ in \eqref{gamopt}.
    \item \textbf{2n-Es}: $L_2$-squared penalty using $\hat\gamma$ in \eqref{gamhat}.
    \item \textbf{2n-Es90}: $L_2$-squared penalty using $\hat\gamma_{90}$ in \eqref{fq90}.
    \item \textbf{2n-Es95}: $L_2$-squared penalty using $\hat\gamma_{95}$ in \eqref{fq95}.
    \item \textbf{2n-Rep}: $L_2$-squared penalty using $\check\sigma^2$ in \eqref{checksigma}, $\alpha=n(2\|Y-P_1Y\|^2)^{-1}$ in \eqref{highdimopt}, and the corresponding $\Bar\gamma_{\rm low}={\rm max}(0,1-1/F_{\rm rep})$, where $F_{\rm rep}=\|P_XY-P_1Y\|^2/({\check\sigma^2({\rm rank}(X)-1)})$.
    \item \textbf{O}: Ordinary least square (OLS) estimator given by
    \begin{align}\label{OLS}
     \hat{\beta}^{(1)}={X}^{-}{Y}.   
    \end{align}
    \item \textbf{R}: Ridge-penalized least squares \citep{hoerl70} 
    \begin{align}
        \hat{\beta}_{\rm Ridge}={\rm argmin}_{b \in \mathbb{R}^p} \|Y-Xb\|^2+\lambda\|b_{-1}\|^2,\label{ridgeform}
    \end{align}
    where $b=(b_1,b_{-1})$ with $b_{-1}=\mathbb{R}^{p-1}$; 10-fold cross validation for the selection of $\lambda$.
    \item \textbf{L}: Lasso-penalized least squares \citep{tibs96}
    \begin{align}\label{LASSO}
        \hat\beta_{\rm LASSO}={\rm argmin}_{b \in \mathbb{R}^p} \|Y-Xb\|^2+\lambda\|b_{-1}\|_1,
    \end{align}
    where 10-fold cross validation is used for the selection of $\lambda$.
\end{itemize}
For the methods that require cross validation,  
$\lambda$ and $\gamma$ were 
selected from $\{10^{-7+0.25j}:j=0,1,\cdots,44\}$ and $\{\frac{k}{99}:k=0,1,\cdots,99\}$, respectively. To facilitate the fairest
comparison between our invariant methods and the ridge/lasso methods, 
we used the following standardization process, which is the default 
process used by the R package {\tt glmnet}:
the ridge/lasso shrunken coefficient estimates are computed using
the standardized design matrix $X_{\bullet}$ defined by
\begin{align*}
    X_{\bullet} &=X\begin{bmatrix} 
    1 & -\Bar{X}_2S_2^{-1} & -\Bar{X}_3S_3^{-1} & \dots & -\Bar{X}_{p-1}S_{p-1}^{-1} & -\Bar{X}_pS_p^{-1} \\
    0 & S_2^{-1} & 0 & \dots & 0 & 0 \\
     & 0 & S_3^{-1} &  & \vdots & \vdots \\
    \vdots & \vdots & 0 & \ddots & 0 & \vdots \\
    0 & 0 & \vdots &  & S_{p-1}^{-1} & 0 \\
    0 & 0 & 0 & \dots & 0 & S_p^{-1} \\
    \end{bmatrix}\\
    &=XT,
\end{align*}
where $\Bar{X}_j = n^{-1} \sum_{i=1}^n X_{ij}$ and $S_j^2 = (n-1)^{-1} \sum_{i=1}^n (X_{ij}-\Bar{X}_j)^2$ for $j \in \{2,\dots,p\}$. 
Let $\hat\beta_{\bullet}$ be the shrinkage estimator of the standardized
coefficients.  Since all the $S_{j}$'s will be positive, we invert $T$ to estimate the original $\beta$ with $T^{-1}\hat\beta_{\bullet}$.
Our proposed fitted-value shrinkage procedures are invariant to this  standardizing transformation of $X$.

We display side-by-side boxplots of the same-X losses from the 50 replications when $p=75$ in Figure \ref{fig:lowdim1}. 
Further numerical summaries
of these results are in Table \ref{table:lowdim.samex.1} in Section  \ref{appd:added.low.tables} of the Supplementary material.
Without surprise, our fitted-value shrinkage with oracle tuning \textbf{2n-Or} performed the best among these candidates. 
Our proposed estimator \textbf{2n-Es} and its two variants \textbf{2n-Es95} and \textbf{2n-Es90} followed and generally outperformed OLS, Ridge, and Lasso when $\tau\ge 10^{-0.5}$. 
Of the fitted-value shrinkage estimators, \textbf{2n-Es} outperformed \textbf{2n-Rep}. 
On the other hand, 
the modified thresholds \textbf{2n-Es90} and \textbf{2n-Es95} performed better than \textbf{2n-Es} 
for smaller values of $\tau$ (that correspond to smaller values of $\delta^2$).
Also, \textbf{2n-G} outperformed Ridge, Lasso, \textbf{2n-ES}, and \textbf{2n-Es95} for smaller values of $\tau$. 
Ridge and Lasso slightly outperformed \textbf{2n-Es} when $\delta^2$ was small, but performed worse when $\delta^2$ was larger.

In Figure \ref{fig:lowdimgr1} in Section \ref{addedboxplots} of the Supplementary material,
we graphed the average same-X loss values over the 50 replications as a function of $\lambda$:
\begin{align}
&f_{\rm FVS}(\lambda)=\|X\beta-X\hat\beta_{\lambda}\|^2/n,\label{lam.gr.FVS}\\
&f_{\rm Ridge}(\lambda)=\|X\beta-X\hat\beta_{{\rm Ridge};\lambda}\|^2/n,\label{lam.gr.Ridge}
\end{align}
where $\hat\beta_{\lambda}$ and $\hat\beta_{{\rm Ridge};\lambda}$ are solutions for \eqref{mainoptlam} and \eqref{ridgeform};
and 
$\lambda\in\{10^{-7+0.01k}:j=0,1,\ldots,1150\}$.
This allows a further comparison of fitted-value shrinkage \eqref{mainopt} to Ridge regression \eqref{ridgeform}. Further numerical summaries of these results are in Table \ref{table:graph table 1} in Section \ref{appd:added.low.tables} of the Supplementary material.
The minimum average same-X loss for fitted value shrinkage \eqref{mainoptlam} was either less than or nearly equal to that of Ridge \eqref{ridgeform}. However, the range of values of $\lambda$ that corresponded to average losses near the minima was much narrower for fitted-value shrinkage than it was for Ridge when medium to large values of $\tau$ were used.

We also display side-by-side boxplots of the observed same-X losses from the 50 replications
when $p=150$ in Figure \ref{fig:lowdim2}
and when $p=250$ in Figure \ref{fig:lowdimadd} in Section \ref{addedboxplots} of the Supplementary material.
In addition, graphs of the 
average loss values as a function of $\lambda$ over $\lambda\in\{10^{-7+0.01k}:j=0,1,\ldots,1100\}$ when $p=150$ and $p=250$ are in Figure \ref{fig:lowdimgr2}--\ref{fig:lowdimgr3} in Section \ref{addedboxplots} of the Supplementary material.
Further numerical summaries
for this simulation with $p\in\{150,250\}$ are in Table \ref{table:lowdim.samex.2}--\ref{table:lowdim.samex.3}, and Table \ref{table:graph table 2}--\ref{table:graph table 3} in Section \ref{appd:added.low.tables} of the Supplementary material.
These results with $p\in\{150,250\}$ are similar to results when $p=75$:
the oracle method \textbf{2n-Or} was the best and
our proposed estimators \textbf{2n-Es90}, \textbf{2n-Es95}, \textbf{2n-Es} were the most competitive when $\tau\ge 1$. However, the performance gap between 
our procedure with non-oracle tuning \textbf{2n-G} and our procedure using
oracle tuning \textbf{2n-Es} when $\tau\ge 1$ has increased (Figure \ref{fig:1506}, \ref{fig:1507}).
We expect this is related to the narrower valley observed in the 
graph of the average same-X loss as a function of the tuning parameter.
As it was when $p=75$, 
\textbf{2n-G} was the most competitive for smaller values of $\tau$.

In Tables \ref{table:paired test table1}--\ref{table:paired test table2} in Section \ref{appd:added.low.tables} of the Supplementary material, we report 95\% simulation-based confidence intervals 
for the expected squared same-X loss difference 
(fitted value shrinkage minus Ridge).
The confidence intervals are based on 50  replications for each pair of $(\tau,p)$.
These differences were statistically insignificant  when $\tau$ was small, but were significant when $\tau$ was larger (with fitted value shrinkage outperforming Ridge).  In Table \ref{table:gamopt table} in Section \ref{appd:added.low.tables} of the Supplementary material, we also report
the average realization of $|\hat\gamma-\gamma_{\rm opt}|^2$ based on 100 independent replications from the same simulation setting. This quantity increased with $p$, but remained stable.

\begin{figure}
     \centering
     \begin{subfigure}[b]{0.45\textwidth}
         \centering
         \includegraphics[width=\textwidth]{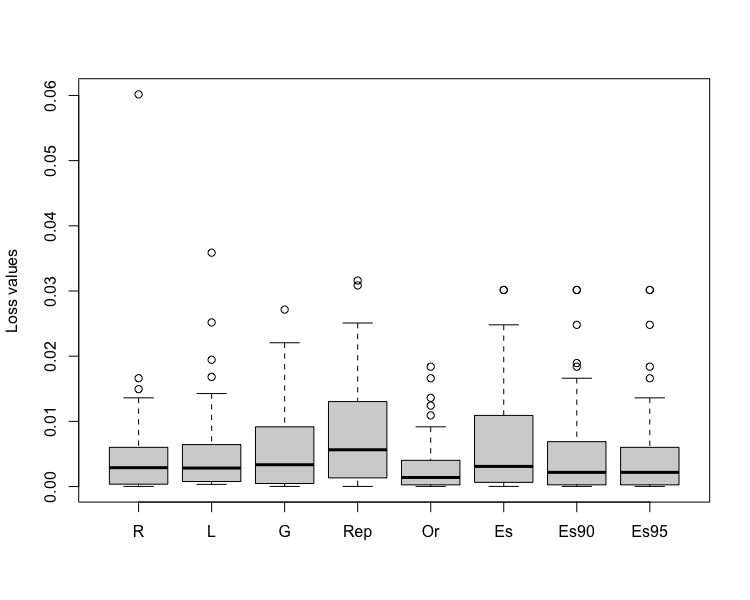}
         \caption{$\tau=10^{-2}$}
         \label{fig:753}
     \end{subfigure}
     \begin{subfigure}[b]{0.45\textwidth}
         \centering
         \includegraphics[width=\textwidth]{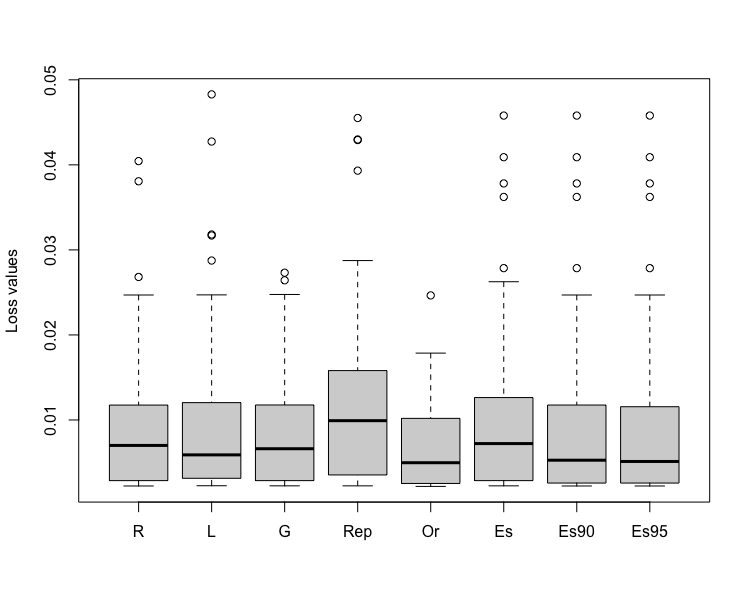}
         \caption{$\tau=10^{-1}$}
         \label{fig:754}
     \end{subfigure}\\
     \begin{subfigure}[b]{0.45\textwidth}
         \centering
         \includegraphics[width=\textwidth]{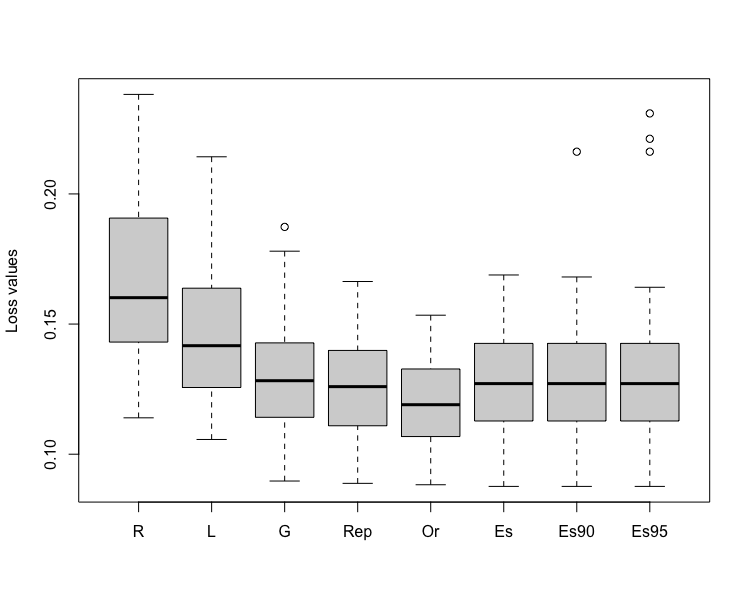}
         \caption{$\tau=1$}
         \label{fig:756}
     \end{subfigure}
     \begin{subfigure}[b]{0.45\textwidth}
         \centering
         \includegraphics[width=\textwidth]{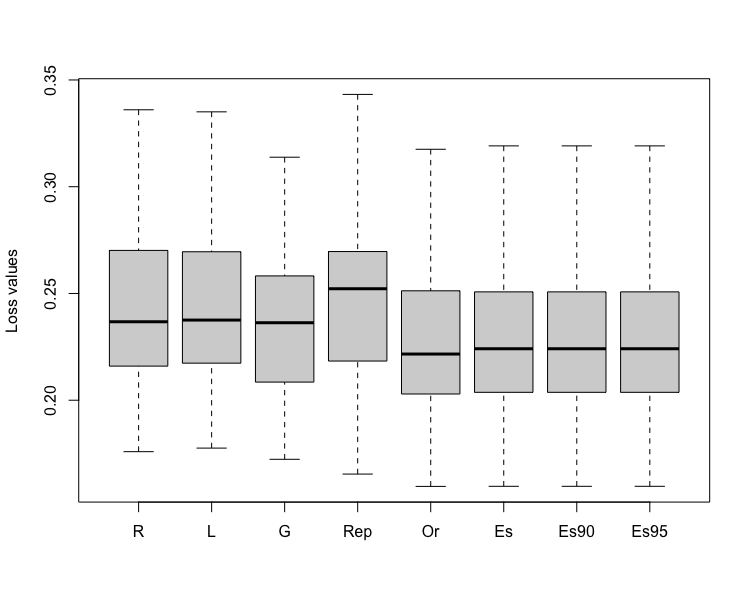}
         \caption{$\tau=10^{0.5}$}
         \label{fig:757}
     \end{subfigure}\\
 \caption{Boxplots of the observed same-X loss values from the 50 replications when $n=300$ and $p=75$. 
 We suppress ``2n" from the third to the last estimators to simplify notation.
 }
\label{fig:lowdim1}
\end{figure}

\begin{figure}
     \centering
     \begin{subfigure}[b]{0.45\textwidth}
         \centering
         \includegraphics[width=\textwidth]{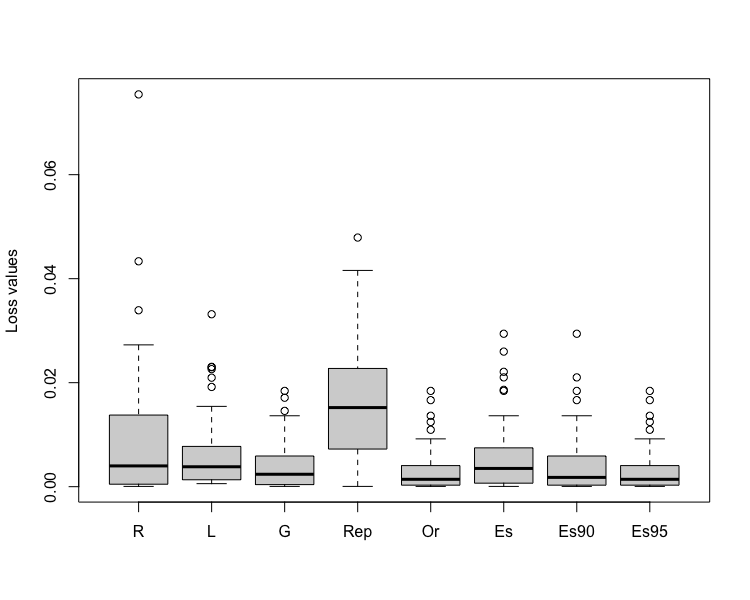}
         \caption{$\tau=10^{-2}$}
         \label{fig:1503}
     \end{subfigure}
     \begin{subfigure}[b]{0.45\textwidth}
         \centering
         \includegraphics[width=\textwidth]{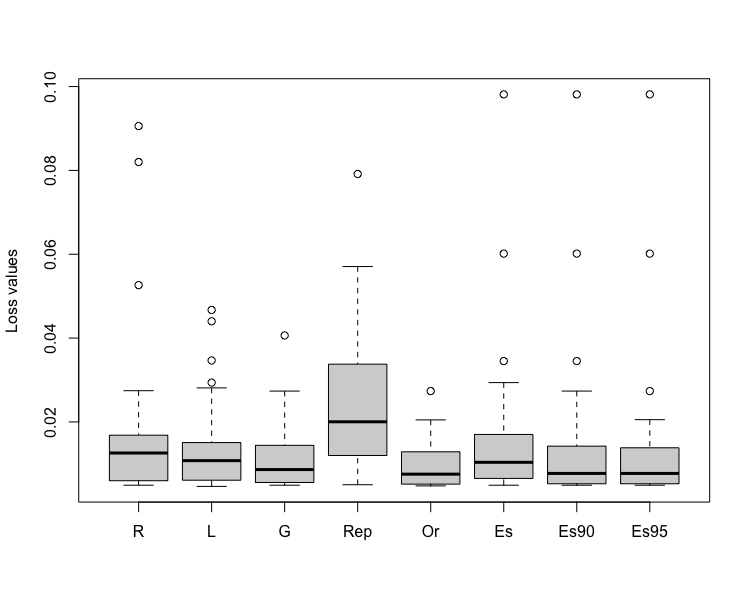}
         \caption{$\tau=10^{-1}$}
         \label{fig:1504}
     \end{subfigure}\\
     \begin{subfigure}[b]{0.45\textwidth}
         \centering
         \includegraphics[width=\textwidth]{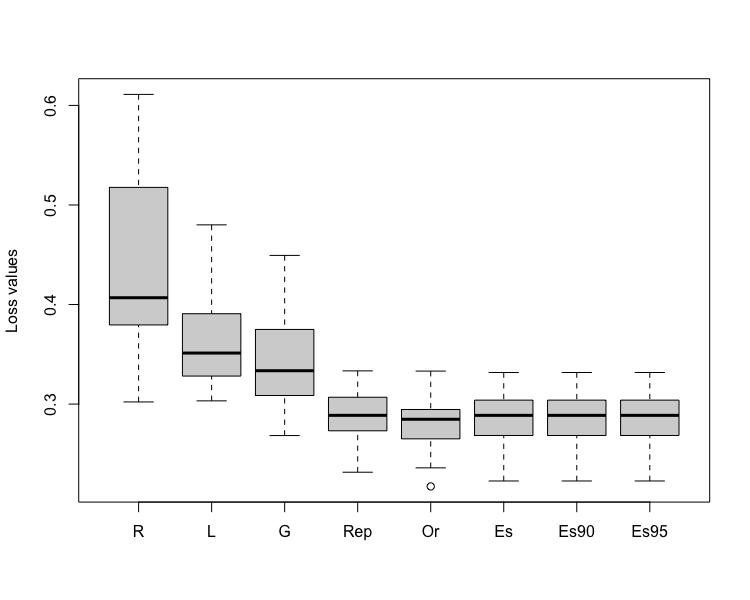}
         \caption{$\tau=1$}
         \label{fig:1506}
     \end{subfigure}
     \begin{subfigure}[b]{0.45\textwidth}
         \centering
         \includegraphics[width=\textwidth]{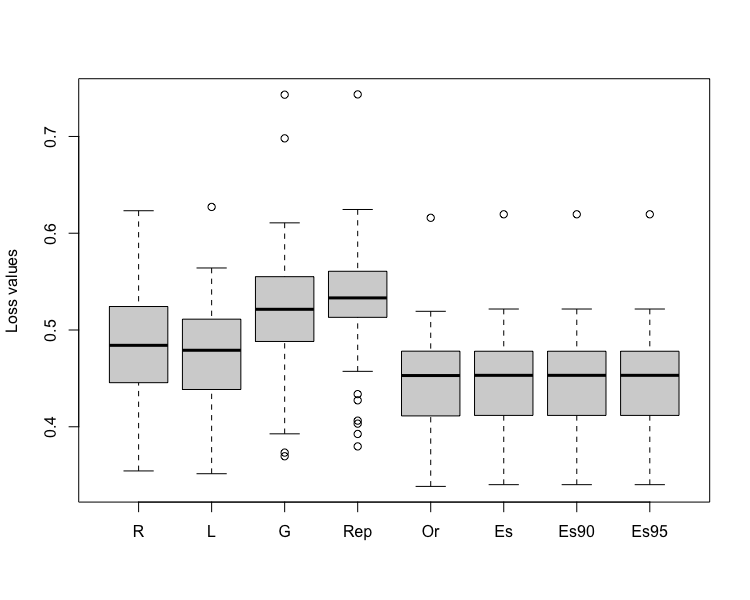}
         \caption{$\tau=10^{0.5}$}
         \label{fig:1507}
     \end{subfigure}\\
 \caption{
 Boxplots of the observed same-X loss values from the 50 replications when $n=300$ and $p=150$. See the caption of Figure \ref{fig:lowdim1} for more details.}
\label{fig:lowdim2}
\end{figure}

\subsection{Impact of invariance I: In the presence of categorical variables with different reference level coding schemes}\label{factorsimul}
We generated $X\in\RR^{n\times p}$ to have $25$ numerical predictors and $3$ categorical predictors, where each categorical predictor had 5 levels, and the 25th numerical predictor had interactions with the three categorical predictors. 
The design matrix $X$ had $p=1+25+3\times 4\times 2=50$ columns
including the intercept.
The $n=100$ observations of the $25$ numerical predictors were independent draws from $N_{25}(0,\Sigma)$, where $\Sigma_{jk}=0.5^{|j-k|}$.
Observations of each categorical predictor
were independently drawn from a 5-category multinomial distribution with equal 
category probabilities.  As they were before, $\epsilon_1,\ldots,\epsilon_n$ are iid $N(0,1)$.

We organize the regression coefficient vector as $\beta=(\beta_c,\beta_f)\in\RR^{50}$, 
 where $\beta_c$ corresponds to the first $26$ columns of $X$, which have the numerical predictors; and $\beta_f$ corresponds
 to the remaining 24 columns of $X$, which have
 the categorical predictors and their 
 interaction with the 25th numerical predictor.
 Let $X_c\in\RR^{100\times 26}$ be the matrix
 with the first 26 columns of $X$.
 We generated
 $\beta_c=X_c^{-}(1_{26}+\tau_c Z_c)$; where $Z_c \sim N_{26}(0,I)$.
We also generated 
\begin{align*}
&\beta_f=(\tau_f X_f^{-} Z_f)*v,\\
&v=(w',w')'\in\mathbb{R}^{24},\\
&w=(1,2,0,0,0,2,1,0,0,0,2,1)'\in\mathbb{R}^{12},
\end{align*}
where $*$ denotes the elementwise product of two vectors of same dimension; $Z_f\sim N_{24}(0,I)$; $X_f\sim N_{24}(0,\Sigma)$ with $\Sigma_{jk}=0.5^{|j-k|}$;
$Z_c$ is independent of $Z_f$; and $X_c$ is independent of $X_f$.
We consider $(\tau_c,\tau_f)$ in 
\begin{align*}
\{(10^{-0.5},10^{-0.5}),(1,10^{-0.5}),(10^{-0.5},1),(1,1),(10^{0.5},1),(1,10^{0.5}),(10^{0.5},10^{0.5})\}.
\end{align*}
The first 12 entries of $\beta_f$ correspond
to the $3*4$ main effects columns for the 3 categorical predictors.   The remaining 12
entries of $\beta_f$ correspond to interactions
between the 3 categorical predictors and the 25th numerical predictor.

The first level of each categorical predictor
was coded as the reference level in this design matrix $X$.  We call this ``Coding-1".  To illustrate the effects of a coding change, we will also use a ``Coding-2" design matrix $X_{\bullet}$, 
which is an invertible linear transformation of $X$ that uses the second, third, and fifth levels of the three categorical predictors as the reference levels for the first, second, and third categorical predictor, respectively.

The competitors were \textbf{2n-Es}, \textbf{2n-Es95}, \textbf{2n-Or}, \textbf{2n-G}, \textbf{2n-Rep}, \textbf{O}, \textbf{R}, \textbf{L}, and Group Lasso and its two variants.
The Group lasso competitors were the following:
\begin{itemize}
    \item \textbf{GL}: Group Lasso estimator proposed by \citet{Grouplasso.YuanandLin}.
    \begin{align*}
        \hat{\beta}_{\rm GL}=\argmin_{b\in\mathbb{R}^p} \|Y-Xb\|^2+\lambda\sum_{i=1}^L\sqrt{p_i}\|b^{(i)}\|_2,
    \end{align*}
    where $i\in\{1,\ldots,L\}$ is the index of each group except for one intercept column; $p_i$ is the size of the $i$-th group; and $b^{(i)}$ is the corresponding subvector.
    \item \textbf{MGL}: Modified Group Lasso proposed by \citet{choi2012lasso}.
    \begin{align*}
        \hat{\beta}_{\rm MGL}=\argmin_{b\in\mathbb{R}^p} \|Y-Xb\|^2+\lambda\sum_{i=1}^L\|b^{(i)}\|_2.
    \end{align*}
    \item \textbf{SGL}: Standardized Group Lasso proposed by \citet{simon2012standardization}.
    \begin{align*}
        \hat{\beta}_{\rm SGL}=\argmin_{b\in\mathbb{R}^p} \|Y-Xb\|^2+\lambda\sum_{i=1}^L\sqrt{p_i}\|X^{(i)}b^{(i)}\|_2,
    \end{align*}
    where $X^{(i)}$ is the $i$-th submatrix.
\end{itemize}
For the estimators that require tuning parameter selection, 10-fold cross validation was used.
For \textbf{GL}, \textbf{MGL}, and \textbf{SGL}, we formed 31 groups: 25 single-member groups corresponding to the numerical predictors, 3 groups
corresponding to the categorical predictors' main effects, and 3 groups corresponding to the interaction between the 25th numerical predictor and the three categorical predictors.

We display side-by-side boxplots of the same-X
loss values from 50 replications in Figure \ref{fig:coding.change}.
Further numerical summaries of these results are in Table \ref{table:low dim factor table} in Section \ref{appd:factorsimul} of the Supplementary material. 
Except for $(\tau_c,\tau_f)=(10^{0.5},10^{0.5})$ and $ (10^{0.5},10^{-0.5})$, \textbf{2n-Or} was the best in Coding-2.
Generally, \textbf{2n-Es} was second-best in Coding-2.  The methods \textbf{R}, \textbf{L}, \textbf{GL}, \textbf{MGL} and \textbf{SGL} generally performed well in Coding-1, which is not surprising considering the sparsity in the coefficient vector; however, in Coding-2, these methods were generally worse than \textbf{2n-Es}, which is invariant to this change in coding.
This illustrates that the invariance of our proposed
estimators is advantageous.

\begin{figure}
     \centering
     \begin{subfigure}[b]{0.45\textwidth}
         \centering
         \includegraphics[width=\textwidth]{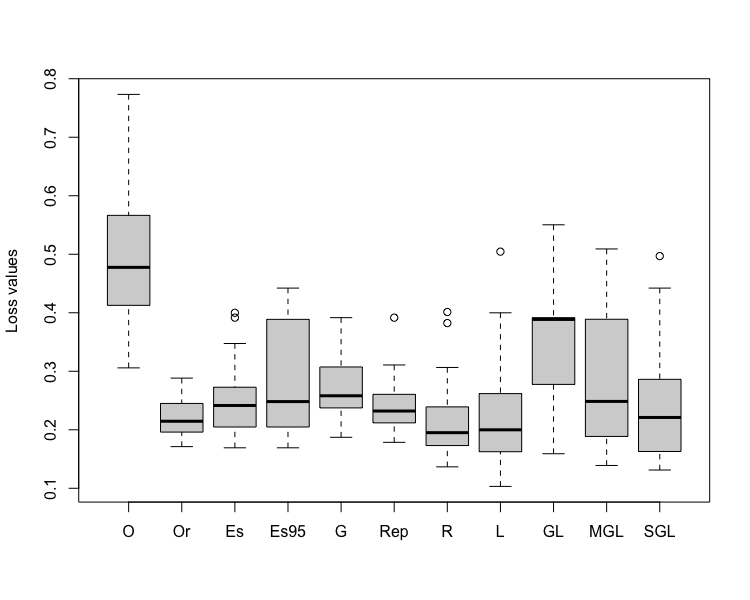}
         \caption{$(10^{-0.5},1,1)$}
         \label{fig:tdl131}
     \end{subfigure}
     \begin{subfigure}[b]{0.45\textwidth}
         \centering
         \includegraphics[width=\textwidth]{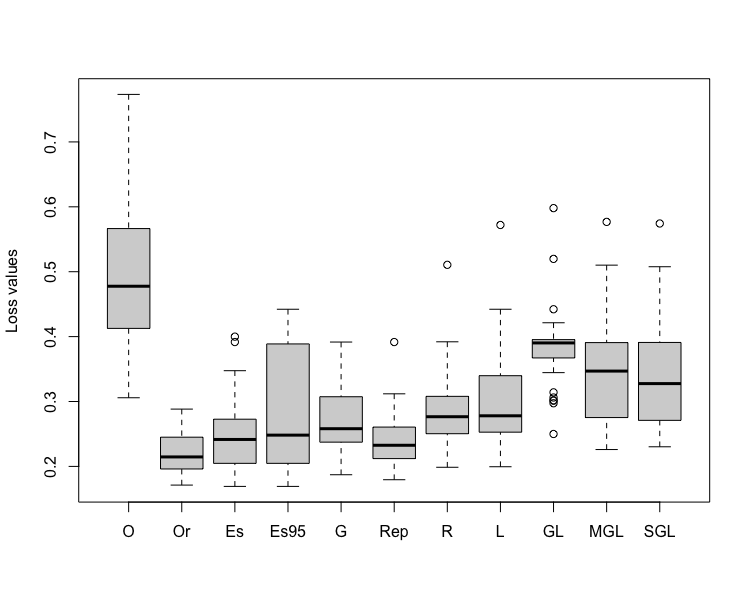}
         \caption{$(10^{-0.5},1,2)$}
         \label{fig:tdl132}
     \end{subfigure}\\
     \begin{subfigure}[b]{0.45\textwidth}
         \centering
         \includegraphics[width=\textwidth]{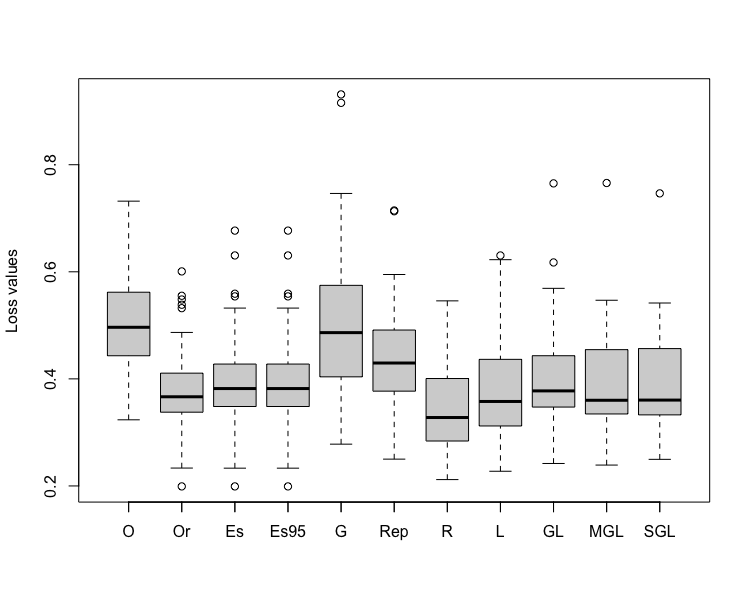}
         \caption{$(1,1,1)$}
         \label{fig:tdl141}
     \end{subfigure}
     \begin{subfigure}[b]{0.45\textwidth}
         \centering
         \includegraphics[width=\textwidth]{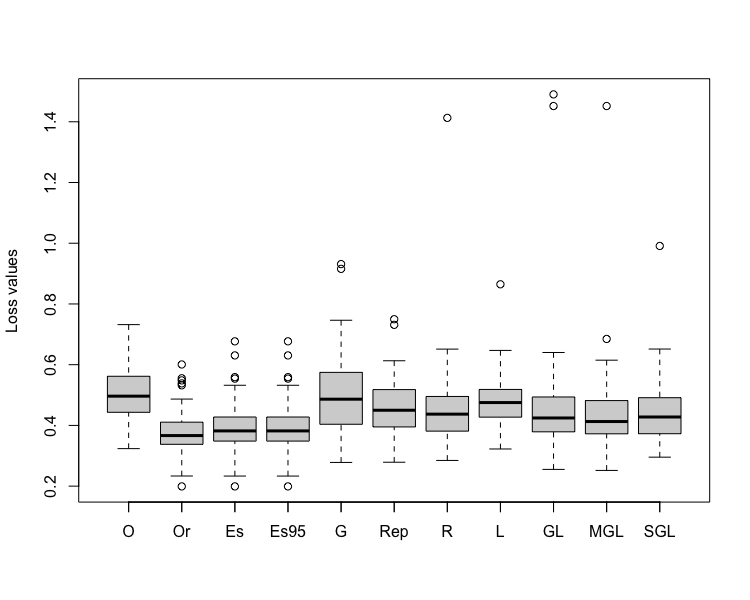}
         \caption{$(1,1,2)$}
         \label{fig:tdl142}
     \end{subfigure}\\
     \begin{subfigure}[b]{0.45\textwidth}
         \centering
         \includegraphics[width=\textwidth]{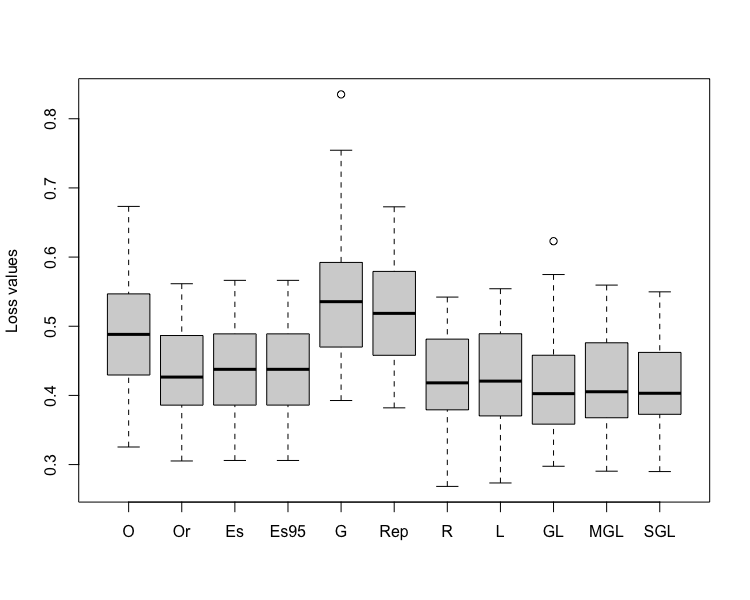}
         \caption{$(10^{0.5},1,1)$}
         \label{fig:tdl151}
     \end{subfigure}
     \begin{subfigure}[b]{0.45\textwidth}
         \centering
         \includegraphics[width=\textwidth]{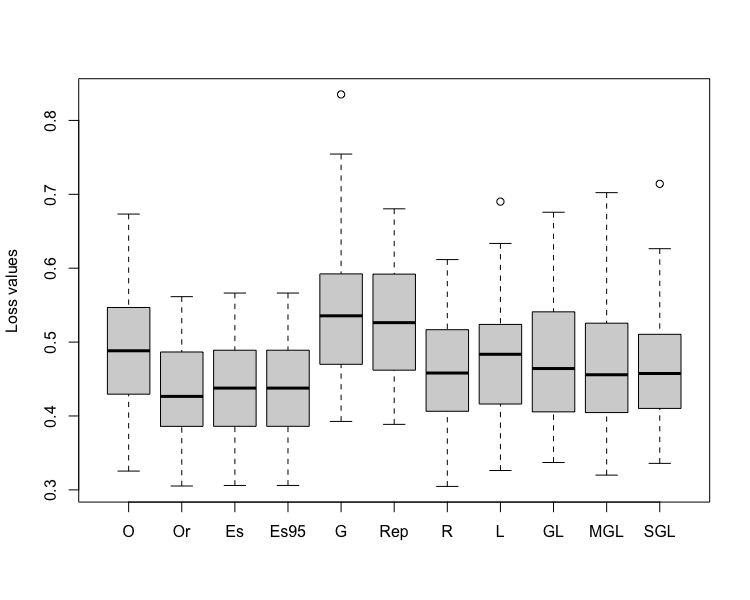}
         \caption{$(10^{0.5},1,2)$}
         \label{fig:tdl152}
     \end{subfigure}\\
 \caption{Boxplots of the observed same-X loss values from the 50 replications in the simulation settings of Section \ref{factorsimul}. 
 We suppress ``2n" from the third to the last estimators for notational simplicity.
 Each plot is labeled with the $(\tau_c,\tau_f,{\rm Coding~number})$ value used.
 }
\label{fig:coding.change}
\end{figure}

\subsection{Impact of invariance II: General full-rank transformation of the model matrix (Low dimension)}\label{fullrankmtx}
In this example, the design matrix $X\in\RR^{n\times p}$ has ones in its first column and its $n$ row vectors (excluding the first entry) are drawn independently from $N_{p-1}(0,\Sigma)$ where $\Sigma_{jk}=0.5^{|j-k|}$ and $(n,p)=(300,150)$. We generated the regression coefficient vector $\beta$ as $\beta=u*v$, where each element of $u\in\RR^p$ is an independent draw from the Uniform distribution on $(2^{-\psi-1},2^{-\psi})$ with $\psi\in\{0,1\}$; and $v=(1,v_{-1})\in\RR^{p}$, where each element of $v_{-1}\in\RR^{p-1}$ is an independent draw from ${\rm Ber}(s)$ with $s\in\{0.025,0.05,0.1,0.2,0.3\}$. We generated $Y$ by \eqref{gencasesvec}.

We denote the model fitting scheme with this $X$ as ``Coding-1".  We also use a transformed design matrix $X_\bullet=XT$, where $T$ is a Gram-Schmidt orthogonalization of a matrix with all of its entries drawn independently from $N(0,1)$. 
We refer the transformed design matrix $X_\bullet$
as ``Coding-2".

In Figure \ref{fig:XT.low}, 
we display side-by-side boxplots of the observed same-X losses from the 50 replications from representative settings when $\psi=0$.
Further numerical summaries of the results are in Table \ref{table:low dim XT table} in Section \ref{appd:low.XT} of the Supplementary material. When $\tau=1$, Lasso was the best in Coding-1 with Ridge following next.
However, in Coding-2, \textbf{2n-Or} was the best and \textbf{2n-Es} was the second-best regardless of $\psi$ and $s$. 
Ridge and Lasso were significantly worse than our proposed estimator as well as OLS under Coding 2.  This further illustrates the benefits of invariance.

\begin{figure}
     \centering
     \begin{subfigure}[b]{0.45\textwidth}
         \centering
         \includegraphics[width=\textwidth]{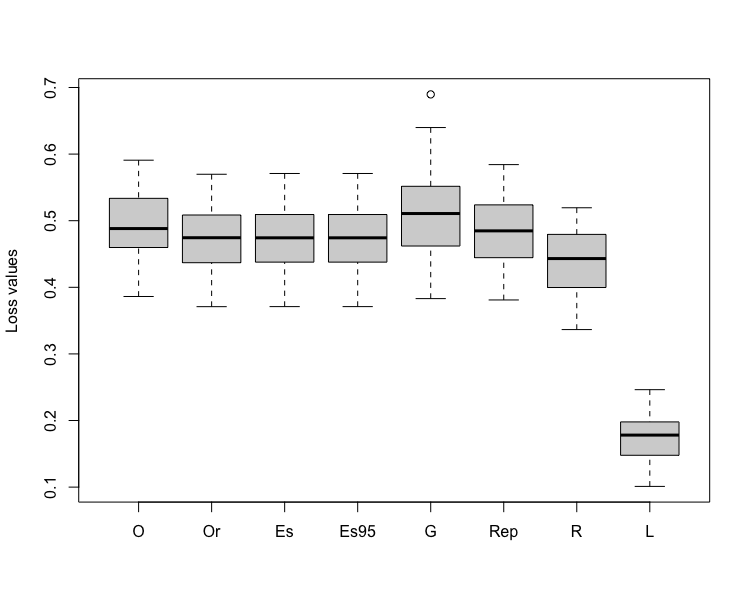}
         \caption{$(0.1,1)$}
         \label{fig:tdl231}
     \end{subfigure}
     \begin{subfigure}[b]{0.45\textwidth}
         \centering
         \includegraphics[width=\textwidth]{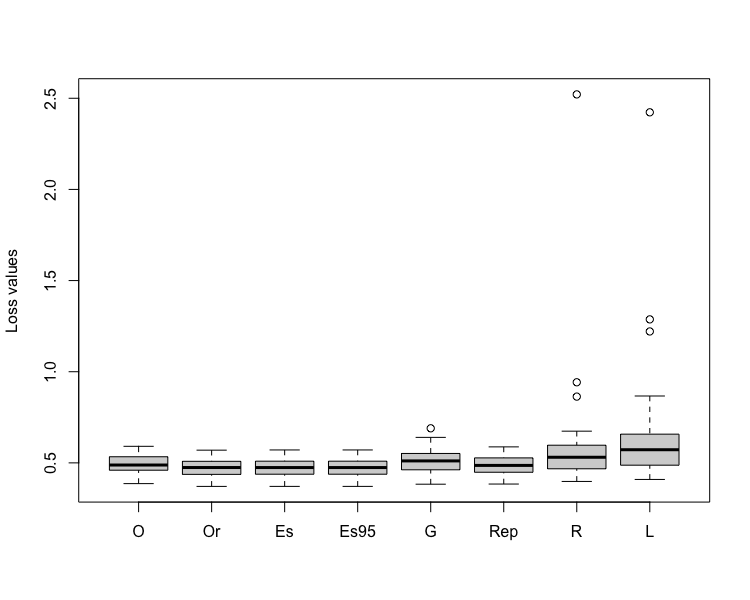}
         \caption{$(0.1,2)$}
         \label{fig:tdl232}
     \end{subfigure}\\
     \begin{subfigure}[b]{0.45\textwidth}
         \centering
         \includegraphics[width=\textwidth]{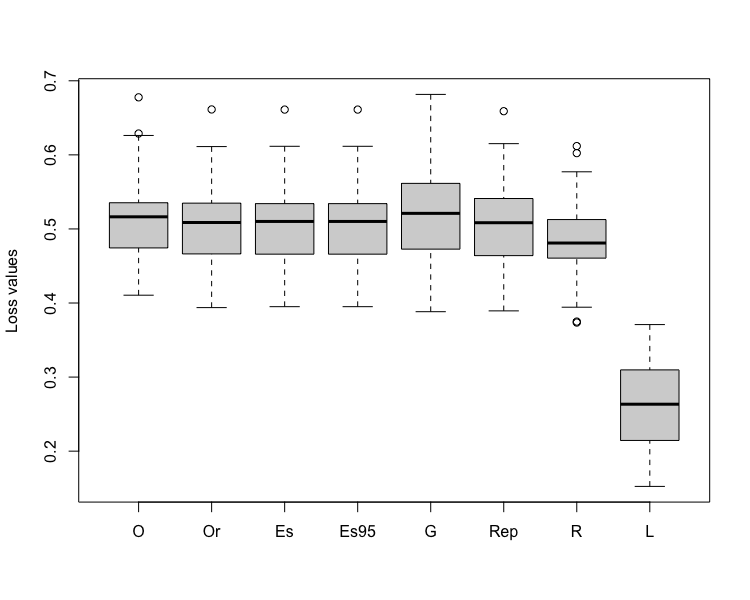}
         \caption{$(0.2,1)$}
         \label{fig:tdl241}
     \end{subfigure}
     \begin{subfigure}[b]{0.45\textwidth}
         \centering
         \includegraphics[width=\textwidth]{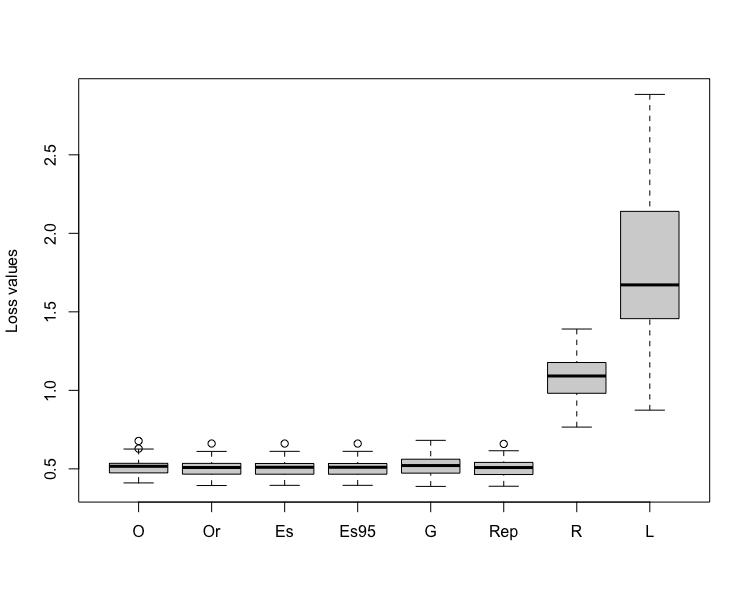}
         \caption{$(0.2,2)$}
         \label{fig:tdl242}
     \end{subfigure}\\
     \begin{subfigure}[b]{0.45\textwidth}
         \centering
         \includegraphics[width=\textwidth]{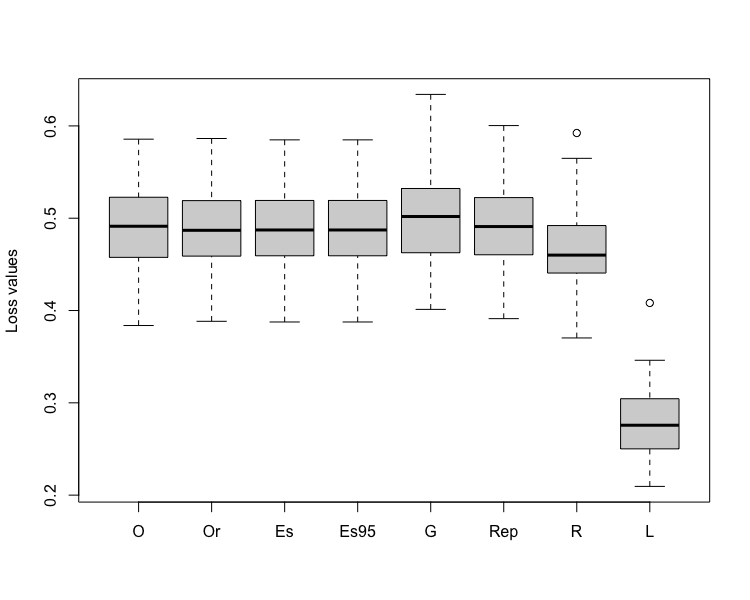}
         \caption{$(0.3,1)$}
         \label{fig:tdl251}
     \end{subfigure}
     \begin{subfigure}[b]{0.45\textwidth}
         \centering
         \includegraphics[width=\textwidth]{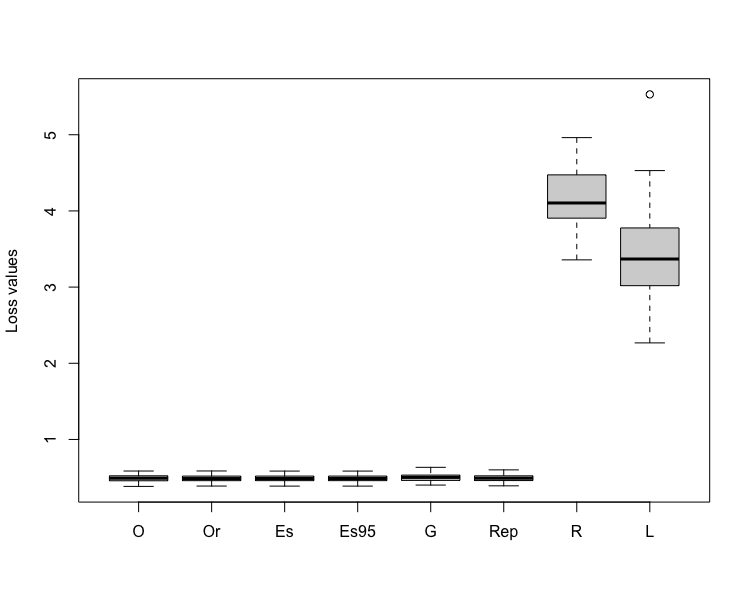}
         \caption{$(0.3,2)$}
         \label{fig:tdl252}
     \end{subfigure}\\
 \caption{Boxplots of the observed same-X loss values from the 50 replications in the simulation settings of Section \ref{fullrankmtx} when $\psi=0$. 
 We suppress ``2n" from the third to the last estimators for notational simplicity.
 Each plot is labeled with the $(s,{\rm Coding~number})$ value used.
 }
\label{fig:XT.low}
\end{figure}

\subsection{High-dimensional experiments}\label{highsimulsection}
We used the same data generating model as in Section \ref{lowsimulsection}
except that $n=200$, $p=300$, and $\sigma\in\{2,3\}$.  
Since \textbf{2n-Es} is not applicable in high dimensions, we tested variants of \textbf{2n-G}
that used 5, 10, and $n$-fold cross validation. 
However, except for a few cases, the number of folds used for tuning parameter selection for \textbf{2n-G} did not have a significant impact on the performance. 
So we only present results from 10-fold cross validation with the same label \textbf{2n-G} as in the previous sections.
In addition, we tried different $\alpha$ values that control the matrix $K$ in \eqref{highdimcorrected} for the following variants of \textbf{2n-Rep}:
\begin{itemize}
    \item \textbf{2n-Rep1}: Same as original \textbf{2n-Rep} with $\alpha=n(2\|Y-P_1Y\|^2)^{-1}$.
    \item \textbf{2n-Rep2}: $\Bar\gamma_c$ \eqref{highdimcorrected} using $\check\sigma^2$ with $\alpha=n^{3/2}(2\|Y-P_1Y\|^2)^{-1}$.
    \item \textbf{2n-Rep3}: Same as original \textbf{2n-Rep} with $\alpha=n^2(2\|Y-P_1Y\|^2)^{-1}$. 
    \item \textbf{2n-Rep4}: Same as original \textbf{2n-Rep} with $\alpha=n^3(2\|Y-P_1Y\|^2)^{-1}$.
\end{itemize}

\textbf{2n-Rep1}, \textbf{2n-Rep3}, and \textbf{2n-Rep4} are based on our variance estimator \eqref{highdimour}; 
and \textbf{2n-Rep2} is based on the variance estimator proposed by \citet{Liu2020EstimationOE}.
From Proposition \ref{highdimconsistency},
\textbf{2n-Rep1} is suitable when $\delta^2=\|\mu-P_1\mu\|^2$ is large relative to $n\sigma^2$, e.g.  $\gamma_{\rm opt}\rightarrow 1$.  In contrast, based on Proposition \ref{highdimconsistency} and \ref{highdimconsistency:2}, \textbf{2n-Rep2},
\textbf{2n-Rep3}, and \textbf{2n-Rep4} are suitable when $\delta^2=\|\mu-P_1\mu\|^2$ is small relative to $n\sigma^2$, e.g.  $\gamma_{\rm opt}\rightarrow 0$.
As we will illustrate below, \textbf{2n-Rep3} generally outperforms \textbf{2n-Rep2} and
\textbf{2n-Rep4}.
So we recommend using \textbf{2n-Rep1} for stronger signals and \textbf{2n-Rep3} for weaker signals.

In Figure \ref{fig:highdimbox}, 
we display side-by-side boxplots of the observed same-X losses from the 50 replications when $n=200$ and $p=300$.  There are additional boxplots displayed in Figure \ref{fig:highdimadd} in Section \ref{addedboxplots} of the Supplementary material.
Further numerical summaries of the results are in Table \ref{table:highdim.samex.1}--\ref{table:highdim.samex.3} in Section \ref{appd:added.high.tables} of the Supplementary material.
In general, when $\delta^2=\|\mu-P_1\mu\|^2$ was large relative to $n\sigma^2$, which corresponds to the situation $\gamma_{\rm opt}\rightarrow 1$, \textbf{2n-Rep1} performed substantially better than Ridge, Lasso, \textbf{2n-G}, \textbf{2n-Rep2}, and \textbf{2n-Rep4} (Figure \ref{fig:hi7}, \ref{fig:hi17}). 
The method \textbf{2n-Rep1} also outperformed \textbf{2n-Rep3}, but the same-X prediction difference was relatively smaller than the others. 
Furthermore, larger $\delta^2$ led to improved tuning-parameter selection for \textbf{2n-Rep1}. 
However, \textbf{2n-G}, Ridge, and Lasso performed better when $\delta^2$ was small relative to $n\sigma^2$. In this setting, which corresponds to $\gamma_{\rm opt}\rightarrow 0$, \textbf{2n-Rep2}, \textbf{2n-Rep3}, and \textbf{2n-Rep4} performed as well as or better than Ridge, Lasso, and \textbf{2n-G} (see Figure \ref{fig:hi5}, \ref{fig:hi15}, \ref{fig:hi25}). On the other hand, \textbf{2n-Rep1} struggled for this case. 
The method \textbf{2n-Rep3} was relatively more stable than the other \textbf{2n-Rep} versions, and it consistently outperformed Ridge and Lasso regardless of $\delta^2$.

In Figure \ref{fig:highdimgr} (see Section \ref{addedboxplots} of the Supplementary material), we display average loss values over the 50 replications as a function of $\lambda$ over $\lambda\in\{10^{-7+0.01k}:j=0,1,\ldots,1100\}$ when $n=200$ and $p=300$ (see \eqref{lam.gr.FVS} and \eqref{lam.gr.Ridge}).
Further numerical summaries of these results are in Table \ref{table:graph table 4} in Section \ref{appd:added.high.tables} of the Supplementary material.
These results look similar to lower-dimensional results displayed in Figure \ref{fig:lowdimgr1}--\ref{fig:lowdimgr3} (see Section \ref{addedboxplots} of the Supplementary material), except the curve valleys are narrower for fitted-value shrinkage. 

\begin{figure}
     \centering
     \begin{subfigure}[b]{0.45\textwidth}
         \centering
         \includegraphics[width=\textwidth]{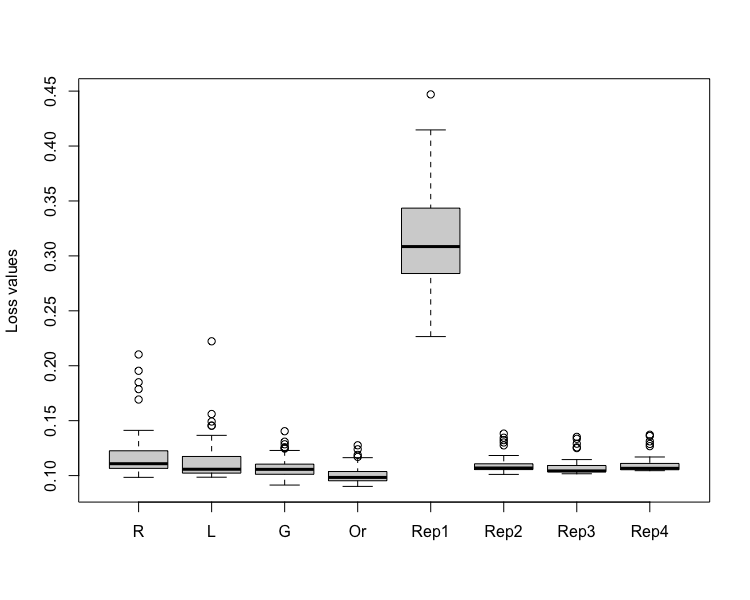}
         \caption{$(1,10^{-0.5})$}
         \label{fig:hi5}
     \end{subfigure}
     \begin{subfigure}[b]{0.45\textwidth}
         \centering
         \includegraphics[width=\textwidth]{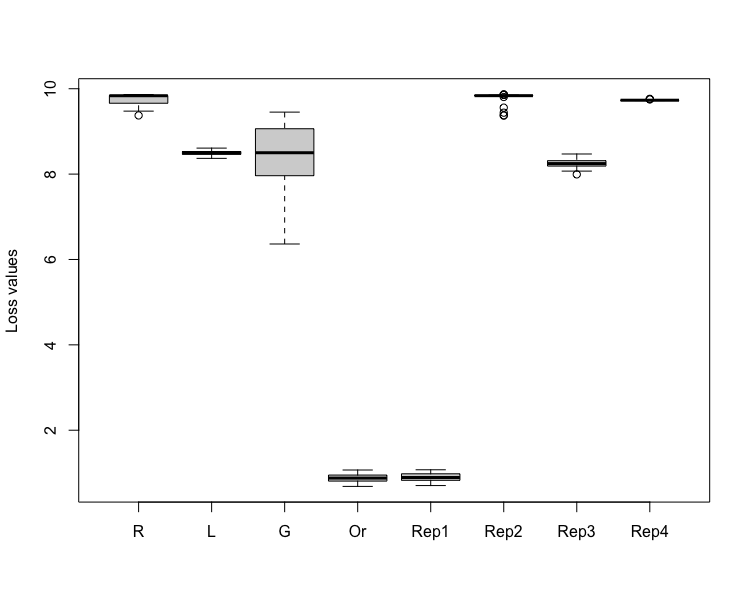}
         \caption{$(1,10^{0.5})$}
         \label{fig:hi7}
     \end{subfigure}\\
     
     \begin{subfigure}[b]{0.45\textwidth}
         \centering
         \includegraphics[width=\textwidth]{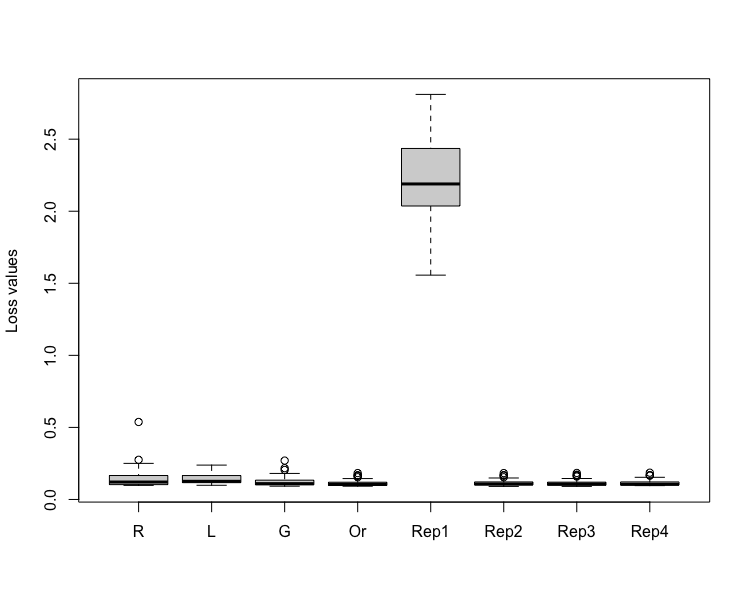}
         \caption{$(2,10^{-0.5})$}
         \label{fig:hi15}
     \end{subfigure}
     \begin{subfigure}[b]{0.45\textwidth}
         \centering
         \includegraphics[width=\textwidth]{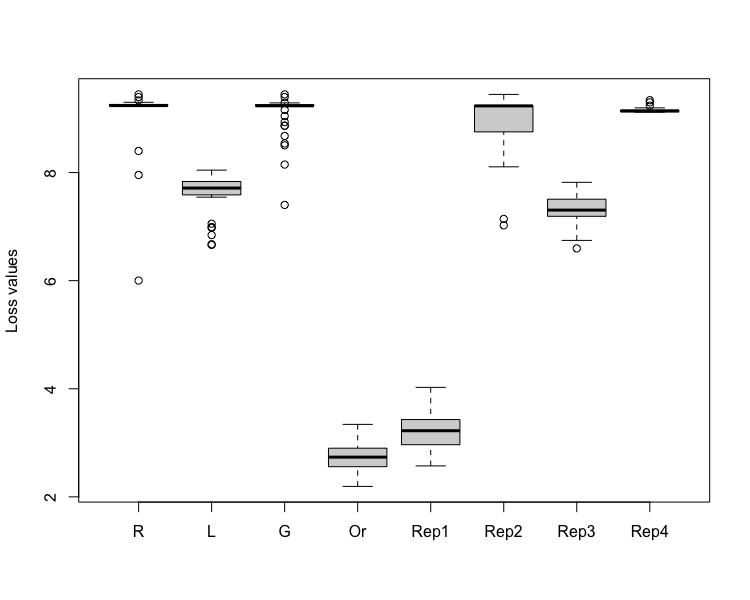}
         \caption{$(2,10^{0.5})$}
         \label{fig:hi17}
     \end{subfigure}\\
     \begin{subfigure}[b]{0.45\textwidth}
         \centering
         \includegraphics[width=\textwidth]{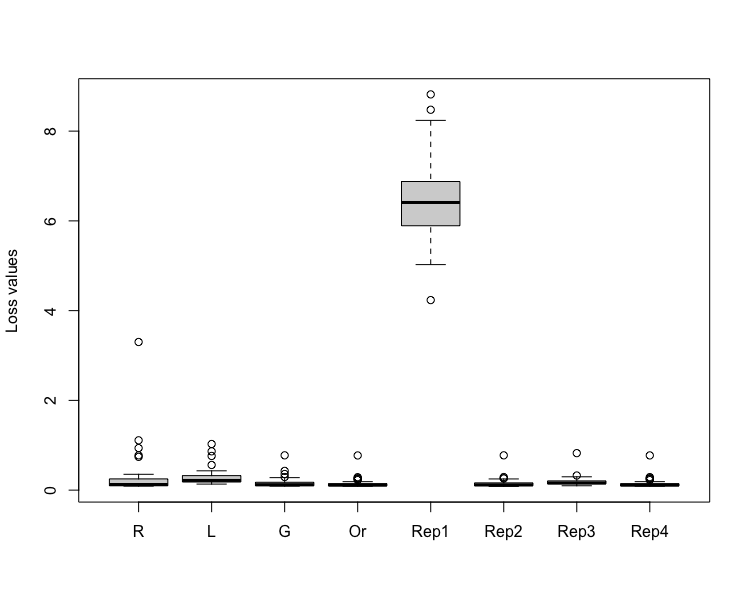}
         \caption{$(3,10^{-0.5})$}
         \label{fig:hi25}
     \end{subfigure}
     \begin{subfigure}[b]{0.45\textwidth}
         \centering
         \includegraphics[width=\textwidth]{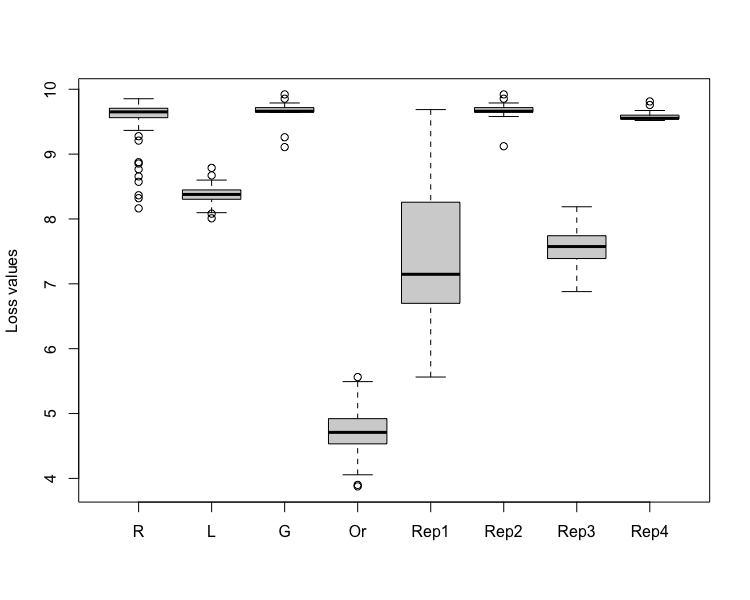}
         \caption{$(3,10^{0.5})$}
         \label{fig:hi27}
     \end{subfigure}\\
 \caption{
 Boxplots of the observed same-X loss values from the 50 replications
 when $(n,p)=(200,300)$. Each plot is labeled with the $(\sigma,\tau)$ value used.}
\label{fig:highdimbox}
\end{figure}

\subsection{Impact of (approximate) invariance III: General full-rank transformation of the design matrix (High dimension)}\label{fullrankmtx.high}
We again generated $X\in\RR^{n\times p}$ to have ones in its first column and
and its $n$ row vectors (excluding the first entry) were drawn independently from
$N_{p-1}(0,\Sigma)$ where $\Sigma_{jk}=0.5^{|j-k|}$ and $(n,p)=(200,300)$. 
We used \eqref{gencasesvec} again for the data generating procedure, where we consider $\sigma\in\{1,3\}$.

We generated $\beta=u*v$,
where each element of $u\in\RR^p$ is an independent draw from the Uniform distribution on $(1/4,1/2)$ when $\sigma=1$; and on $(1/6,1/3)$ when $\sigma=3$; $v=(1,v_{-1})\in\RR^{p}$, where each element of $v_{-1}\in\RR^{p-1}$ is an independent draw from ${\rm Ber}(s)$.
When $\sigma=1$, we vary $s\in\{0.025,0.05,0.1,0.2,0.3\}$, and consider \textbf{2n-Rep1} for the proposed estimator.
On the other hand, when $\sigma=3$, we vary $s\in\{0.01,0.02,0.03,0.04,0.05\}$, and consider \textbf{2n-Rep3} for the proposed estimator as well as \textbf{2n-Rep2}, which is known to be efficient in the same setting that \textbf{2n-Rep3} is (see Proposition \ref{highdimconsistency}).
We dropped \textbf{2n-Rep4}, since it showed relatively unstable performance in comparison to \textbf{2n-Rep3} in Section \ref{highsimulsection}.

As in Section \ref{fullrankmtx}, we denote the model fitting scheme with this $X$ as ``Coding-1".  We also use a transformed design matrix $X_\bullet=XT$, where $T$ through a Gram-Schmidt orthogonalization of a matrix with all of its entries drawn independently from $N(0,1)$. 
We refer the transformed design matrix $X_\bullet$
as ``Coding-2".

In Figure \ref{fig:XT.high}, 
we display side-by-side boxplots of the observed same-X losses from the 50 replications in few representative settings.
Further numerical summaries of the simulation results are in Table \ref{table:high dim XT table1}--\ref{table:high dim XT table2} in Section \ref{appd:fullrank.high} of the Supplementary material. When $\tau=1$, Lasso was generally the best in Coding-1 and Ridge was second best.  However, in Coding-2, \textbf{2n-Or} was the best and \textbf{2n-Rep1} was second best, regardless of $\psi$ and $s$.  Ridge and Lasso performed  significantly worse than \textbf{2n-Rep1} after transformation.

On the other hand, when $\sigma=3$, \textbf{2n-Or} achieved the best performance for both Coding 1 and 2 except for $s=0.05$ where Ridge was the best for both Coding 1 and 2.  For $s\in\{0.01,0.02,0.03,0.04\}$, \textbf{2n-Rep3} was competitive, so were \textbf{2n-G} and \textbf{2n-Rep2}, which is similar to the pattern observed in Section \ref{highsimulsection} when $\delta^2$ was low.

Although the \textbf{2n-Rep} methods lack exact invariance in high dimensions, their performance difference between Coding 1 and Coding 2 was not significant.

\begin{figure}
     \centering
     \begin{subfigure}[b]{0.35\textwidth}
         \centering
         \includegraphics[width=\textwidth]{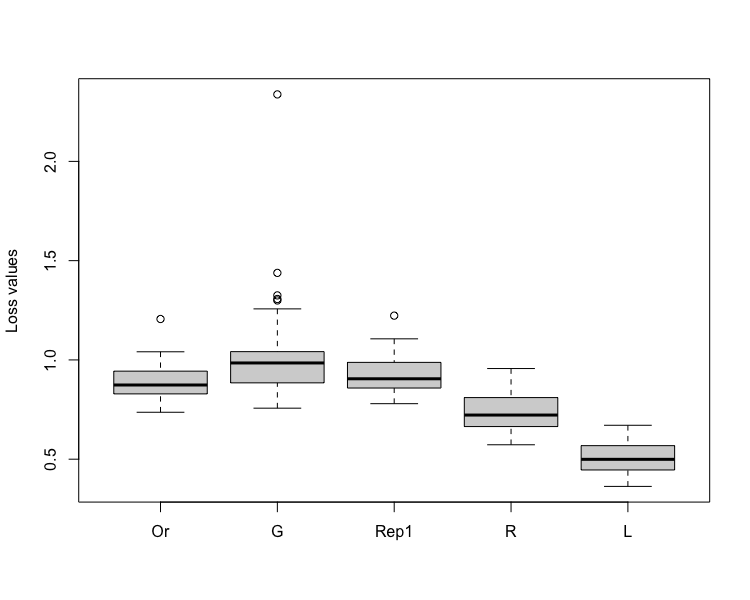}
         \caption{$(1,0.1,1)$}
         \label{fig:tdl331}
     \end{subfigure}
     \begin{subfigure}[b]{0.35\textwidth}
         \centering
         \includegraphics[width=\textwidth]{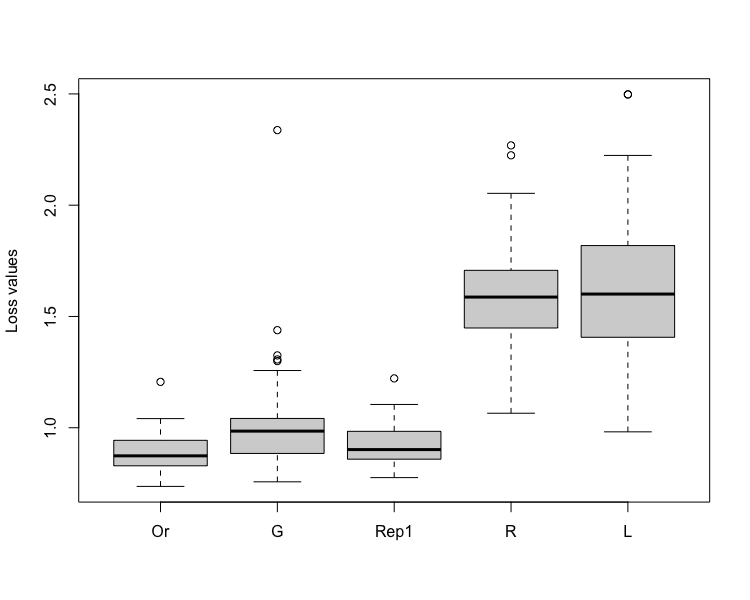}
         \caption{$(1,0.1,2)$}
         \label{fig:tdl332}
     \end{subfigure}\\
     \begin{subfigure}[b]{0.35\textwidth}
         \centering
         \includegraphics[width=\textwidth]{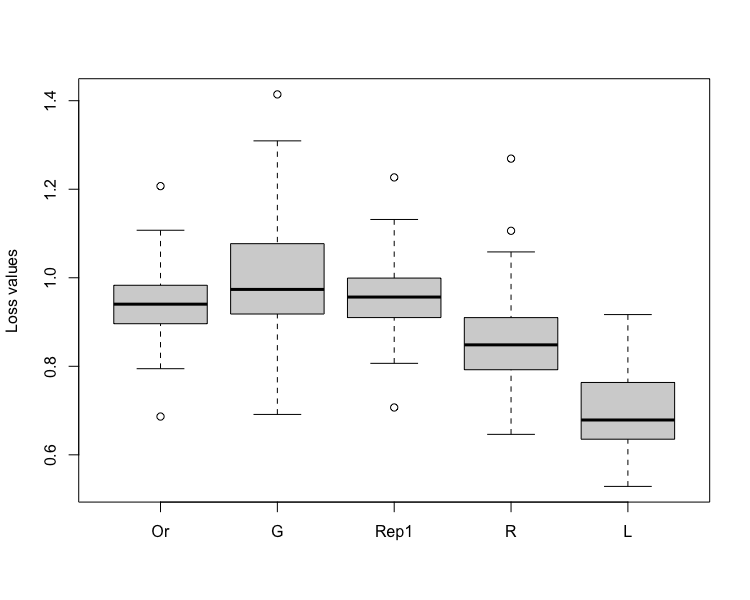}
         \caption{$(1,0.2,1)$}
         \label{fig:tdl341}
     \end{subfigure}
     \begin{subfigure}[b]{0.35\textwidth}
         \centering
         \includegraphics[width=\textwidth]{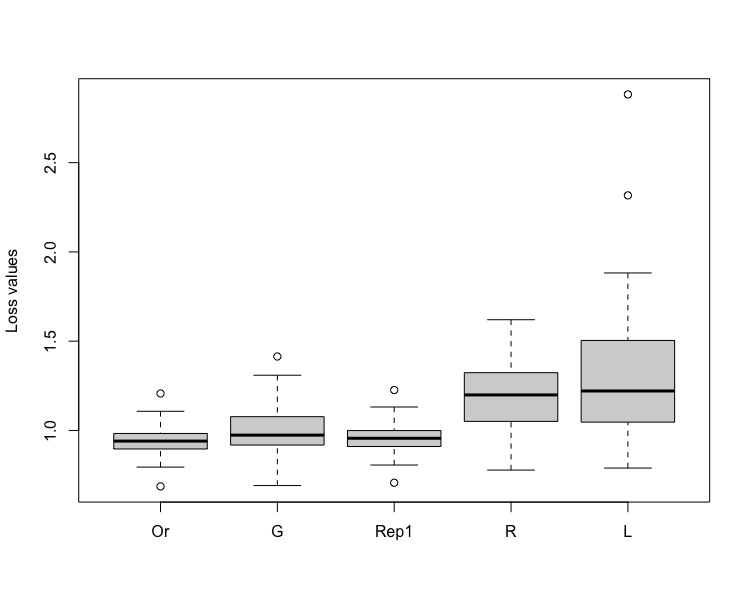}
         \caption{$(1,0.2,2)$}
         \label{fig:tdl342}
     \end{subfigure}\\
     \begin{subfigure}[b]{0.35\textwidth}
         \centering
         \includegraphics[width=\textwidth]{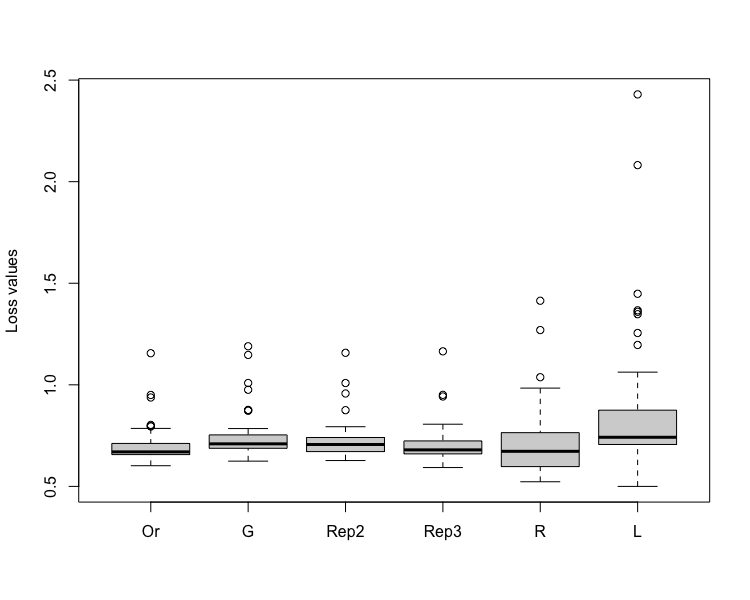}
         \caption{$(3,0.03,1)$}
         \label{fig:tdl381}
     \end{subfigure}
     \begin{subfigure}[b]{0.35\textwidth}
         \centering
         \includegraphics[width=\textwidth]{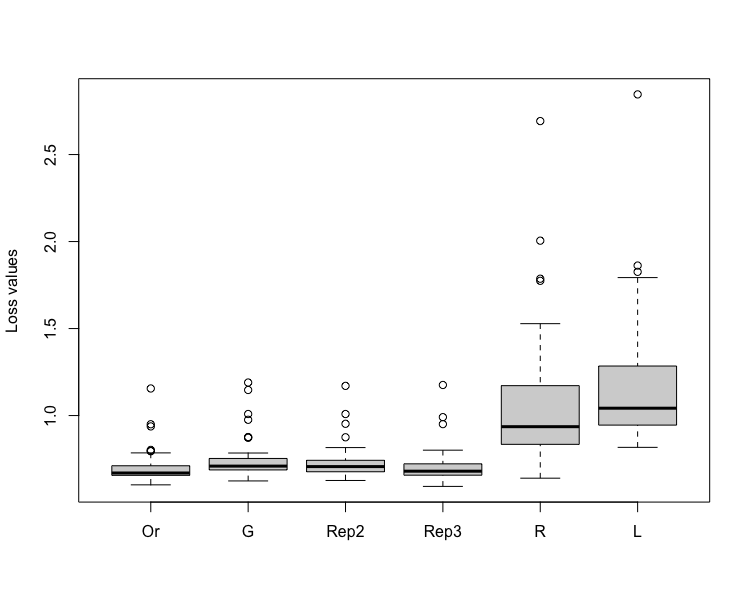}
         \caption{$(3,0.03,2)$}
         \label{fig:tdl382}
     \end{subfigure}\\
    \begin{subfigure}[b]{0.35\textwidth}
         \centering
         \includegraphics[width=\textwidth]{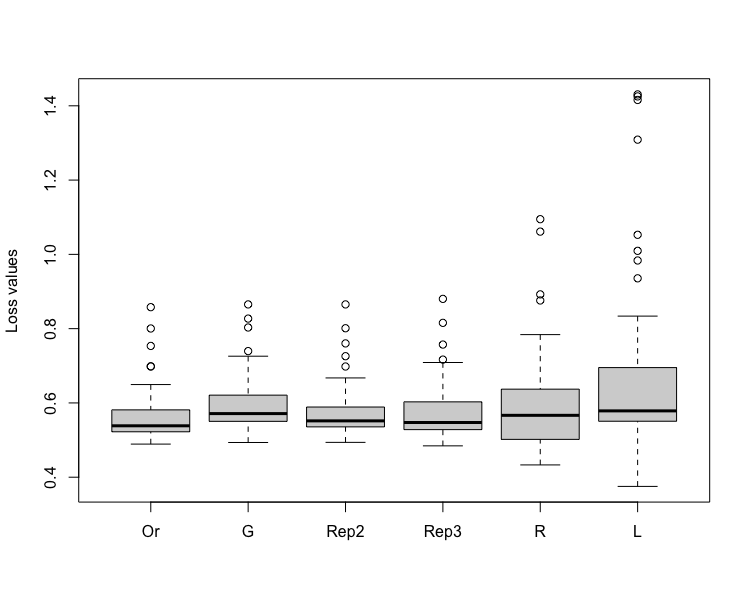}
         \caption{$(3,0.04,1)$}
         \label{fig:tdl391}
     \end{subfigure}
     \begin{subfigure}[b]{0.35\textwidth}
         \centering
         \includegraphics[width=\textwidth]{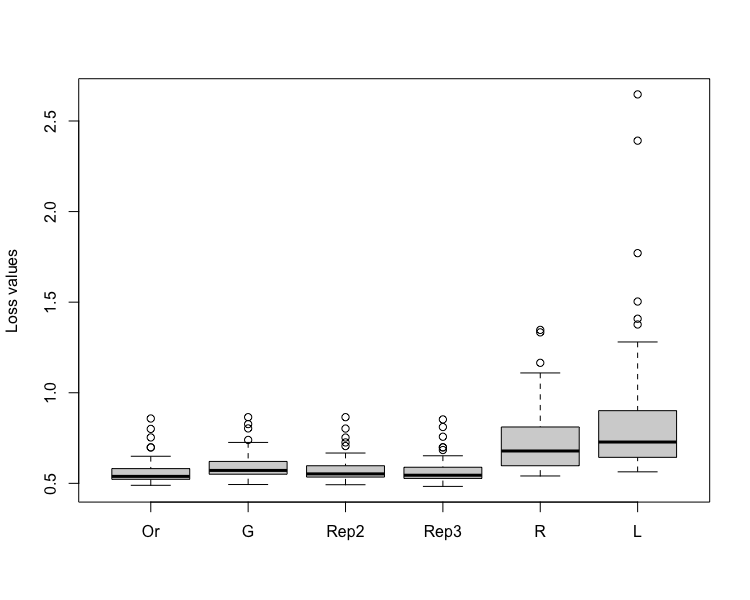}
         \caption{$(3,0.04,2)$}
         \label{fig:tdl392}
     \end{subfigure}\\
 \caption{Boxplots of the observed same-X loss values from the 50 replications in the simulation settings of Section \ref{fullrankmtx.high}. 
 We suppress ``2n" from the third to the last estimators for notational simplicity.
 Each plot is labeled with the $(\sigma,s,{\rm Coding~number})$ value used.
 }
\label{fig:XT.high}
\end{figure}

\section{Data examples}

\subsection{Low dimensional data experiments}\label{lowdimreal}
We compared our proposed fitted-value shrinkage procedures to competitors 
on three data sets. We used the same 
non-oracle estimators as the previous section except we excluded \textbf{2n-Es90} because it performed similarly to \textbf{2n-Es95}. 
Each data example was analyzed using the following procedure:
For 50 independent replications, we randomly selected 70\% of the subjects for the training set and used the remaining subjects as the test set.  Tuning parameter selection was done using the training set and prediction performance was measured using squared error loss on the test set. The following is the short description of the three low-dimensional data set we examined.
\begin{enumerate}
    \myitem{(FF)}: \label{data:FF} The Forest Fire \ref{data:FF} data are from from \citet{ff} and are stored at the UCI Machine learning repository via \url{https://archive.ics.uci.edu/dataset/162/forest+fires}. There are 517 observations corresponding to forest fires in Portugal from 2000 to 2003. The response is the total burned $\textbf{area}$ (in $ha$) from the fire, which was transformed with $x\mapsto{\rm ln}(x+1)$, which was suggested by \citet{ff}. There were originally 13 attributes. However, in the pre-processing step, since we are not focusing on spatio-temporal methods, we excluded time, date and location coordinates.  After this processing, the full-data design matrix had $(n,p)=(517,9)$ with 8 numerical-variable columns and one intercept column.
    \myitem{(GDP)}: \label{data:GDP} The GDP data \ref{data:GDP} are from \citet{barroorigin}. These data consist of 161 observations of GDP growth rates for the two periods 1965-1975 and 1975-1985. The data are also in the R package {\tt quantreg} \citep{quantregpackage}. The response is \textbf{Annual change per capita GDP}.  There are 13 numerical predictors, e.g. Initial per capita GDP, Life expectancy.
    We also added a quadratic term for the predictor \textbf{Black Market Premium}.  After processing, the full-data design matrix has $(n,p)=(161,15)$. 
    \myitem{(FC)}: \label{data:Forecast} The Forecast data set \ref{data:Forecast} is from \citet{bcuref} for the purpose of bias correction for the Local Data Assimilation and Prediction System (LDAPS), which is a numerical weather report model used by Korea  Administration (KMA), Seoul, South Korea. It has a public access through \url{https://archive.ics.uci.edu/dataset/514/bias+correction+of+numerical+prediction+model+temperature+forecast}. The data are regional observations from 2013 to 2017, from which we randomly selected 500. We used the true maximal temperature of the next day as the response and removed date, station ID, and true minimal temperature of the next day.
    The full-data design matrix had $(n,p)=(500,20)$.
\end{enumerate}

In Table \ref{table:real data table}, we display mean squared prediction errors averaged over 50 training/test set splits for
the three data examples \ref{data:FF}, \ref{data:GDP}, and \ref{data:Forecast}. Our fitted-value shrinkage estimators performed similarly to Ridge and Lasso, which both lack invariance to invertible linear transformations of the design matrix.

\begin{table}[h!]
\centering
\resizebox{\columnwidth}{!}{%
\begin{tabular}{ |c||c|c|c|c|c|c|c|  }
 \hline
 \multicolumn{8}{|c|}{Performance comparison table} \\
 \hline
 Data set &  2n-Rep & 2n-G & 2n-Es & 2n-Es95 & OLS & Ridge & LASSO\\
 \hline
 Forest Fire  & 2.0031 & 2.0160 & \underline{2.0017} & \textbf{1.9792} & 2.1834 & 2.0189 & 2.0572\\
 \ref{data:FF} & (0.0301) &  (0.0316) & (0.0300) & (0.0293) & (0.0551) & (0.0322) & (0.0411)\\
 \hline
    GDP growth  & \textbf{3.106e-04} & 3.125e-04 & \underline{3.139e-04} & 3.139e-04 & 3.231e-04 & 3.112e-04 & 3.149e-04\\
 \ref{data:GDP} & (7.465e-06) &  (7.615e-06) & (7.706e-06) & (7.706e-06) & (8.120e-06) & (8.126e-06) & (8.307e-06)\\   
  \hline
  Forecast & 1.9045 & 1.8997 & \underline{1.9028} & \textbf{1.8933} & 2.0489 & 1.8990 & 1.8978\\
\ref{data:Forecast}  & (0.0327) & (0.0325) & (0.0327) & (0.0318) & (0.0371) & (0.0321) & (0.0318) \\   
  \hline

\end{tabular}
}
 \caption{The performance comparison table for three data examples in Section \ref{lowdimreal}. The values are the mean squared prediction errors averaged over 50 training/test set splits.  The numbers in parentheses are normalized sample standard deviations. The column labels are defined in Section \ref{lowsimulsection}. Boldface indicates the best model. Underlined is our main proposed estimator.}
\label{table:real data table}
\end{table}

\subsection{Low dimensional data analyses with categorical variables and their interactions}\label{lowdimcate}
We analyzed two data sets from existing R packages to illustrate the performance of our estimators when categorical variables with interactions are present in the model. The competitors and setup for the data experiments are nearly identical to the previous Section \ref{factorsimul}, 
except we added three new fitted-value shrinkage estimators that shrink
toward the submodel without interactions instead of the intercept-only model.
We refer the readers to Section \ref{submodelshrink} of the Supplementary material for the definition of the submodel shrinkage.
These new submodel shrinkage methods are labeled 
\textbf{2n-Repsb}, \textbf{2n-Gsb}, \textbf{2n-Essb}, \textbf{2n-Es95sb},
and they respectively correspond to \textbf{2n-Rep}, \textbf{2n-G}, \textbf{2n-Es}, and \textbf{2n-Es95}.

The following is a description of the data examples:
\begin{enumerate}
    \myitem{(Dia-1)}:\label{data:dia1} The Diamonds data set is from Diamonds data frame in the R package {\tt Stat2Data} \citep{stat2datapackage}, and it was obtained from \url{https://awesomegems.com/}. The are $n=351$ subjects and the response is the price of the diamond (in dollars). We divided price per 1000 and used it for the response variable. We further removed total price from the predictors. There are 2 categorical predictors: color (with levels D to J) and clarity (with levels IF, VVS1, VVS2, VS1, VS2, SI1, SI2, and SI3). We divided color into 5 levels (D, E, F, G, and (H,I,J)), and categorized the clarity into 3 levels ((IF, VVS1, VVS2), (VS1, VS2), (SI1, SI2, SI3)). We used reference-level coding in the design matrix, where (H,I,J) was the reference level for color; and (SI1, SI2, SI3) was the reference level for clarity. Interactions between color and depth as well as clarity and depth were added.  The full design matrix has $(n,p)=(351,15)$ and the submodel with linear terms only has $p=9$.
\myitem{(Dia-2)}:\label{data:dia2} The setting is identical to that of \ref{data:dia1}, except we used the category (VS1, VS2) as the reference level for coding the categorical predictor clarity.

\myitem{(NG-1)}: \label{data:ng1} The NaturalGas data is from \citet{bal2022}, and is in the R package {\tt AER} \citep{aerpackage}.
There are 138 observations on 10 variables.  We removed state name and year
and added an interaction between state code and heating degree days. 
We set the response as consumption divided by 10000.
The reference level for state code, which is the only categorical predictor, 
was set to 35 (NY).  The full-data design matrix had
$(n,p)=(138,17)$ and the submodel without interactions had $p=12$.

\myitem{(NG-2)}: \label{data:ng2} This is the same as \ref{data:ng1}, except the reference level for state code was set to 5 (CA).
\end{enumerate}

We display mean squared prediction errors averaged over 50 training/test set splits for these data examples in Table \ref{table:real data table2}. Our proposed estimators performed similarly or better than Ridge and Lasso.  We also notice that changing the way that categorical predictors were encoded in the design matrix changes the performance of Ridge and Lasso, which lack invariance. Furthermore, Group Lasso and its two standardized versions had unstable performance when the base level coding was changed, which does
not happen to our invariant methods.

\begin{table}[h!]
\centering
\resizebox{\columnwidth}{!}{%
\begin{tabular}{ |c||c|c|c|c|c|c|c|c|c|c|c|c|c|c|  }
 \hline
 \multicolumn{15}{|c|}{Performance comparison table} \\
 \hline
 Data set &  2n-Rep & 2n-Repsb & 2n-G & 2n-Gsb & 2n-Es & 
2n-Essb & 2n-Es95 & 2n-Es95sb & OLS & Ridge & LASSO & GL & MGL & SGL\\
 \hline
 Diamonds & 1.2538 & 1.2557 & 1.2545 & \textbf{1.2501} & \underline{1.2557} & \underline{1.2507} & 1.2557 & 1.2599 & 1.2587 & 1.2668 & 1.2675 & 1.2699 & 1.2693 & 1.2569 \\
\ref{data:dia1} & (0.0303) &  (0.0300) & (0.0298) & (0.0293) & (0.0295)  & (0.0294) & (0.0295) & (0.0294) & (0.0293) & (0.0297) & (0.0297) & (0.0290) & (0.0290) & (0.0301)\\
 \hline
 Diamonds & 1.2538 & 1.2557 & 1.2545 & \textbf{1.2501} & \underline{1.2557} & \underline{1.2507} & 1.2557 & 1.2599 & 1.2587 & 1.2615 & 1.2624 & 1.2697 & 1.2677 & 1.2592 \\
\ref{data:dia2} & (0.0303) &  (0.0300) & (0.0298) & (0.0293) & (0.0295)  & (0.0294) & (0.0295) & (0.0294) & (0.0293) & (0.0296) & (0.0299) & (0.0286) & (0.0287) & (0.0296)\\ 
 \hline
 Natural gas & 5.2558 & 5.2485 & 5.2899 & 5.2711 & \underline{5.2559} & \underline{5.2523} & 5.2559 & 5.2523 & 5.2545 & 5.2913 & \textbf{5.1254} & 5.2972 & 5.3010 & 5.1584 \\
 \ref{data:ng1} & (0.2285) &  (0.2262) & (0.2282) & (0.2217) & (0.2285) & (0.2257) & (0.2285) & (0.2257) & (0.2279) & (0.2195) &(0.2285) & (0.2241) & (0.2242) & (0.2216)\\
 \hline
 Natural gas & 5.2558 & 5.2486 & 5.2899 & 5.2711 & \underline{5.2559} & \underline{5.2523} & 5.2559 & 5.2523 & 5.2545 & 5.3563 & 5.2889 & 5.2261 & \textbf{5.2010} & 6.7549 \\
 \ref{data:ng2} & (0.2285) &  (0.2262) & (0.2282) & (0.2217) & (0.2285) & (0.2257) & (0.2285) & (0.2257) & (0.2279) & (0.2195) &(0.2285) & (0.2171) & (0.2168) & (0.2877) \\ 
 \hline
\end{tabular}
}
 \caption{The performance comparison table for two data sets in  Section \ref{lowdimcate} with two different coding strategies. The values are mean squared prediction errors averaged over 50 training/test splits. The numbers in parentheses are normalized sample standard deviations. The labels are defined in Section \ref{lowsimulsection} except that last three are illustrated in Section \ref{lowdimcate}. The suffix "sb" for each estimator corresponds to shrinking towards the submodel without interactions instead of shrinking towards the intercept-only model (see Section \ref{submodelshrink} for details). Boldface indicates the best model. Underlined are our main estimator and its submodel shrinkage variant.}
\label{table:real data table2}
\end{table}

\subsection{High dimensional data experiments}\label{highdimreal}
For high-dimensional data examples, we randomly selected subjects
from existing data sets so that there were fewer subjects than predictors.
We used the same splitting and evaluation procedure that we used in Sections \ref{lowdimreal} and \ref{lowdimcate}. The competitors are same as those considered in Section \ref{highsimulsection} except that we excluded \textbf{2n-Rep4} which had nearly identical performance to \textbf{2n-Rep3} in the simulation study. The following is a description of the examples: 

\begin{enumerate}
    \myitem{(mtp)}:\label{data:mtp} The data set mtp comes from \citet{mtpdata} and is available at the OpenML repository via \url{https://www.openml.org/search?type=data&status=active&id=405}. There are 4450 subjects with 203 numerical measurements. The response is oz203.   We randomly selected 120 subjects and removed the 23 predictors that had fewer than 30 distinct values,  which ensured that there were no constant columns in the 120-row design matrix other than the intercept column.  The full-data design matrix had $(n,p)=(120,180)$.
    
    \myitem{(topo)}:\label{data:topo}  The topo.2.1 data set is from \citet{topodata} and is available through the OpenML repository via \url{https://www.openml.org/search?type=data&sort=runs&id=422&status=active}.  There are 8885 subjects with 267 numerical measurements. The response is oz267.  We randomly selected 180 subjects and removed the 22 predictors that had fewer than 30 distinct values.  After this, there were 34 constant columns (other than the intercept) that were also removed.  The R code for this processing is in Section \ref{mtpcodes} of the Supplementary material. The full-data design matrix has $(n,p)=(180,214)$.
    \myitem{(tecator)}:\label{data:tecator} The tecator data set comes from \citet{Tecatordataset}, and is also available from OpenML repository via \url{https://www.openml.org/search?type=data&status=active&sort=runs&order=desc&id=505}. We randomly selected 100 subjects and removed 22 principal components. The response is fat content. The full-data design matrix had $(n,p)=(100,103)$.
\end{enumerate}

In Table \ref{table:real data table3}, we report the 
mean squared prediction errors averaged over 50 training/test set splits for these data examples.  We see that \textbf{2n-Rep2}, \textbf{2n-Rep3} had similar prediction performance compared to \textbf{2n-Rep1}. In contrast to its same-X loss performance in simulations, the cross validation version of our method \textbf{2n-G} gave reasonable out-of-sample prediction performance. Generally, \textbf{2n-Rep1} and \textbf{2n-Rep2} performed competitively compared with Ridge and Lasso, and \textbf{2n-Rep3} followed the next within \textbf{2n-Rep}s.

\begin{table}[h!]
\centering
\resizebox{\columnwidth}{!}{%
\begin{tabular}{ |c||c|c|c|c|c|c|c|  }
 \hline
 \multicolumn{8}{|c|}{Performance comparison table} \\
 \hline
 Data set &  OLS & 2n-G & Ridge & LASSO & 2n-Rep1 & 2n-Rep2 & 2n-Rep3 \\
 \hline
 mtp & 0.5075 & 0.0266 & 0.0262 & \textbf{0.0227} & 0.0257 & 0.0261 & 0.0258 \\
 \ref{data:mtp} & (0.0592) &  (0.0010) & (0.0027) & (0.0017) & (0.0008) & (0.0009) & (0.0008) \\
 \hline
 topo.2.1 & 0.07524 & 0.00085 & 0.00096 & 0.00087 & 0.00084 & \textbf{0.00083} & 0.00084 \\
 \ref{data:topo} & (2.72e-02) &  (2.75e-05) & (7.71e-05) & (3.27e-05) & (3.08e-05) & (2.73e-05) & (3.08e-05) \\ 
 \hline
 tecator & 2.4400 & 2.4592 & 2.7251 & 2.7120 & \textbf{2.4266} & 2.4318 & 2.4733 \\
 \ref{data:tecator} & (0.2359) &  (0.2348) & (0.1913) & (0.1940) & (0.2202) & (0.2188) & (0.2148) \\
 \hline
\end{tabular}
}
 \caption{The performance comparison table for three high dimensional real data sets in Section \ref{highdimreal}. The values are mean squared prediction errors averaged over 50 replications.  The numbers in parentheses are normalized sample standard deviations. The labels are defined in Section \ref{highsimulsection}. Boldface indicates the best model.}
\label{table:real data table3}
\end{table}

\section{Discussion}
Lasso and ridge-penalized least squares are popular and powerful in practice.  However, their fitted values lack invariance to invertible linear transformations of the design matrix, which is undesirable when there are categorical predictors and interactions.  Our simulation studies and data analyses illustrated that our proposed method performed comparably to ridge-penalized least squares, and so we recommend that practitioners use our method
in any problem that they would use ridge-penalized least squares.  
Our method serves as a companion to ridge-penalized least squares, with the advantage of preserving invariance to invertible linear transformations of the design matrix.

The fitted-value shrinkage idea presented here to fit linear regression models with invariance can be extended to more complicated settings.  
For example, one could fit a logistic regression model by minimizing the negative loglikelihood plus the penalty $\lambda \|X\beta - \hat w 1_n\|^2$, where $\hat w$ is the sample log-odds that the response takes its first category. 
We are currently developing this procedure.  

A Bayesian formulation of our method may also be interesting.  
In addition, one could study methods that combine our proposed invariant shrinkage penalty with regular shrinkage penalties like the lasso or ridge.

\section*{Acknowledgement}
The authors thank the editor, associate editor, and referees for their helpful comments and suggestions.


\clearpage

\title{Supplementary material for Fitted value shrinkage}
\maketitle

\begin{alphasection}
\section{Technical details}

\subsection{Proof of Proposition \ref{emsefvs2}}\label{appendthm1}
\begin{proof}
We start with the following standard decomposition:
\begin{align*}
    \mathbb{E}\|X\hat\beta^{(\gamma)}-X\beta\|^2={\rm tr}\Big({\rm var}\Big(X\hat\beta^{(\gamma)}-X\beta\Big)\Big)+
    \Big\|\mathbb{E}\Big(X\hat\beta^{(\gamma)}-X\beta\Big)\Big\|^2.
\end{align*}
From \eqref{mainestimator}, 
\begin{align*}
    {\rm tr}\Big({\rm var}\Big(X\hat\beta^{(\gamma)}-X\beta\Big)\Big)&={\rm tr}\Big(\sigma^2 \Big(\gamma^2 P_X +(1-\gamma^2) P_1\Big) \Big),\\
    &=\sigma^2(\gamma^2 r+1-\gamma^2)\\
    \mathbb{E}\Big(X\hat\beta^{(\gamma)}-X\beta\Big)&=(1-\gamma)(\mu-P_1\mu),
\end{align*}
where $\mu = X\beta$. Combining above two equations, we conclude that the statement holds.
\end{proof}

\subsection{Proof of Proposition \ref{consistency}}\label{appendthm2}

\begin{lemma}\label{quadformvariance}
When a $n$-variate random variable $V=(V_1,\dots,V_n)$ has i.i.d. elements $V_i$ following a distribution that has mean $\mu$, variance $\sigma^2$, and has finite $4$-th moment with $\mathbb{E}(V_i^4)\le M$, for all $i=1,\dots,n$. Let an $(i,j)$-th element of $A$ be $a_{ij}$, then, for a symmetric non-negative definite $A\in\mathbb{R}^{n\times n}$,
\begin{align*}
    {\rm var}(V'AV)\le 2(M-3\sigma^4)\sum_{i=1}^n a_{ii}^2+4\sigma^4{\rm tr}(A^2)+8\sigma^2\mu'A^2\mu
\end{align*}
holds.
\end{lemma}
\begin{proof}
    For the simplest case, we first consider when $\mu=0$. Then, by the i.i.d property of $\{V_i\}$ and the moment conditions,
    \begin{align*}
    \mathbb{E}\Big[(V'AV)^2\Big] &\le M\sum_{i=1}^n a_{ii}^2+\sigma^4\sum\sum_{1\le i\neq j\le n} a_{ii} a_{jj} + 2\sigma^4\sum\sum_{1\le i\neq j\le n} a_{ij} a_{ji}\\
    &= M\sum_{i=1}^n a_{ii}^2 + \sigma^4 \Big(({\rm tr}(A))^2-\sum_{i=1}^n a_{ii}^2\Big) +2 \sigma^4 \Big({\rm tr}(A^2)-\sum_{i=1}^n a_{ii}^2\Big)\\
    &=(M-3\sigma^4)\sum_{i=1}^n a_{ii}^2+\sigma^4\Big(({\rm tr}(A))^2+2{\rm tr}(A^2)\Big)
    \end{align*}
    Together with, $\mathbb{E}(V'AV)=\sigma^2{\rm tr}(A)$, we know that,
    \begin{align*}
        {\rm var}\Big[(V'AV)^2\Big]\le (M-3\sigma^4)\sum_{i=1}^n a_{ii}^2+2\sigma^4{\rm tr}(A^2)
    \end{align*}

    For general $V$ with $\mu\neq0$, we denote $V=V_0+\mu$, then,
    \begin{align*}
      {\rm var}(V'AV)&={\rm var}(V_0'AV_0+2\mu'AV_0)\\
      &\le 2{\rm var}(V_0'AV_0)+8{\rm var}(\mu'AV_0).
    \end{align*}
    This completes the proof.
\end{proof}

Now we turn to the proof of Proposition \ref{consistency}.
\begin{proof}
Let $\delta^2=\|\mu - P_1 \mu\|^2$ and
define $A_n=\|P_XY-P_1Y\|^2/(r-1), B_n=\|Y-P_XY\|^2/(n-r)$ and their expected values $a_n=\sigma^2+{\delta^2}/({r-1}),b_n=\sigma^2$, respectively. Then we can express $\gamma_{\rm opt}=1-{1}/{f_n}$ with its sample counterpart $1-{1}/{F_n}$ where $F_n={A_n}/{B_n},f_n={a_n}/{b_n}$. At first glance, considering that $$
f_n = \frac{a_n}{b_n}=1+\frac{\delta^2}{\sigma^2({\rm rank}(X)-1)},
$$
it is obvious that $f_n\ge 1$. By the assumption, we have the error distribution to be $\mathbb{E}(e_i^4)\le M$, for $M>0,i=1,\dots,n$.

For the first step, knowing that $\mathbb{E}(B_n)=b_n$, 
\begin{align*}
    \mathbb{E}(B_n-b_n)&=0\\
    {\rm var}(B_n-b_n)&={\rm var}(B_n)\\
    & = \frac{1}{(n-r)^2} {\rm var} \|(I-P_X)Y\|^2\\
    & = \frac{1}{(n-r)^2} {\rm var} \left\{Y'(I - P_X)Y\right\}\\
    &\le \frac{2(M-\sigma^4)}{n-r},
\end{align*}
which converges to 0 because $n\rightarrow\infty$ and $r/n\rightarrow \tau\in[0,1)$. The last inequality is due to Lemma \ref{quadformvariance}, with the fact that $I-P_X$ is an idempotent matrix with rank $n-r$.

So $B_n-b_n\rightarrow_{q.m}0$, which implies that $B_n/b_n\rightarrow_{q.m.} 1$ because $b_n=\sigma^2$ is constant. 

Next, knowing that $\mathbb{E}(A_n)=a_n$, similarly,
\begin{align}
    \mathbb{E}\Big(\frac{A_n}{a_n}-1\Big)&=0 \nonumber\\
    {\rm var}\Big(\frac{A_n}{a_n}-1\Big)&=\frac{1}{a_n^2}{\rm var}(A_n)\nonumber\\
    &=\frac{1}{a_n^2 (r-1)^2}{\rm var}\left\{\|(P_X - P_1)Y\|^2\right\}\nonumber\\
    &=\frac{1}{a_n^2 (r-1)^2}{\rm var}\left\{Y'(P_X - P_1)Y\right\}\nonumber\\
    &\le \frac{2(M-\sigma^4)(r-1)+8\delta^2\sigma^2}{(\sigma^2(r-1)+\delta^2)^2} \label{quadpropagain}\\
    &= \frac{2(M-\sigma^4)(r-1)+8\delta^2\sigma^2}{\sigma^4(r-1)^2+\delta^4+2\sigma^2\delta^2(r-1)}, \label{lasttermvaran}
\end{align}
where we used Lemma \ref{quadformvariance} for \eqref{quadpropagain} with the fact that $P_X-P_1$ is an idempotent matrix with rank $r-1$. If $r \rightarrow \infty$ holds, since \eqref{lasttermvaran} is bounded above by $4(r-1)^{-1}+{2(M-\sigma^4)}({\sigma^4(r-1)})^{-1}$, it converges to $0$. 

On the other hand, under the second condition, we have the same convergence result. This is because we can have an upper bound of \eqref{lasttermvaran} as $({M-\sigma^4})({2\sigma^2\delta^2})^{-1}+{8\sigma^2}({\delta^2})^{-1}$, it converges to $0$ with $\delta^2\rightarrow\infty$.

Combining two convergence results, $B_n/b_n\rightarrow_{q.m}1,A_n/a_n\rightarrow_{q.m}1$, it implies $B_n/b_n\rightarrow_{p}1,A_n/a_n\rightarrow_{p}1$, and further yields $F_n/f_n=\frac{A_n/a_n}{B_n/b_n}\rightarrow_p 1$ by Slutsky's theorem. 

Additionally, let us denote that $F^0_n={\rm max}(1,F_n)$, then we have $\hat\gamma=1-{1}/{F^0_n}$. Finally, we get 
\begin{align*}
    \mathbb{P}(|\hat\gamma-\gamma_{\rm opt}|\ge \epsilon)&=
    \mathbb{P}\Big(\Big|1-\frac{1}{F^0_n}-1+\frac{1}{f_n}\Big|\ge \epsilon\Big)\\
    &=\mathbb{P}\Big(\Big|\frac{1}{F^0_n}-\frac{1}{f_n}\Big|\ge \epsilon\Big)\\
    &=\mathbb{P}\Big(\frac{|F^0_n-f_n|}{F^0_n f_n}\ge\epsilon\Big)\\
    &\le\mathbb{P}\Big(\frac{|F_n-f_n|}{F_nf_n}\ge\epsilon,F_n\ge 1\Big)+\mathbb{P}\Big({|1-f_n|}\ge\epsilon,F_n< 1\Big)\\
    &\le\mathbb{P}\Big(\frac{|F_n-f_n|}{f_n}\ge\epsilon,F_n\ge 1\Big)+\mathbb{P}\Big(f_n\ge 1+\epsilon,F_n< 1\Big)\\
    &\le \mathbb{P}\Big(\Big|\frac{F_n}{f_n}-1\Big|\ge\epsilon\Big)+\mathbb{P}\Big(\Big|\frac{F_n}{f_n}-1\Big|\ge\frac{\epsilon}{1+\epsilon}\Big) \rightarrow 0.
\end{align*}    
\end{proof}

\subsection{Proof of Proposition \ref{convergencerate}}\label{appendthm3}
\begin{lemma}\label{lem1}
If a random variable $X$ follows a chi-squared distribution with a degree of freedom $K\ge 5$, and a non-central parameter $\lambda>0$ which is denoted as $\chi^2(K;\lambda)$, the followings holds.
\begin{align}\label{inversechisq}
&\mathbb{E}(X^{-1})\le {\rm min}\Big\{\frac{1}{K-2},\frac{1}{2\sqrt{\lambda(K-2)}}\Big\}\\
    &\mathbb{E}(X^{-2})\le {\rm min}\Big\{\frac{1}{2(K-4)},\frac{(K-4+\lambda)/(2\sqrt{\lambda(K-4)})-1}{2\lambda}\Big\}
\end{align}
\end{lemma}
\begin{proof}
  \citet{xie1988simple} verified the following three results for $K\ge 5$.
  \begin{align*}
      &\mathbb{E}^1_K=\int_0^1 s^{K-3}e^{\frac{\lambda(s^2-1)}{2}}ds,\\
      &\mathbb{E}^1_K=1/\lambda-((K-4)/\lambda)\mathbb{E}^1_{K-2},\\
      &\mathbb{E}^n_K-\mathbb{E}^n_{K+2}=2n \mathbb{E}^{n+1}_{K+2},
  \end{align*}
where $\mathbb{E}^n_K=\mathbb{E}(1/(\chi^2(K;\lambda))^n)$. Hence
\begin{align*}
    \mathbb{E}(X^{-2})=\mathbb{E}^2_K&=\frac{\mathbb{E}^1_{K-2}-\mathbb{E}^1_{K}}{2}\\
    &=\frac{(K-4+\lambda)\mathbb{E}^1_{K-2}-1}{2\lambda}
\end{align*}
holds. Furthermore,
\begin{align*}
 \mathbb{E}(X^{-1})=\mathbb{E}^1_K&= \int_0^1 s^{K-3}e^{\frac{\lambda(s^2-1)}{2}}ds\\
 &\le \int_0^1 s^{K-3} ds=\frac{1}{K-2}.
\end{align*}
On the other hand, from the same equation,
\begin{align*}
    &\int_0^1 s^{K-3}e^{\frac{\lambda(s^2-1)}{2}}ds\\
    &\le e^{-\lambda/2} \Big(\int_0^1 s^{2K-5}ds\Big)^{1/2} \Big(\int_0^1 s e^{\lambda s^2}ds\Big)^{1/2}\\
    & =e^{-\lambda/2} \frac{1}{\sqrt{2K-4}}\frac{\sqrt{e^{\lambda}-1}}{\sqrt{2\lambda}}\\
    &\le \frac{1}{2\sqrt{\lambda(K-2)}}.
\end{align*}
This completes the proof.
\end{proof}

With Lemma \ref{lem1}, we prove the statement of Proposition \ref{convergencerate}.
\begin{proof}
From the definition of $\hat\gamma$ \eqref{mainestimator}, we obtain
\begin{align}
    |\hat\gamma-\gamma_{\rm opt}|&=\Big|I(F_n \ge 1)(1-1/F_n)-(1-1/f_n)\Big|\nonumber\\
    &\le \Big|\frac{1}{f_n}-\frac{1}{F_n}\Big|,\label{diffbyF}
\end{align}
where $\gamma_{\rm opt}=1-1/f_n$, $f_n= \frac{\sigma^2(r-1)+\delta^2}{\sigma^2(r-1)}=1+\frac{\delta^2}{\sigma^2(r-1)}\in [1,\infty)$.

Let us follow the same notations in Proposition \ref{consistency} as $F_n={A_n}/{B_n},f_n={a_n}/{b_n}$, where $A_n=\|P_X Y - P_1 Y \|^2/(r-1),B_n=\|Y - P_X Y \|^2/(n-r),a_n=\sigma^2+{\delta^2}/({r-1}),b_n=\sigma^2$. Since we have normal errors, $A_n,B_n$ are independent. In addition, $(r-1)A_n/\sigma^2\sim \chi^2(r-1;\delta^2/ \sigma^2)$, $(n-r)B_n/\sigma^2\sim \chi^2(n-r)$, and we denote $A'_n:=A_n/\sigma^2,B'_n:=B_n/\sigma^2$. First, we consider the first case where $\delta^2=O(r)$. Since $F_n$ follows non-central $F$ distribution,
\begin{align*}
    \mathbb{E}(F_n)&=\frac{(n-r)(r-1+\delta^2/\sigma^2)}{(r-1)(n-r-2)}\\
    &=\Big(\frac{2}{n-r-2}+1\Big)f_n,
\end{align*}
when $n-r>2$. And, with $n-r>4$,
\begin{align}
    {\rm var}(F_n)=2\frac{(n-r-2)(n-r)^2}{(n-r-4)(n-r-2)^2}\Big[\frac{(r-1+\delta^2/\sigma^2)^2}{(n-r-2)(r-1)^2} + \frac{r-1+2\delta^2/\sigma^2}{(r-1)^2}\Big].\label{varF}
\end{align}
Then,
\begin{align*}
 \mathbb{E}(r(F_n-f_n)^2)&=2r\frac{(n-r-2)(n-r)^2}{(n-r-4)(n-r-2)^2}\left[\frac{(r-1+\delta^2/\sigma^2)^2}{(n-r-2)(r-1)^2} \right.\\
 &\qquad+\left. \frac{r-1+2\delta^2/\sigma^2}{(r-1)^2}\right]+\frac{4r}{(n-r-2)^2}f_n^2,
\end{align*}
which asymptotically bounded by the assumptions. Furthermore, since we know that for arbitrary small $\epsilon > 0$,
\begin{align}
  \hat\gamma-\gamma_{\rm opt}=-\gamma_{\rm opt}I(F_n\le 1-\epsilon)+I(F_n\ge 1-\epsilon)(\hat\gamma-\gamma_{\rm opt})\label{decomp}
\end{align}
holds, and the second term is bounded by \eqref{diffbyF}. Indeed, on the event $F_n\ge 1-\epsilon$, by Cauchy-Schwarz inequality and the fact that $f_n\rightarrow f_{\infty}\in[1,\infty)$, 
\begin{align*}
\Big(\mathbb{E}|1/f_n-1/F_n|\Big)^2\le  \frac{1}{f_n^2}\mathbb{E}|F_n-f_n|^2\mathbb{E}\Big(\frac{1}{F_n^2}\Big)&\le \frac{1}{(1-\epsilon)^2f_n^2}\mathbb{E}|F_n-f_n|^2\mathbb{E}\Big(\frac{1}{F_n^2}\Big)\\
&\asymp \mathbb{E}|F_n-f_n|^2.
\end{align*}
On the other hand, we show that the first term in \eqref{decomp} is negligible. Indeed, $P(F_n\le 1-\epsilon)$ is at its largest when $f_n\rightarrow 1$, since $F_n/f_n\rightarrow 1$. Hence, it suffices to show that the tail bound of $P(F_n\le 1-\epsilon)$ in the case of $f_n\rightarrow 1$ is negligible. Indeed, 
\begin{align}
    P(F_n\le 1-\epsilon)&=P\Big(\frac{A'_n}{B'_n}\le 1-\epsilon,A'_n>1-\frac{\epsilon}{2}\Big)+P\Big(\frac{A'_n}{B'_n}\le 1-\epsilon,A'_n\le 1-\frac{\epsilon}{2}\Big)\nonumber\\
    &=P\Big(B'_n\ge \frac{1-\frac{\epsilon}{2}}{1-\epsilon}\Big)+P\Big(A'_n\le 1-\frac{\epsilon}{2}\Big).\label{Fnterms}
\end{align}
Lemma 1 from \citet{laurent2000adaptive} with $x=nt/10$ yields $P(Z\ge 2nt)\le {\rm exp}({-nt/10})$ for all $t\ge 1$, where a random variable $Z$ follows $\chi^2(n)$. This further implies that the first term in \eqref{Fnterms} is controlled by 
\begin{align*}
{\rm exp}({-(n-r)(1-\frac{\epsilon}{2})/(20(1-\epsilon))}),
\end{align*}
which has exponential decay in $n$. Furthermore, Theorem 7 of \citet{zhang2020non} yields that the second term in \eqref{Fnterms} is bounded above by 
\begin{align*}
{\rm exp}({-\epsilon^2(r-1)^2/(16(r-1+\delta^2/\sigma^2)})).
\end{align*}
Based on the case that $\delta^2=o(r)$, the upper bound has exponential decay in $r$.
Consequently, $\mathbb{E}|1/f_n-1/F_n|=O(r^{-1/2})$ for $\delta^2=O(r)$.

However, when $r=o(\delta^2)$,  $f_n\rightarrow\infty$, hence the same reasoning cannot be applied. We use Lemma \ref{lem1} instead. Indeed,
\begin{align*}
    \mathbb{E}|1/f_n-1/F_n|&=\mathbb{E}\Big|\frac{b_n A_n-B_n a_n}{A_n a_n}\Big|\\
    &\le\mathbb{E}\frac{b_n}{A_n}\Big|\frac{A_n}{a_n}-1\Big|+\mathbb{E}\frac{1}{A_n}|B_n-b_n|\\
    &\le \sigma^2\Big[\mathbb{E}\Big(\frac{A_n}{a_n}-1\Big)^2\mathbb{E}\Big(\frac{1}{A_n}\Big)^2\Big]^{1/2}+\Big[\mathbb{E}(B_n-b_n)^2\Big]^{1/2}\mathbb{E}\Big(\frac{1}{A_n}\Big),
\end{align*}
by Cauchy-Schwarz inequality and independence of $A_n,B_n$. Since 
\begin{align*}
\mathbb{E}\Big(\frac{1}{A_n}\Big)=\frac{(r-1)}{\sigma^2}\mathbb{E}(1/\chi^2(r-1;\delta^2))&\le \frac{(r-1)}{\sigma^2}{\rm min}\Big\{\frac{1}{r-3},\frac{1}{2\sqrt{\delta^2(r-3)}}\Big\}\\
&:=c_1(r,\delta^2)/\sigma^2,
\end{align*}
and, 
\begin{align*}
\mathbb{E}\Big(\frac{1}{A_n}\Big)^2&=\frac{(r-1)^2}{\sigma^4}\mathbb{E}(1/\chi^2(r-1;\delta^2))^2\\
&\le \frac{(r-1)^2}{\sigma^4}{\rm min}\Big\{\frac{1}{2(r-5)},\frac{(r-5+\delta^2)/(2\sqrt{\delta^2(r-5)})-1}{2\delta^2}\Big\}\\
&:=c_2(r,\delta^2)/\sigma^4    
\end{align*}
from Lemma \ref{lem1}, and combining results from the Proposition \ref{consistency},
\begin{align}
   &\mathbb{E}|1/f_n-1/F_n|\nonumber\\ 
   &\le \frac{\sqrt{4\sigma^4(r-1)+8\delta^2\sigma^2}}{\sigma^2(r-1)+\delta^2}\sqrt{c_2(r,\delta^2)}+\frac{\sqrt{2}}{\sqrt{n-r}}c_1(r,\delta^2)\nonumber\\
   &\le \frac{2\sqrt{2}\sigma\sqrt{c_2(r,\delta^2)}}{\sqrt{\sigma^2(r-1)+\delta^2}}+\frac{\sqrt{2}}{\sqrt{n-r}}c_1(r,\delta^2)\label{convratelastineq}
\end{align}
If $\delta^2/r=o(1)$ holds with growing $r$, the first term in \eqref{convratelastineq} is of order $O_P((\delta^2/r)^{-3/4})$, and the second term is of $O_P(n^{-1/2}(\delta^2/r)^{-1/2})$. Similarly we can have analogous rates for the case where $r$ is not divergent, and this completes the proof.
\end{proof}

\subsection{Proof of Proposition \ref{highdimconsistency}}\label{appendthm4}
\begin{proof}
    We denote $C_n=\|Y-P_1Y\|^2,c_n=\delta^2+ \sigma^2(n-1)$. Referring to Lemma \ref{quadformvariance} with an rank $n-1$ idempotent matrix $A=I-P_1$, we have
    \begin{align*}
       \frac{ {\rm var}C_n}{c_n^2}\le \frac{2(M-\sigma^4)(n-1)+8\sigma^2\delta^2}{(\delta^2+\sigma^2(n-1))^2}\rightarrow 0.
    \end{align*}
Thus, with $\mathbb{E}C_n=c_n$, we obtain $C_n/c_n\rightarrow_P 1$. Now,
\begin{align}
  \frac{Y'(I-P_1-\frac{n-1}{nC}(I-K))Y}{Y'(I-P_1)Y}-\frac{\delta^2}{\delta^2+\sigma^2(n-1)} &=\frac{C_n-(n-1)\check\sigma^2_c}{C_n}-\frac{c_n-\sigma^2(n-1)}{c_n}\nonumber\\
  &=\frac{(n-1)C_n\sigma^2-(n-1)c_n\check\sigma^2_c}{C_nc_n}\nonumber\\
  &=\frac{(n-1)\sigma^2(C_n-c_n)-(n-1)c_n(\sigma^2-\check\sigma^2_c)}{C_nc_n}.\label{thm4term1}
\end{align}
Since $C_n/c_n\rightarrow_P 1$, ${(n-1)\sigma^2}/{c_n}\le 1$, the first term in the \eqref{thm4term1} converges to 0 in probability. Next, we prove the consistency of $\check\sigma^2_c$. Indeed, since $K$ and $P_1$ are simultaneously diagonalizable and $KP_1=P_1K=P_1$,
\begin{align*}
    \mathbb{E}\frac{1}{nC}Y'(I-K)Y&=\frac{1}{nC}(\sigma^2{\rm tr}(I-K)+(\mu-P_1\mu)'(I-K)(\mu-P_1\mu))\\
    &=\sigma^2+\frac{1}{nC}(\mu-P_1\mu)'(I-K)(\mu-P_1\mu).
\end{align*}
The eigenvalues of $I-K$ are $\frac{2n\alpha}{d_j+2n\alpha}$, for $j=2,\dots,n$, and thus, $nC={\rm tr}(I-K)=\sum_{j=2}^n\frac{2n\alpha}{d_j+2n\alpha}$. Hence, 
\begin{align*}
    \frac{1}{nC}(\mu-P_1\mu)'(I-K)(\mu-P_1\mu)&\le
    \frac{\frac{1}{d_n+2n\alpha}}{\sum_{j=2}^n \frac{1}{d_j+2n\alpha}}\delta^2\\
    &\le \frac{\frac{1}{d_n+2n\alpha}}{\frac{n-1}{d_2+2n\alpha}}\delta^2=\frac{d_2+2n\alpha}{d_n+2n\alpha}\frac{\delta^2}{n-1}\rightarrow 0,
\end{align*}
since the assumptions $d_2(n\alpha)^{-1}=o_P(1)$, and $\delta^2=o(n)$, forces $\lim_{n\rightarrow\infty} \frac{d_2+2n\alpha}{d_n+2n\alpha}= 1$. This yields $\mathbb{E}\check\sigma^2_c-\sigma^2\rightarrow 0$. In addition, by Lemma \ref{lem1}, letting $\phi_j=\frac{1}{d_j+2n\alpha}$,
\begin{align*}
    {\rm var}\Big(\frac{1}{nC}Y'(I-K)Y\Big)&\le \frac{4\sigma^4}{n^2C^2}{\rm tr}(I-K)^2+\frac{8\sigma^2}{n^2C^2}(\mu-P_1\mu)'(I-K)^2(\mu-P_1\mu)\\
    &\le 4\sigma^4\frac{\phi_2^2+ \dots+\phi_n^2}{(\phi_2+\dots+\phi_n)^2}+ 8\sigma^2\delta^2\frac{\phi_n^2}{(\phi_2+\dots+\phi_n)^2}\\
    &\le 4\sigma^4\frac{\phi_n^2}{n\phi_2^2}+ 8\sigma^2\delta^2\frac{\phi_n^2}{n^2\phi_2^2}\rightarrow 0.
\end{align*}
Consequently, $\check\sigma^2_c-\sigma^2\rightarrow_P 0$, and this further yields that the second term in \eqref{thm4term1} converges to 0 in probability. Finally, due to $\gamma_{\rm opt}\ge 0$, $\Big|\frac{Y'(I-P_1-\frac{n-1}{nC}(I-K))Y}{Y'(I-P_1)Y}-\gamma_{\rm opt}\Big|\ge |\Bar{\gamma}_c-\gamma_{\rm opt}|$ holds. Thus, this concludes that $\Bar{\gamma}_c-\gamma_{\rm opt}\rightarrow_P 0$.
\end{proof}

\subsection{Proof of Proposition \ref{highdimconsistency:2}}\label{appendthm:highdim2}
\begin{proof}
Similarly to the proof of Proposition \ref{highdimconsistency}, it suffices to show 
\begin{align*}
    \frac{(n-1)(\check\sigma^2-\sigma^2)}{c_n}\rightarrow_P 0.
\end{align*}
Since, analogous to the proof of Proposition \ref{highdimconsistency}, we have
\begin{align*}
    \mathbb{E}n^{-1}Y'(I-K)Y&=\frac{1}{n}(\sigma^2{\rm tr}(I-K)+(\mu-P_1\mu)'(I-K)(\mu-P_1\mu))\\
    &\le \frac{\sigma^2}{n}\sum_{i=2}^n\frac{2n\alpha}{d_i+2n\alpha}+\frac{\delta^2}{n}\frac{2n\alpha}{d_n+2n\alpha}.
\end{align*}
This yields,
\begin{align*}
\mathbb{E}n^{-1}Y'(I-K)Y-\sigma^2 &\le \frac{\sigma^2}{n}+\frac{\sigma^2}{n}\sum_{i=2}^n\frac{d_i}{d_i+2n\alpha}+\frac{\delta^2}{n}\frac{2n\alpha}{d_n+2n\alpha}\\
&\le \frac{\sigma^2}{n}+\frac{(n-1)\sigma^2}{n}\frac{d_2}{d_2+2n\alpha}+\frac{\delta^2}{n}\frac{2n\alpha}{d_n+2n\alpha}.
\end{align*}
Hence, up to a constant,
\begin{align}\label{highdimtheorem2-term1}
\frac{n-1}{c_n}\Big(\mathbb{E}n^{-1}Y'(I-K)Y-\sigma^2\Big) \lesssim \frac{n-1}{\sigma^2(n-1)+\delta^2} \Big(\sigma^2\frac{d_2}{d_2+2n\alpha}+\frac{\delta^2}{n}\frac{2n\alpha}{d_n+2n\alpha}\Big).  
\end{align} 
Under the first condition of $\delta^2=o(n)$ and $d_2(n\alpha)^{-1}=o_P(1)$, right hand side of \eqref{highdimtheorem2-term1} goes to $0$. On the other hand, under $n=\delta^2$ and $n\alpha(d_n)^{-1}=o_P(1)$, \eqref{highdimtheorem2-term1} also converges to $0$ as well.
In addition, analogously to Proposition \ref{highdimconsistency}, recalling that $c_n=(n-1)\sigma^2+\delta^2$, up to a constant, we have
\begin{align}
&\frac{(n-1)^2}{c_n^2}{\rm var}\Big(\frac{1}{n}Y'(I-K)Y\Big)\\
&\qquad\le \frac{(n-1)^2}{c_n^2}\Big( \frac{4\sigma^4}{n^2}{\rm tr}(I-K)^2+\frac{8\sigma^2}{n^2}(\mu-P_1\mu)'(I-K)^2(\mu-P_1\mu)\Big)\nonumber\\
&\qquad\lesssim \frac{4n\sigma^4}{c_n^2}\Big(\frac{2n\alpha}{d_n+2n\alpha}\Big)^2+ \frac{8\sigma^2\delta^2}{c_n^2}\Big(\frac{2n\alpha}{d_n+2n\alpha}\Big)^2\label{highdimtheorem2-term2}.
\end{align}
One can easily check that \eqref{highdimtheorem2-term2} converges to $0$ if any of two conditions holds. 
Consequently, $\check\sigma^2_c-\sigma^2\rightarrow_P 0$, and this completes the proof.
\end{proof}

\section{Additional plots and tables for numerical experiments}

\subsection{Additional plots for simulations}\label{addedboxplots}
Additional supplement plots are in Figure \ref{fig:lowdimadd}--\ref{fig:highdimgr}.

\begin{figure}
     \centering
     \begin{subfigure}[b]{0.45\textwidth}
         \centering
         \includegraphics[width=\textwidth]{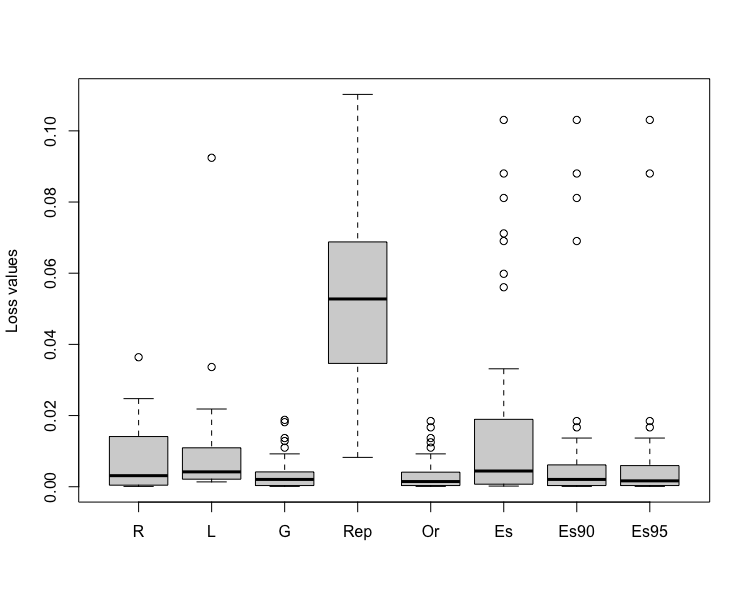}
         \caption{$\tau=10^{-2}$}
         \label{fig:2503}
     \end{subfigure}
     \begin{subfigure}[b]{0.45\textwidth}
         \centering
         \includegraphics[width=\textwidth]{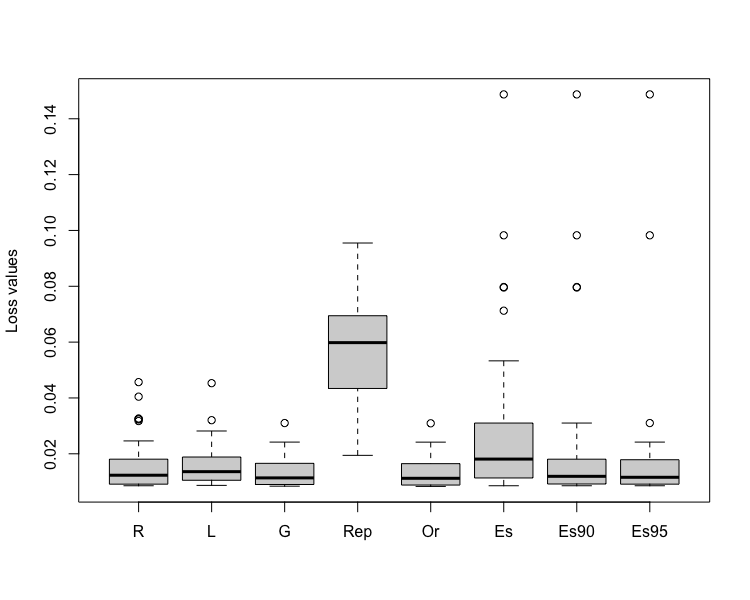}
         \caption{$\tau=10^{-1}$}
         \label{fig:2504}
     \end{subfigure}\\
     \begin{subfigure}[b]{0.45\textwidth}
         \centering
         \includegraphics[width=\textwidth]{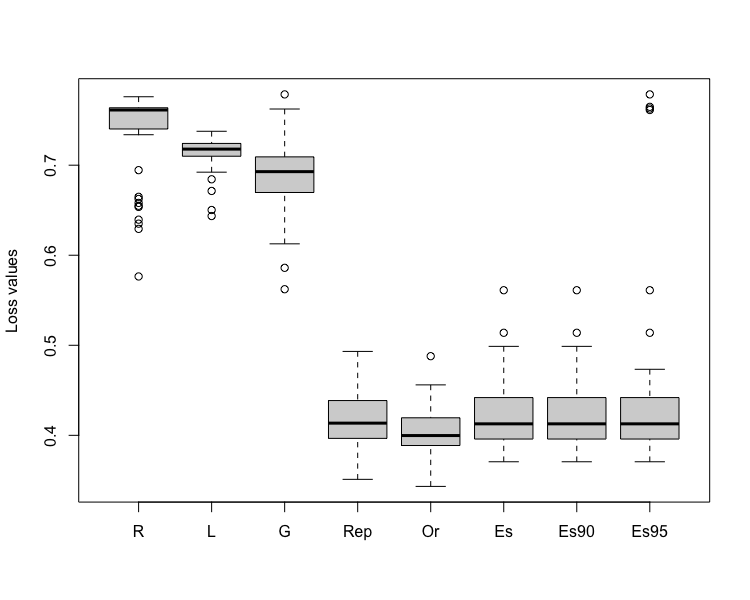}
         \caption{$\tau=1$}
         \label{fig:2506}
     \end{subfigure}
     \begin{subfigure}[b]{0.45\textwidth}
         \centering
         \includegraphics[width=\textwidth]{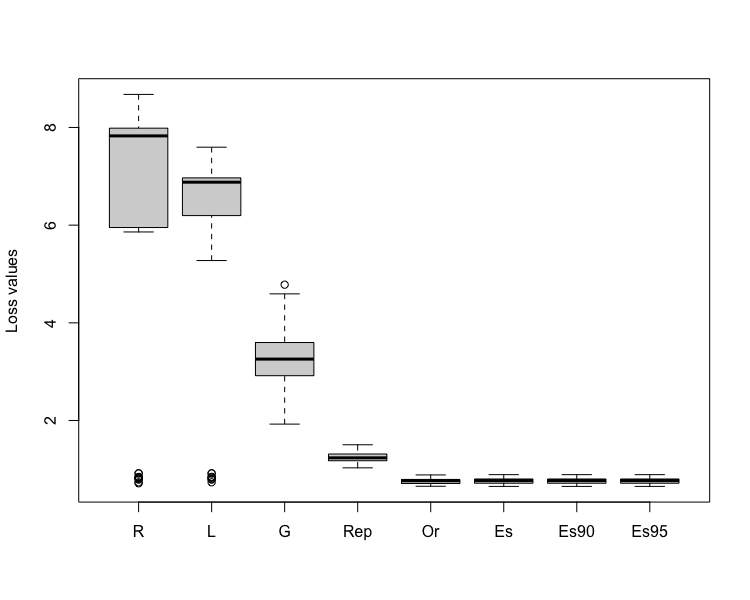}
         \caption{$\tau=10^{0.5}$}
         \label{fig:2507}
     \end{subfigure}\\
 \caption{
 Boxplots of the observed same-X loss values from the 50 replications when $n=300$ and $p=250$. See the caption of Figure \ref{fig:lowdim1} for more details.}
\label{fig:lowdimadd}
\end{figure}

\begin{figure}
     \centering
     \begin{subfigure}[b]{0.45\textwidth}
         \centering
         \includegraphics[width=\textwidth]{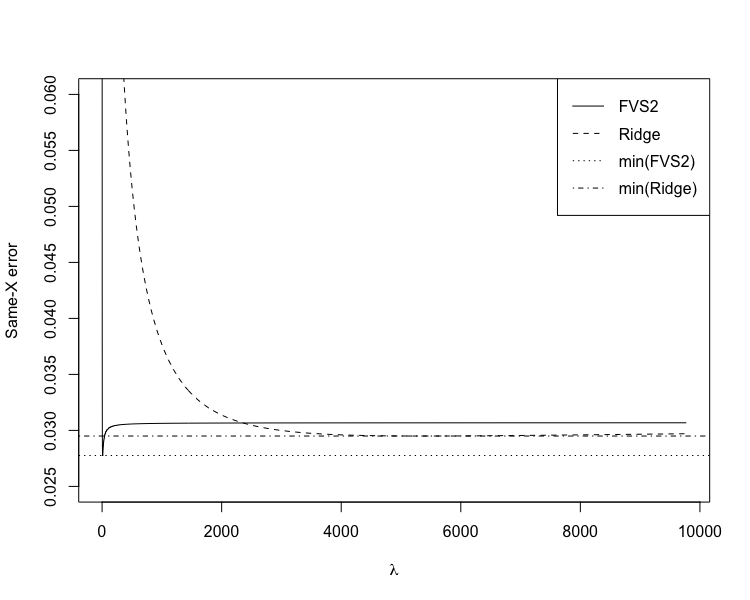}
         \caption{$\tau=10^{-0.5}$}
         \label{fig:gr755}
     \end{subfigure}
     \begin{subfigure}[b]{0.45\textwidth}
         \centering
         \includegraphics[width=\textwidth]{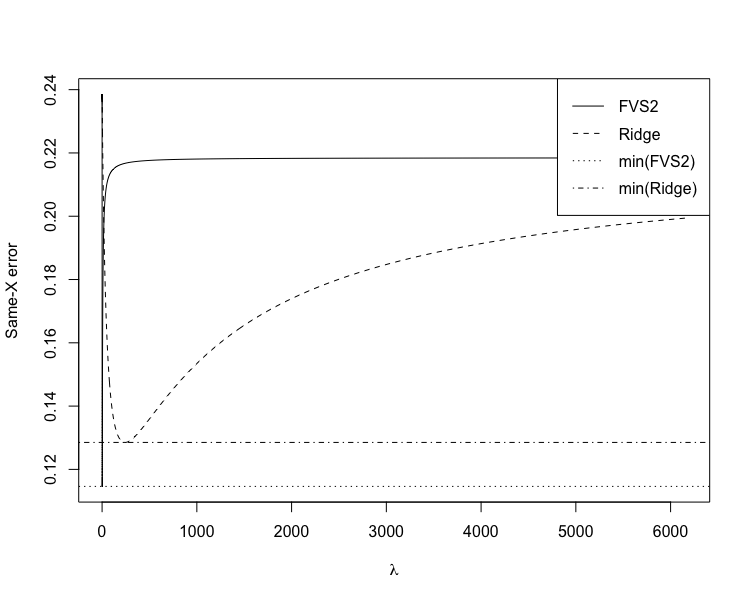}
         \caption{$\tau=1$}
         \label{fig:gr756}
     \end{subfigure}\\
       \begin{subfigure}[b]{0.45\textwidth}
         \centering
         \includegraphics[width=\textwidth]{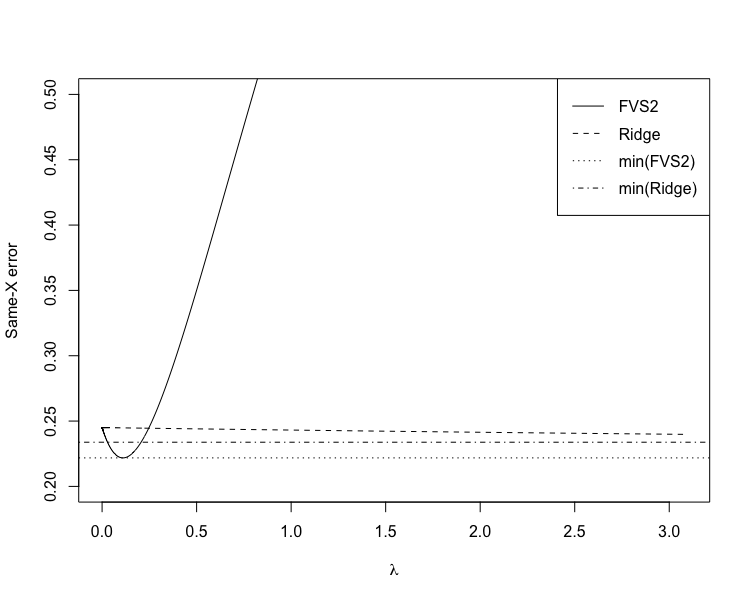}
         \caption{$\tau=10^{0.5}$, magnified}
         \label{fig:gr757mag}
     \end{subfigure}
     \begin{subfigure}[b]{0.45\textwidth}
         \centering
         \includegraphics[width=\textwidth]{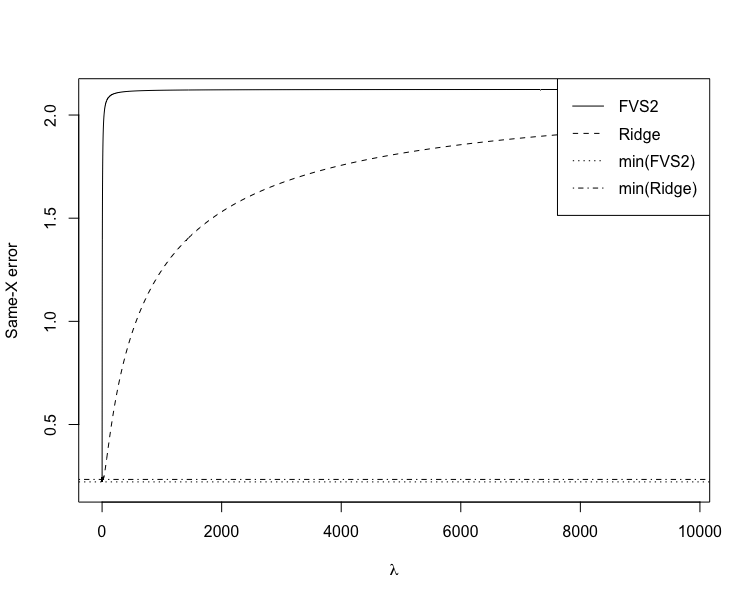}
         \caption{$\tau=10^{0.5}$}
         \label{fig:gr757}
     \end{subfigure}\\
 \caption{Additional displays to Section \ref{lowsimulsection} when $(n,p)=(300,75)$. Graphs of the average same-X loss values over the 50 replications as a function of $\lambda$ comparing our proposed method \eqref{mainoptlam} to Ridge \eqref{ridgeform}.
 The label FVS2 is fitted value shrinkage with $L_2$-squared penalty \eqref{mainoptlam}, and the horizontal lines, min(FVS2) and min(Ridge), are the minimal average same-X loss values obtained from each estimator within the set of $\lambda$.
 }
\label{fig:lowdimgr1}
\end{figure}

\begin{figure}
     \centering
     \begin{subfigure}[b]{0.45\textwidth}
         \centering
         \includegraphics[width=\textwidth]{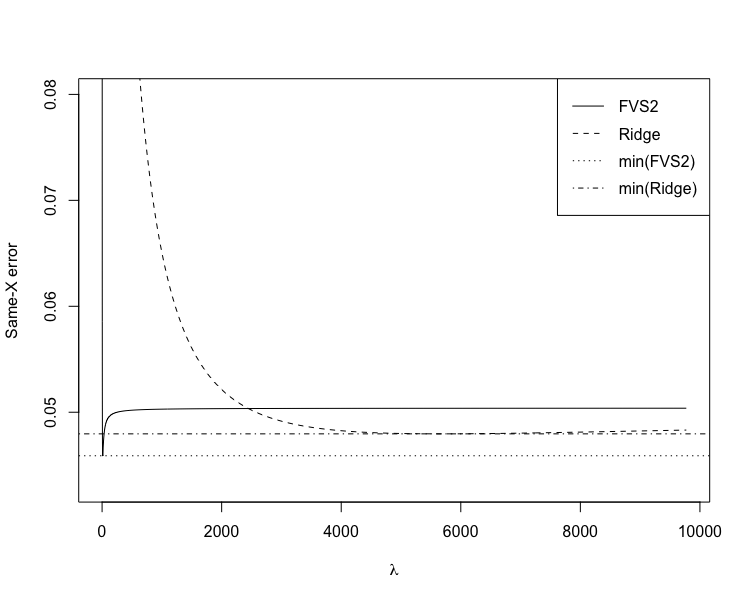}
         \caption{$\tau=10^{-0.5}$}
         \label{fig:gr1505}
     \end{subfigure}
     \begin{subfigure}[b]{0.45\textwidth}
         \centering
         \includegraphics[width=\textwidth]{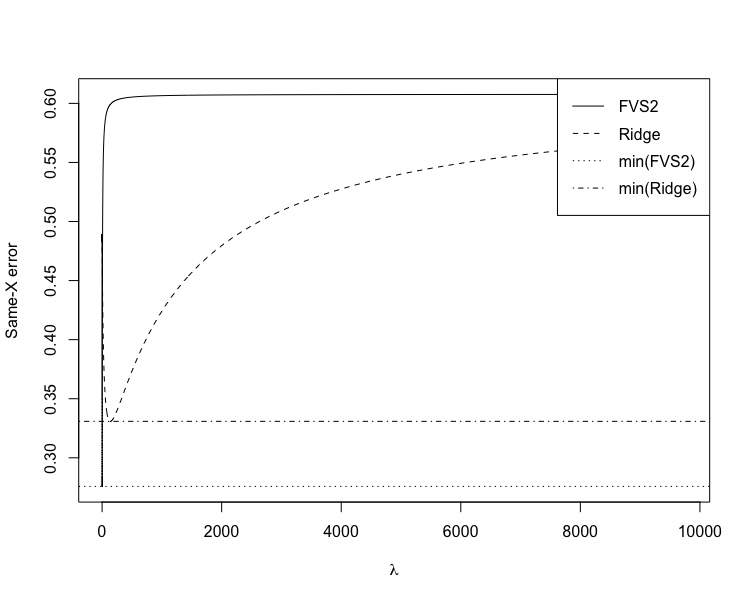}
         \caption{$\tau=1$}
         \label{fig:gr1506}
     \end{subfigure}\\
       \begin{subfigure}[b]{0.45\textwidth}
         \centering
         \includegraphics[width=\textwidth]{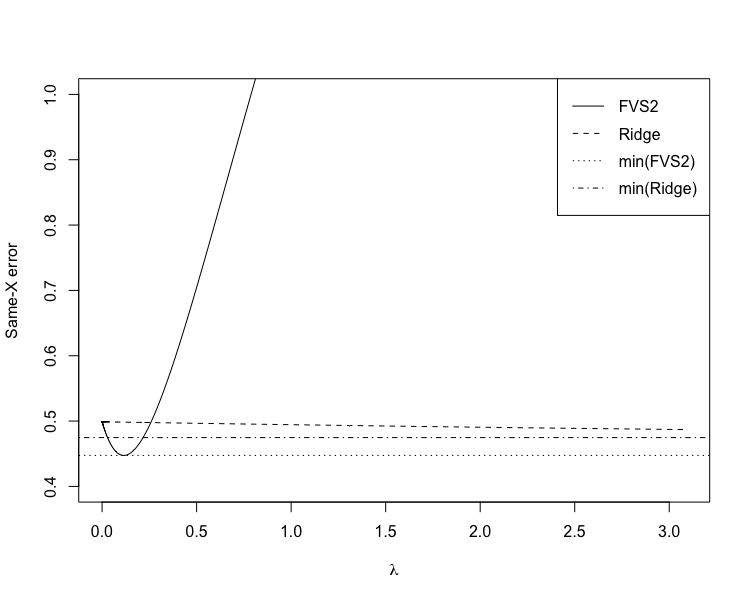}
         \caption{$\tau=10^{0.5}$, magnified}
         \label{fig:gr1507mag}
     \end{subfigure}
     \begin{subfigure}[b]{0.45\textwidth}
         \centering
         \includegraphics[width=\textwidth]{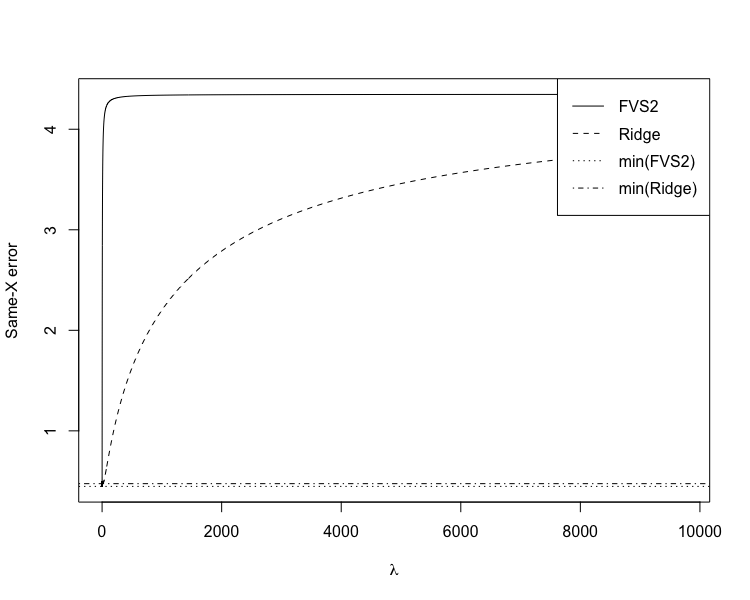}
         \caption{$\tau=10^{0.5}$}
         \label{fig:gr1507}
     \end{subfigure}\\
 \caption{Additional displays to Section \ref{lowsimulsection} when $(n,p)=(300,150)$. Graphs of the average same-X loss values over the 50 replications as a function of $\lambda$ comparing our proposed method \eqref{mainoptlam} to Ridge \eqref{ridgeform}.
 The label FVS2 is fitted value shrinkage with $L_2$-squared penalty \eqref{mainoptlam}, and the horizontal lines, min(FVS2) and min(Ridge), are the minimal average same-X loss values obtained from each estimator within the set of $\lambda$.
 }
\label{fig:lowdimgr2}
\end{figure}

\begin{figure}
     \centering
     \begin{subfigure}[b]{0.45\textwidth}
         \centering
         \includegraphics[width=\textwidth]{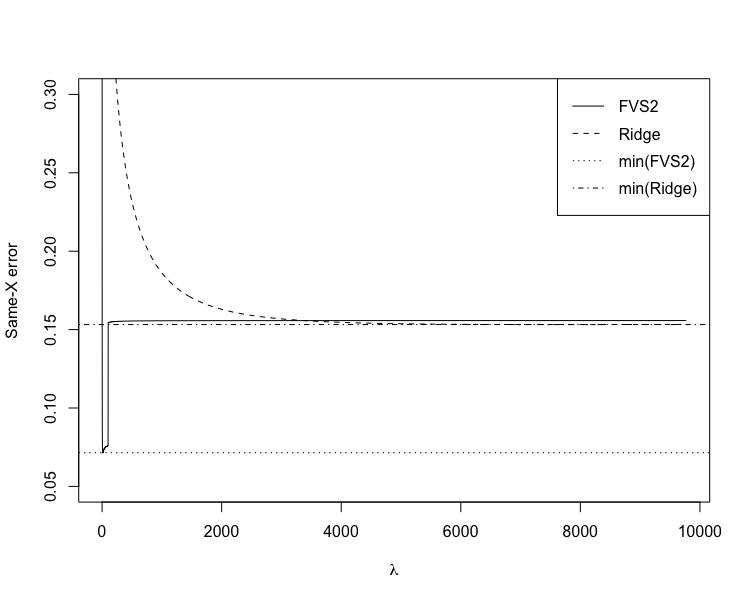}
         \caption{$\tau=10^{-0.5}$}
         \label{fig:gr2505}
     \end{subfigure}
     \begin{subfigure}[b]{0.45\textwidth}
         \centering
         \includegraphics[width=\textwidth]{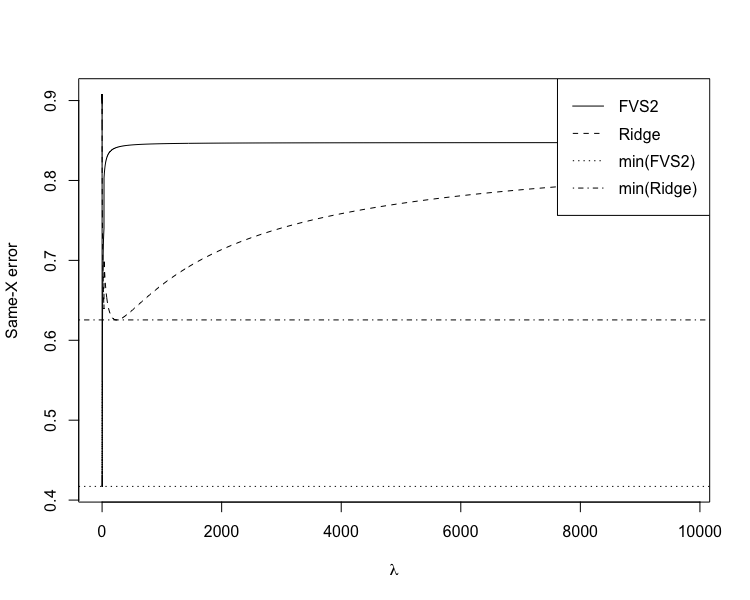}
         \caption{$\tau=1$}
         \label{fig:gr2506}
     \end{subfigure}\\
       \begin{subfigure}[b]{0.45\textwidth}
         \centering
         \includegraphics[width=\textwidth]{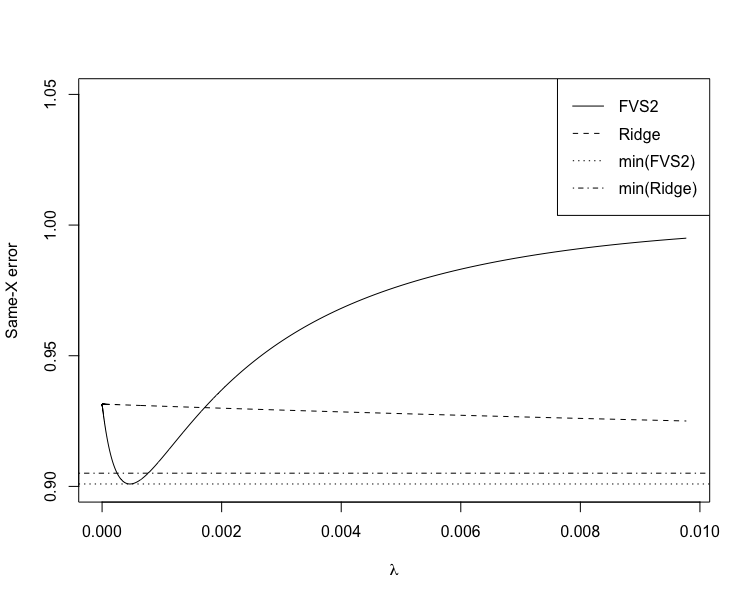}
         \caption{$\tau=10^{0.5}$, magnified}
         \label{fig:gr2507mag}
     \end{subfigure}
     \begin{subfigure}[b]{0.45\textwidth}
         \centering
         \includegraphics[width=\textwidth]{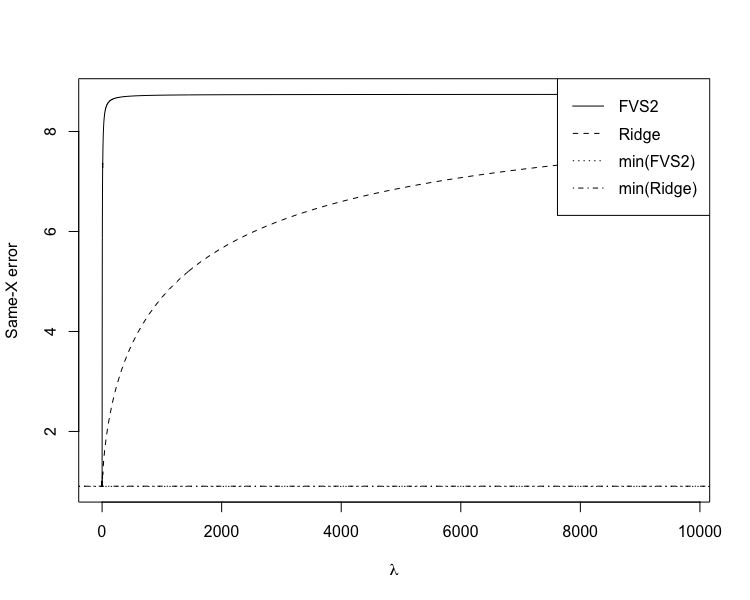}
         \caption{$\tau=10^{0.5}$}
         \label{fig:gr2507}
     \end{subfigure}\\
 \caption{Additional displays to Section \ref{lowsimulsection} when $(n,p)=(300,250)$. Graphs of the average same-X loss values over the 50 replications as a function of $\lambda$ comparing our proposed method \eqref{mainoptlam} to Ridge \eqref{ridgeform}.
 The label FVS2 is fitted value shrinkage with $L_2$-squared penalty \eqref{mainoptlam}, and the horizontal lines, min(FVS2) and min(Ridge), are the minimal average same-X loss values obtained from each estimator within the set of $\lambda$.
 }
\label{fig:lowdimgr3}
\end{figure}

\begin{figure}
     \centering
     \begin{subfigure}[b]{0.40\textwidth}
         \centering
         \includegraphics[width=\textwidth]{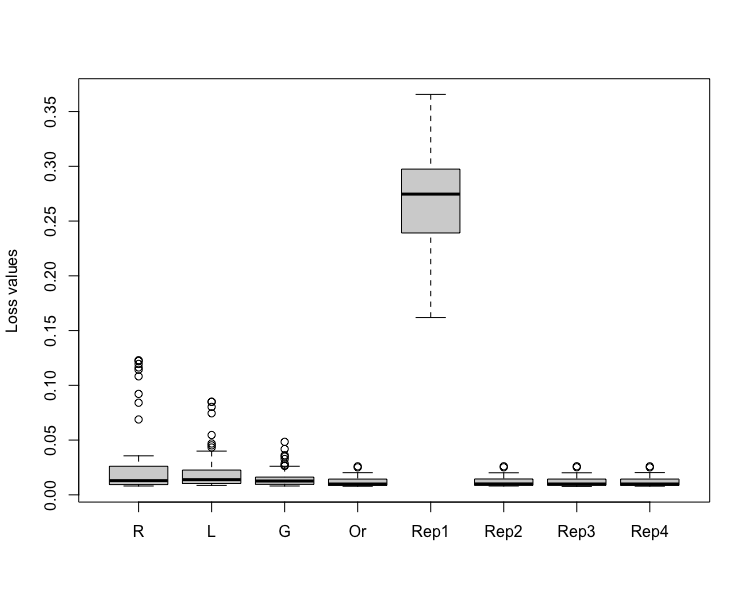}
         \caption{$(1,10^{-1})$}
         \label{fig:hi4}
     \end{subfigure}
     \begin{subfigure}[b]{0.40\textwidth}
         \centering
         \includegraphics[width=\textwidth]{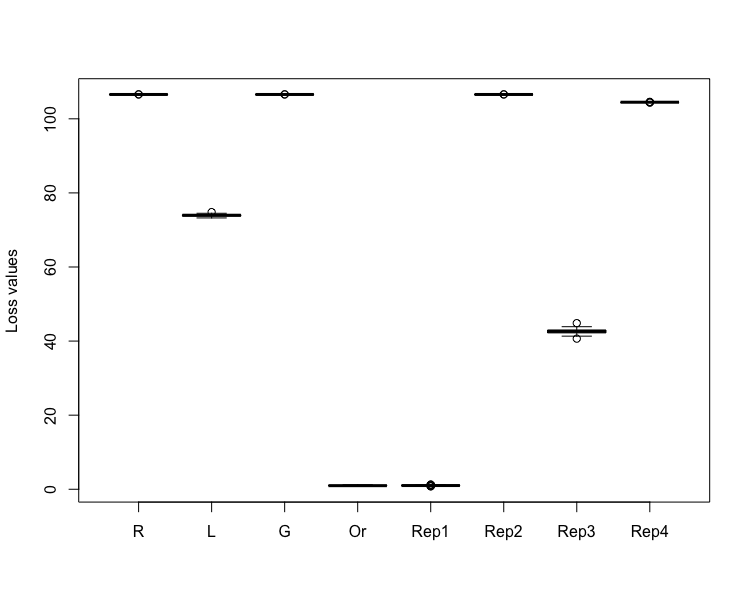}
         \caption{$(1,10)$}
         \label{fig:hi8}
     \end{subfigure}\\
    \begin{subfigure}[b]{0.40\textwidth}
         \centering
         \includegraphics[width=\textwidth]{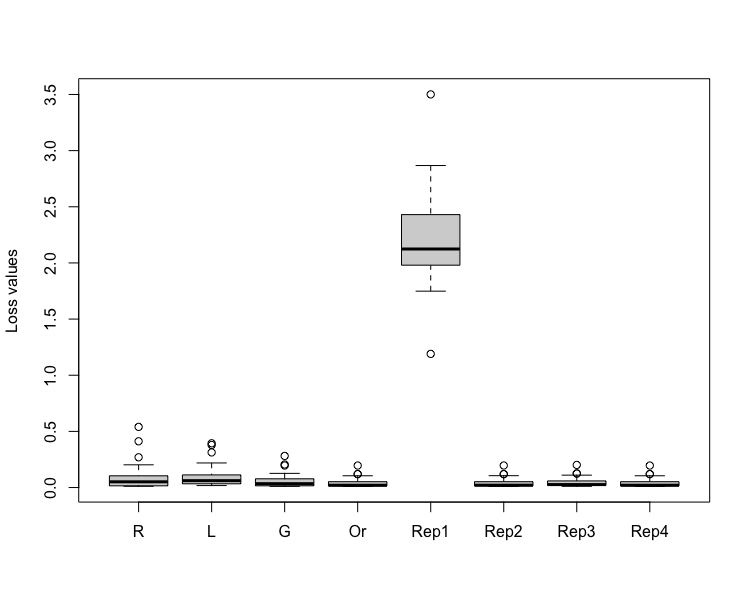}
         \caption{$(2,10^{-1})$}
         \label{fig:hi14}
     \end{subfigure}
     \begin{subfigure}[b]{0.40\textwidth}
         \centering
         \includegraphics[width=\textwidth]{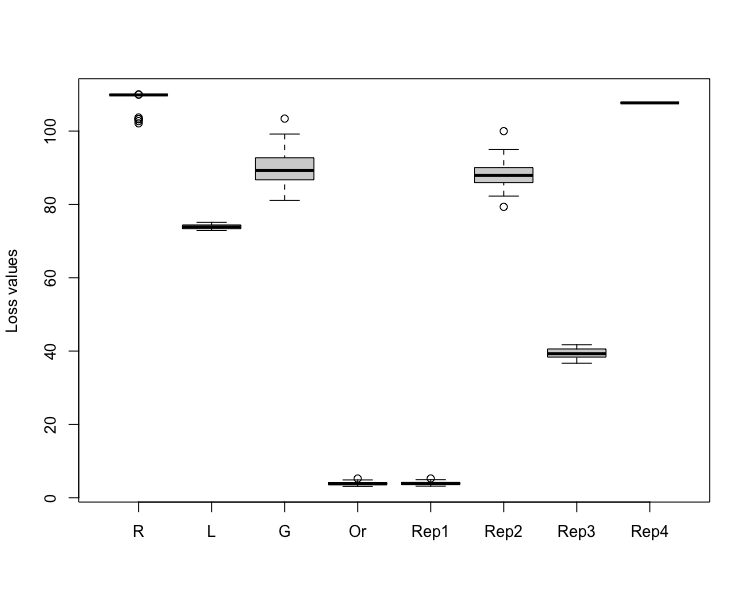}
         \caption{$(2,10)$}
         \label{fig:hi18}
     \end{subfigure}\\
    \begin{subfigure}[b]{0.40\textwidth}
         \centering
         \includegraphics[width=\textwidth]{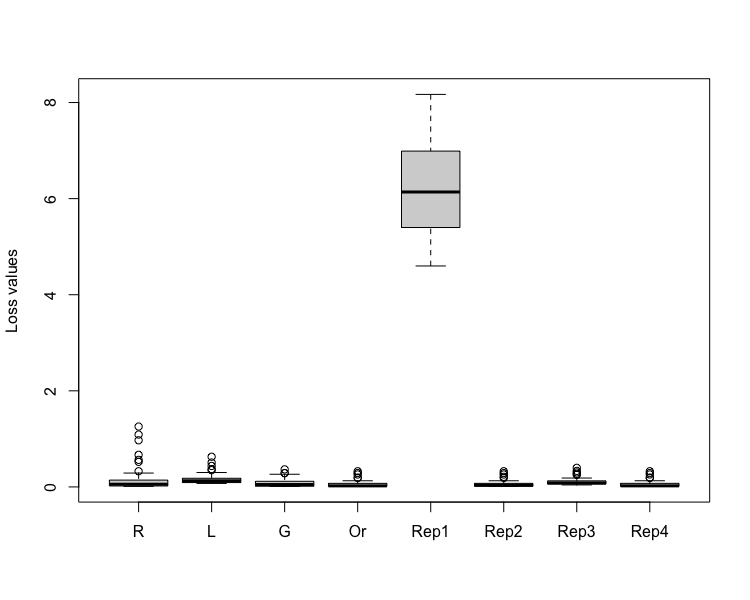}
         \caption{$(3,10^{-1})$}
         \label{fig:hi24}
     \end{subfigure}
     \begin{subfigure}[b]{0.40\textwidth}
         \centering
         \includegraphics[width=\textwidth]{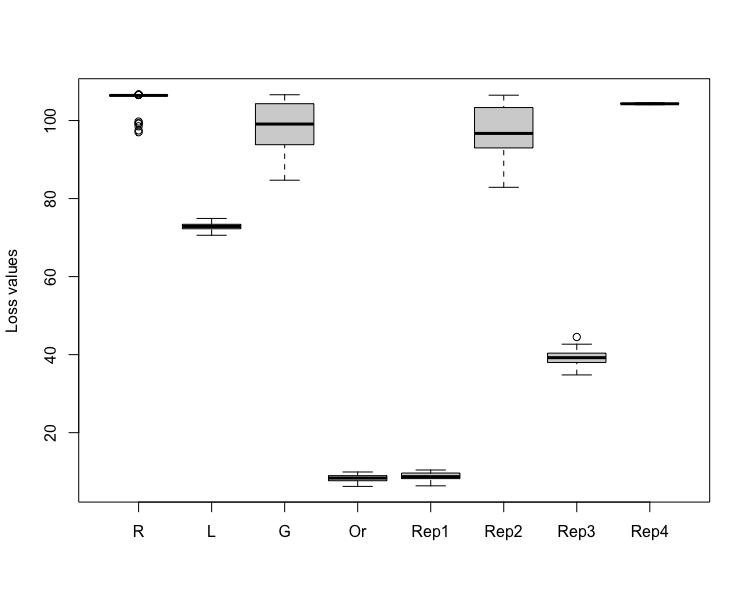}
         \caption{$(3,10)$}
         \label{fig:hi28}
     \end{subfigure}
 \caption{Additional displays to Figure \ref{fig:highdimbox}. Boxplots of the observed same-X loss values from the 50 replications when $(n,p)=(200,300)$.
 }
\label{fig:highdimadd}
\end{figure}

\begin{figure}
     \centering
     \begin{subfigure}[b]{0.45\textwidth}
         \centering
         \includegraphics[width=\textwidth]{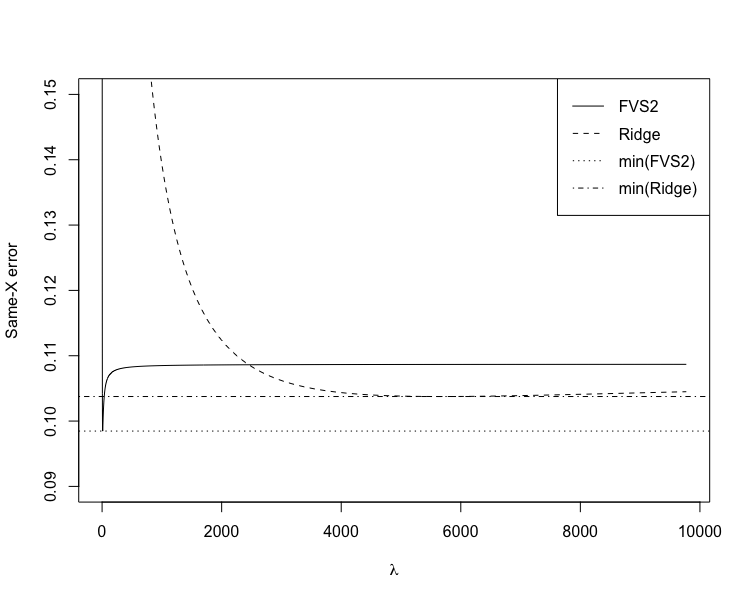}
         \caption{$(1,10^{-0.5})$}
         \label{fig:gr3005}
     \end{subfigure}
     \begin{subfigure}[b]{0.45\textwidth}
         \centering
         \includegraphics[width=\textwidth]{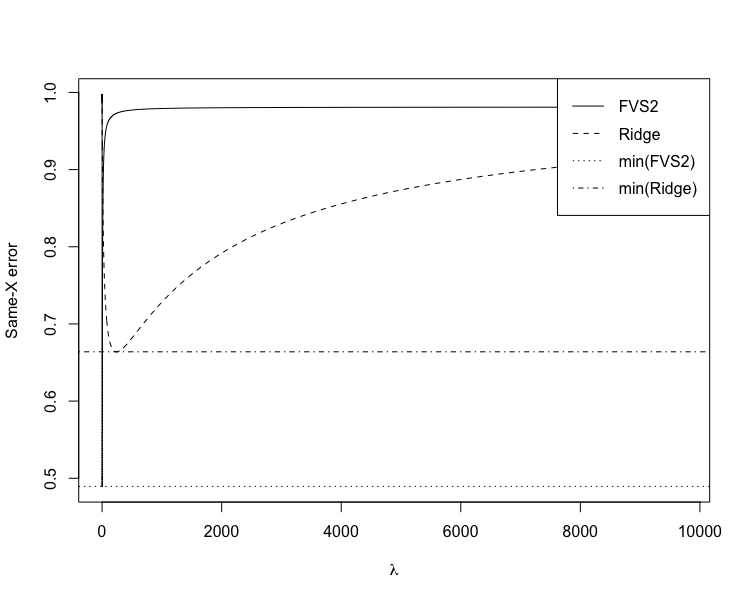}
         \caption{$(1,1)$}
         \label{fig:gr3006}
     \end{subfigure}\\
    \begin{subfigure}[b]{0.45\textwidth}
         \centering
         \includegraphics[width=\textwidth]{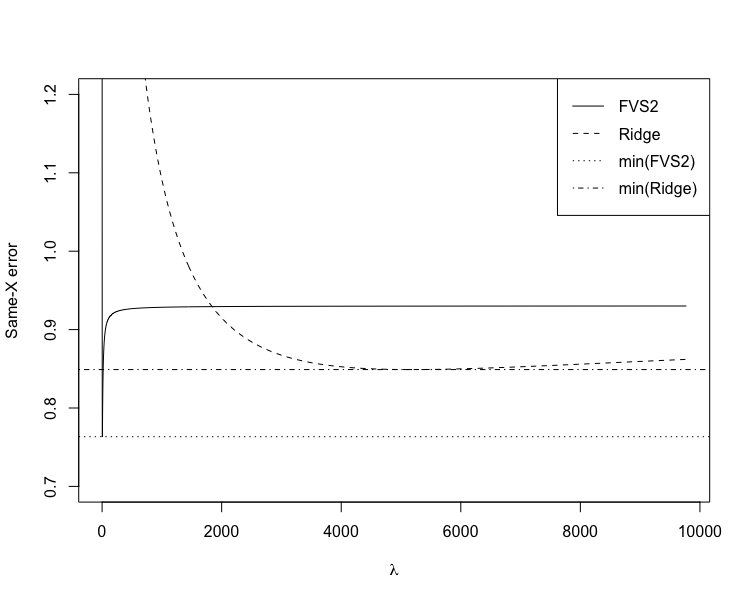}
         \caption{$(2,1)$}
         \label{fig:gr30016}
     \end{subfigure}
     \begin{subfigure}[b]{0.45\textwidth}
         \centering
         \includegraphics[width=\textwidth]{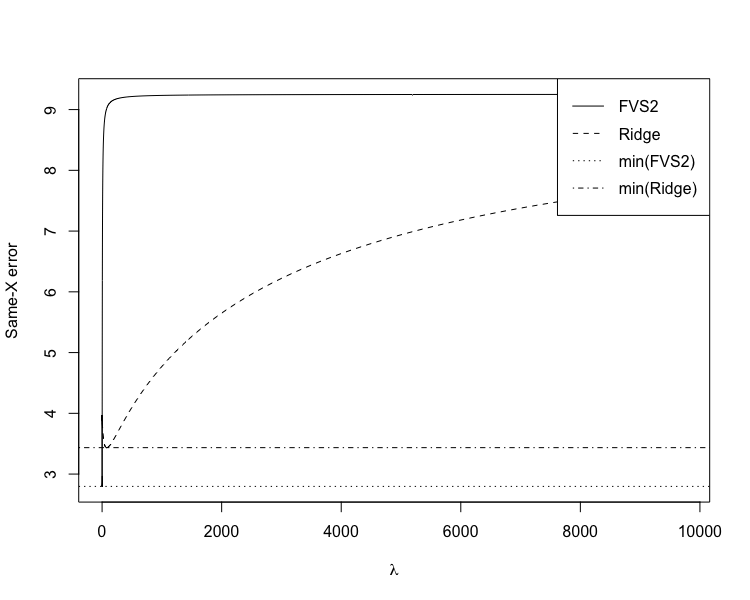}
         \caption{$(2,10^{0.5})$}
         \label{fig:gr30017}
     \end{subfigure}\\
     \begin{subfigure}[b]{0.45\textwidth}
         \centering
         \includegraphics[width=\textwidth]{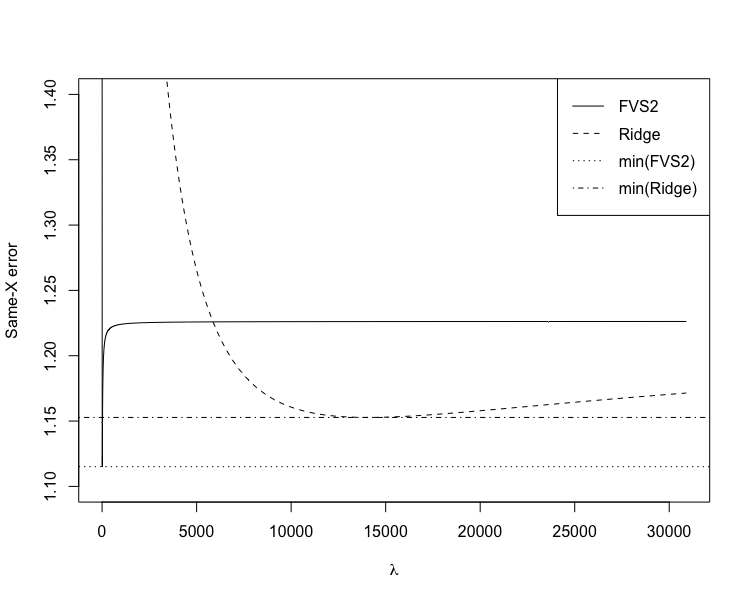}
         \caption{$(3,1)$}
         \label{fig:gr30026}
     \end{subfigure}
     \begin{subfigure}[b]{0.45\textwidth}
         \centering
         \includegraphics[width=\textwidth]{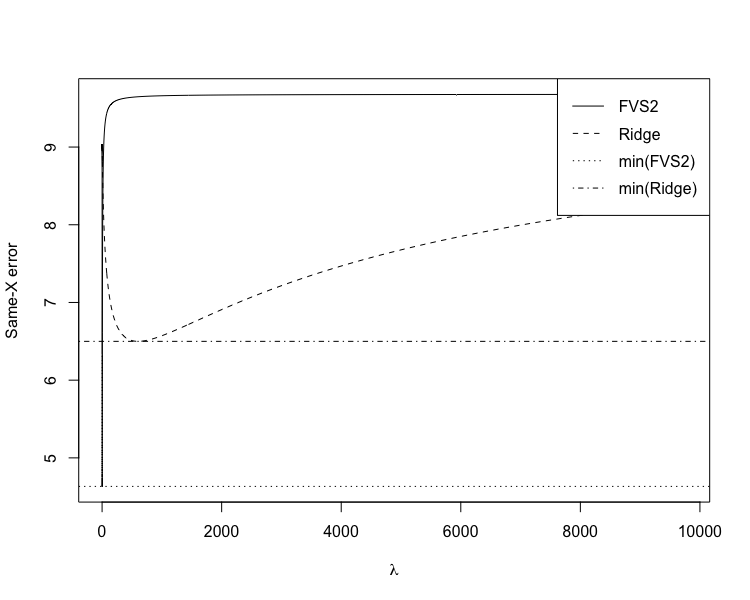}
         \caption{$(3,10^{0.5})$}
         \label{fig:gr30027}
     \end{subfigure}\\
 \caption{Additional displays to Section \ref{highsimulsection} when $(n,p)=(200,300)$. Graphs of the average same-X loss values labeled with $(\sigma,\tau)$ over the 50 replications as a function of $\lambda$ comparing our proposed method \eqref{mainoptlam} to Ridge \eqref{ridgeform}.
 The label FVS2 is fitted value shrinkage with $L_2$-squared penalty \eqref{mainoptlam}, and the horizontal lines, min(FVS2) and min(Ridge), are the minimal average same-X loss values obtained from each estimator within the set of $\lambda$.
 }
\label{fig:highdimgr}
\end{figure}

\clearpage

\subsection{Additional tables for Section \ref{lowsimulsection}}\label{appd:added.low.tables}

Additional supplement tables are in Table \ref{table:lowdim.samex.1}--\ref{table:graph table 3}.

\begin{table}[h!]
\centering
\resizebox{\columnwidth}{!}{%
\begin{tabular}{ |c||c|c|c|c|c|c|c|c|  }
 \hline
 \multicolumn{9}{|c|}{Performance comparison table when $(n,p)=(300,75)$} \\
 \hline
 $\tau$ & OLS & Lasso & Ridge & 2n-G & 2n-Rep & 2n-Or & 2n-Es & 2n-Es95 \\
 \hline
0 & 0.25402  &  0.01053  &  0.00807  &  0.00689  &  0.01076  &  {\bf 0.00410}  & \underline{0.00883}  &  0.00706 \\
& (0.00619) & (0.00179) & (0.0013) & (0.00098) & (0.00141) & (0.00071) & (0.00133) & (0.00136)\\
\hline
$10^{-4}$ & 0.25624  &  0.01554  &  0.01039  &  0.00666  &  0.01049  &  {\bf 0.00349}  &  \underline{0.00807}  &  0.00681 \\
& (0.00625) & (0.00334) & (0.00214) & (0.00107) & (0.00153) & (0.00056) & (0.00145) & (0.00149) \\
\hline
$10^{-2}$ & 0.24498  &  0.00511  &  0.00536  &  0.00583  &  0.00836  &  {\bf 0.00332}  &  \underline{0.00676}  &  0.00557 \\
& (0.00626) & (0.00128) & (0.00100) & (0.00095) & (0.00120) & (0.00063) & (0.00114) & (0.00111)\\
\hline
$10^{-1}$ & 0.25250  &  0.00923  &  0.01054  &  0.00873  &  0.01286  &  {\bf 0.00657}  &  \underline{0.01096}  &  0.00990 \\
& (0.00628) & (0.00119) & (0.00150) & (0.00099) & (0.00163) & (0.00068) & (0.00153) & (0.00154)\\
\hline
$10^{-0.5}$ & 0.24418  &  0.03563  &  0.03294  &  0.03053  &  0.03180  &  {\bf 0.02890}  &  \underline{0.03162}  &  0.03281 \\
& (0.00585) & (0.00157) & (0.00108) & (0.00076) & (0.00100) & (0.00072) & (0.00092) & (0.00087) \\
\hline
$1$ & 0.26053  &  0.16830  &  0.14650  &  0.13044  &  0.12520  &  {\bf 0.12140}  &  \underline{0.12749}  &  0.12843 \\
& (0.00628) & (0.00452) & (0.00363) & (0.00316) & (0.00261) & (0.00242) & (0.00278) & (0.00320)\\
\hline
$10^{0.5}$ & 0.25099  &  0.24172  &  0.24065  &  0.23429  &  0.24686  &  \bf{0.22584}  & 
\underline{0.22714}  &  0.22714 \\
& (0.00548) & (0.00526) & (0.00505) & (0.00487) & (0.00513) & (0.00474) & (0.00472) & (0.00472)\\
\hline
$10$ & 0.24962  &  0.24876  &  0.24805  &  0.24849  &  0.26145  &  {\bf 0.24688}  &  \underline{0.24709}  &  0.24709 \\
& (0.00552) & (0.00554) & (0.00550) & (0.0054) & (0.00567) & (0.00537) & (0.00538) & (0.00538)\\
\hline
$10^{1.5}$ & 0.23895  &  0.23917  &  0.23931  &  0.23895  &  0.24147  &  {\bf 0.23895} &  \underline{0.23897}  &  0.23897 \\
& (0.00570) & (0.00570) & (0.00568) & (0.00570) & (0.00584) & (0.00573) & (0.00572) & (0.00572)\\
\hline
$10^{2}$ & 0.25366  &  0.25340  &  0.25336  &  0.25366  &  {\bf 0.25322}  &  0.25346  &  \underline{0.25346}  &  0.25346 \\
& (0.00638) & (0.00638) & (0.00641) & (0.00638) & (0.00642) & (0.00639) & (0.00639) & (0.00639)\\
\hline
\end{tabular}
}
\caption{The performance comparison table for each simulation setting in Section \ref{lowsimulsection} with $p=75$. The values are the mean squared same-X errors averaged over 50 independent replications.
 The numbers in parentheses are normalized sample standard deviations. The column labels are defined in Section \ref{lowsimulsection}. 
 Boldface indicates the best model. Underlined are our main proposed estimator \textbf{2n-Es}.
}
\label{table:lowdim.samex.1}
\end{table}

\begin{table}[h!]
\centering
\resizebox{\columnwidth}{!}{%
\begin{tabular}{ |c||c|c|c|c|c|c|c|c|  }
 \hline
 \multicolumn{9}{|c|}{Performance comparison table when $(n,p)=(300,150)$} \\
 \hline
 $\tau$ & OLS & Lasso & Ridge & 2n-G & 2n-Rep & 2n-Or & 2n-Es & 2n-Es95 \\
 \hline
0 & 0.49592  &  0.00643  &  0.00602  &  0.00495  &  0.0189  &  {\bf 0.00410}  &  \underline{0.00947}  &  0.00592 \\
& (0.00870) & (0.00106) & (0.00081) & (0.00072) & (0.00181) & (0.00071) & (0.00201) & (0.00191) \\
\hline
$10^{-4}$ & 0.50425  &  0.01029  &  0.00612  &  0.00468  &  0.01963  &  {\bf 0.00349}  &  \underline{0.00875}  &  0.00576 \\
& (0.00778) & (0.00215) & (0.00081) & (0.00069) & (0.00187) & (0.00056) & (0.00175) & (0.00169) \\
\hline
$10^{-2}$ & 0.48596  &  0.00990  &  0.00672  &  0.00406  &  0.01634  &  {\bf 0.00335}  &  \underline{0.00589}  &  0.00335 \\
& (0.00725) & (0.00195) & (0.00104) & (0.00067) & (0.00161) & (0.00063) & (0.00102) & (0.00063) \\
\hline
$10^{-1}$ & 0.49670  &  0.01597  &  0.01338  &  0.01085  &  0.02410  &  {\bf 0.00920}  &  \underline{0.01506}  &  0.01191 \\
& (0.00717) & (0.00239) & (0.00137) & (0.00098) & (0.00211) & (0.00069) & (0.00219) & (0.00214) \\
\hline
$10^{-0.5}$ & 0.49954  &  0.05884  &  0.05137  &  0.04916  &  0.05695  &  \bf{0.04651}  &  \underline{0.05255}  &  0.05373 \\
& (0.00725) & (0.00304) & (0.00114) & (0.00085) & (0.00162) & (0.00083) & (0.00153) & (0.00139) \\
\hline
$1$ & 0.50286  &  0.44147  &  0.36465  &  0.34512  &  0.28957  &  {\bf 0.28137}  &  \underline{0.28761}  &  0.28761 \\
& (0.00767) & (0.01307) & (0.00617) & (0.00664) & (0.00341) & (0.00342) & (0.00365) & (0.00365) \\
\hline
$10^{0.5}$ & 0.49542  &  0.47971  &  0.47440  &  0.52221  &  0.52846  &  {\bf 0.44556}  &  \underline{0.44695}  &  0.44695 \\
& (0.00826) & (0.00773) & (0.00779) & (0.00997) & (0.00898) & (0.00764) & (0.00760) & (0.00760) \\
\hline
$10$ & 0.51361  &  0.51451  &  0.51391  &  0.52422  &  0.58301  &  {\bf 0.51103}  &  \underline{0.51110}  &  0.51110 \\
& (0.00782) & (0.00784) & (0.00785) & (0.00829) & (0.00875) & (0.00777) & (0.00780) & (0.00780) \\
\hline
$10^{1.5}$ & 0.48673  &  0.48670  &  0.48671  &  0.48762  &  0.49737  &  {\bf 0.48617}  &  \underline{0.48626}  &  0.48626 \\
& (0.00837) & (0.00852) & (0.00852) & (0.00826) & (0.00886) & (0.00841) & (0.00841) & (0.00841) \\
\hline
$10^{2}$ & 0.49052  &  0.49088  &  0.49100  &  0.49052  &  0.49097  &  {\bf 0.49041}  &  \underline{0.49040}  &  0.49040  \\
& (0.00858) & (0.00860) & (0.00857) & (0.00858) & (0.00856) & (0.00857) & (0.00857) & (0.00857) \\
\hline
\end{tabular}
}
\caption{The performance comparison table for each simulation setting in Section \ref{lowsimulsection} with $p=150$. The values are the mean squared same-X errors averaged over 50 independent replications.
 The numbers in parentheses are normalized sample standard deviations. The column labels are defined in Section \ref{lowsimulsection}. 
 Boldface indicates the best model. Underlined are our main proposed estimator \textbf{2n-Es}.
}
\label{table:lowdim.samex.2}
\end{table}

\begin{table}[h!]
\centering
\resizebox{\columnwidth}{!}{%
\begin{tabular}{ |c||c|c|c|c|c|c|c|c|  }
 \hline
 \multicolumn{9}{|c|}{Performance comparison table when $(n,p)=(300,250)$} \\
 \hline
 $\tau$ & OLS & Lasso & Ridge & 2n-G & 2n-Rep & 2n-Or & 2n-Es & 2n-Es95 \\
 \hline
0 & 0.82705  &  0.00904  &  0.00897  &  0.00434  &  0.05467  &  {\bf 0.00410}  &  \underline{0.02362}  &  0.01692 \\
& (0.0113) & (0.00165) & (0.00141) & (0.00072) & (0.00307) & (0.00071) & (0.00589) & (0.00562) \\
\hline
$10^{-4}$ & 0.84005  &  0.00856  &  0.00847  &  0.00376  &  0.05760  &  {\bf 0.00349}  &  \underline{0.02110}  &  0.01159 \\
& (0.00935) & (0.00142) & (0.00141) & (0.00057) & (0.00244) & (0.00056) & (0.00521) & (0.00502) \\
\hline
$10^{-2}$ & 0.81920  &  0.00775  &  0.00867  &  0.00356  &  0.05246  &  {\bf 0.00339}  &  \underline{0.01662}  &  0.00707 \\
& (0.01102) & (0.00131) & (0.00194) & (0.00065) & (0.00305) & (0.00063) & (0.00368) & (0.00266) \\
\hline
$10^{-1}$ & 0.82471  &  0.01577  &  0.01542  &  0.01291  &  0.05834  &  {\bf 0.01280}  &  \underline{0.02765}  &  0.01750  \\
& (0.00893) & (0.00126) & (0.00097) & (0.00068) & (0.00251) & (0.00069) & (0.00381) & (0.00324) \\
\hline
$10^{-0.5}$ & 0.83966  &  0.08383  &  0.07713  &  0.07461  &  0.10639  &  \bf{0.06991} &  \underline{0.09452}  &  0.08907 \\
& (0.01143) & (0.00270) & (0.00169) & (0.00077) & (0.00293) & (0.00082) & (0.00535) & (0.00509) \\
\hline
$1$ & 0.84641  &  0.73637  &  0.71328  &  0.68847  &  0.41691  &  \bf{0.40367}  &  \underline{0.42248}  &  0.44553 \\
& (0.01109) & (0.00679) & (0.00271) & (0.00592) & (0.00397) & (0.00393) & (0.00563) & (0.01446) \\
\hline
$10^{0.5}$ & 0.83136  &  6.20342  &  5.83755  &  3.30717  &  1.24991  &  \bf{0.75871}  &  \underline{0.76393}  &  0.76393 \\
& (0.00822) & (0.40666) & (0.31791) & (0.09648) & (0.01594) & (0.00788) & (0.00825) & (0.00825) \\
\hline
$10$ & 0.84872  &  0.84674  &  0.84765  &  2.16844  &  1.46047  &  {\bf 0.83994}  & 
\underline{0.84025}  &  0.84025 \\
& (0.01086) & (0.01077) & (0.01085) & (0.07311) & (0.01548) & (0.01067) & (0.01074) & (0.01074) \\
\hline
$10^{1.5}$ & 0.81350  &  0.81372  &  0.81371  &  1.07344  &  1.02267  &  {\bf 0.81182}  &  \underline{0.81190}  &  0.81190 \\
& (0.00955) & (0.00954) & (0.00954) & (0.03146) & (0.01159) & (0.00950) & (0.00950) & (0.00950) \\
\hline
$10^{2}$ & 0.82234  &  0.82403  &  0.82401  &  0.84015  &  0.84084  & {\bf 0.82251}  &  \underline{0.82254}  &  0.82254 \\
& (0.01043) & (0.01057) & (0.01057) & (0.02032) & (0.01077) & (0.01044) & (0.01044) & (0.01044) \\
\hline
\end{tabular}
}
\caption{The performance comparison table for each simulation setting in Section \ref{lowsimulsection} with $p=250$. The values are the mean squared same-X errors averaged over 50 independent replications.
 The numbers in parentheses are normalized sample standard deviations. The column labels are defined in Section \ref{lowsimulsection}. 
 Boldface indicates the best model. Underlined are our main proposed estimator \textbf{2n-Es}.
}
\label{table:lowdim.samex.3}
\end{table}

\begin{table}[h!]
\centering
\resizebox{\columnwidth}{!}{%
\begin{tabular}{ |c||c|c|c|c|c|  }
 \hline
 \multicolumn{6}{|c|}{Two sample test table} \\
 \hline
 $p$ $\backslash$ $\tau$ & 0 & $10^{-4}$ & $10^{-2}$ & $10^{-1}$ & $10^{-0.5}$ \\
 \hline
$p=75$ & (-0.00097, 0.00249) & {\bf (-0.00429, -0.00034)} & (-0.00017, 0.00297) & (-0.00113, 0.00196) & {\bf(-0.00234, -0.00031)} \\
\hline
$p=150$ & (-0.00060, 0.00749) & (-0.00079. 0.00604) & (-0.00262, 0.00096) & (-0.00193, 0.00528) & (-0.00100, 0.00337) \\
\hline
$p=250$ & (0.00312, 0.02617) & (0.00213, 0.02313) & (0.00011, 0.01577) & (0.00458, 0.01987) & (0.00747, 0.02731) \\
\hline
\end{tabular}
}
\caption{The table includes paired t-test 95\% confidence intervals of $($squared error of our estimator$-$squared error of Ridge$)$ from the simulation setting in Section \ref{lowsimulsection} when $n=300$ and $\tau\le 10^{-0.5}$. The confidence intervals are constructed upon paired 50 independent replications.  
Boldface indicates a case where FVS outperforms Ridge with statistical significance.
}
\label{table:paired test table1}
\end{table}

\begin{table}[h!]
\centering
\resizebox{\columnwidth}{!}{%
\begin{tabular}{ |c||c|c|c|c|c|  }
 \hline
 \multicolumn{6}{|c|}{Two sample test table} \\
 \hline
 $p$ $\backslash$ $\tau$ & 1 & $10^{0.5}$ & $10$ & $10^{1.5}$ & $10^2$ \\
 \hline
$p=75$ & {\bf (-0.02336, -0.01468)} & {\bf (-0.01650, -0.01050)} & (-0.00258, 0.00067) & (-0.00091, 0.00022) & (-0.00047, 0.00068) \\
\hline
$p=150$ & {\bf (-0.08603, -0.06804)} & {\bf (-0.03273, -0.02217)} & {\bf (-0.00454, -0.00109)} & {\bf (-0.00143,  -0.00054)} & {\bf (-0.00091, -0.00029)} \\
\hline
$p=250$ & {\bf (-0.30338, -0.27822)} & {\bf (-5.71148, -4.43574)} & {\bf (-0.00975, -0.00506)} & {\bf (-0.00277, -0.00086)} & (-2.982e-03, 3.625e-05) \\
\hline
\end{tabular}
}
\caption{The table includes paired t-test 95\% confidence intervals of $($squared error of our estimator$-$squared error of Ridge$)$ from the simulation setting in Section \ref{lowsimulsection} when $n=300$ and $\tau\ge 1$. The confidence intervals are constructed upon paired 50 independent replications.  
Boldface indicates a case where FVS outperforms Ridge with statistical significance.
}
\label{table:paired test table2}
\end{table}

\begin{table}[h!]
\centering
\resizebox{\columnwidth}{!}{%
\begin{tabular}{ |c||c|c|c|c|c|c|c|c|c|c|  }
 \hline
 \multicolumn{11}{|c|}{Estimation accuracy table} \\
 \hline
 $p$ $\backslash$ $\tau$ & 0 & $10^{-4}$ & $10^{-2}$ & $10^{-1}$ & $10^{-0.5}$ & 1 & $10^{0.5}$ & $10$ & $10^{1.5}$ & $10^2$ \\
 \hline
$p=75$ & 0.012765  &  0.015367  &  0.012857  &  0.015675  &  0.010584  &  0.006501  &  8.4549e-05  &  8.3909e-07  &  6.86e-09  &  6.3e-11 \\
 & (0.004007) & (0.003858) & (0.003706) & (0.004259) & (0.001927) & (0.001748) & (1.45013e-05) & (1.85515e-07) & (1.296e-09) & (1.5e-11) \\
\hline
$p=150$ & 0.011479  &  0.011239  &  0.009588  &  0.011539  &  0.012060  &  0.005970  &  0.000219  &  9.80702e-07  &  9.625e-09  &  2.3e-10 \\
 & (0.003712) & (0.003248) & (0.002850) & (0.003683) & (0.002737) & (0.001387) & (4.54872e-05) & (1.93014e-07) & (1.969e-09) & (4.6e-11) \\
\hline
$p=250$ & 0.022469  &  0.015587  &  0.020570  &  0.020587  &  0.018257  &  0.015929  &  0.000351  &  4.32618e-06  &  4.5542e-08  &  5.84e-10 \\
 & (0.005912) & (0.004096) & (0.005552) & (0.005919) & (0.004500) & (0.003126) & (6.79379e-05) & (1.11565e-06) & (7.827e-09) & (9.9e-11) \\
\hline
\end{tabular}
}
\caption{The values are the mean squared estimation errors $|\hat\gamma-\gamma_{\rm opt}|^2$ averaged over 100 independent replications from the same simulation setting in Section \ref{lowsimulsection}.
The numbers in parentheses are normalized sample standard deviations.
The sample size is $n=300$, and $p$ varies in $\{75,150,250\}$.}
\label{table:gamopt table}
\end{table}

\begin{table}[h!]
\centering
\resizebox{\columnwidth}{!}{%
\begin{tabular}{ |c||c|c|c|c|c|c|c|c|c|c|  }
 \hline
 \multicolumn{11}{|c|}{Performance comparison table $(n,p)=(300,75)$} \\
 \hline
 $p$ $\backslash$ $\tau$ & 0 & $10^{-4}$ & $10^{-2}$ & $10^{-1}$ & $10^{-0.5}$ & 1 & $10^{0.5}$ & $10$ & $10^{1.5}$ & $10^2$ \\
 \hline
FVS2 & {\bf 0.00320} & {\bf 0.00293} & {\bf 0.00332} & {\bf 0.00607} & {\bf 0.02776} & {\bf 0.11462} & {\bf 0.22182} & {\bf 0.25013} & {\bf 0.25000} & 0.25799 \\
 & (31622.78) & (31622.78) & (5623.413) & (128.8250) & (8.70964) & (1.09648) & (0.10965) & (0.01148) & (0.00066) & (2.13796e-05) \\
\hline
Ridge & 0.00323 & 0.00297 & 0.00337 & 0.00610 & 0.02945 & 0.12852 & 0.23381 & 0.25079 & 0.25007 & {\bf 0.25799} \\
 & (31622.78) & (31622.78) & (31622.78) & (31622.78) & (5370.318) & (245.4709) & (12.58925) & (5.49541) & (0.67608) & (0.10233)\\
\hline
\end{tabular}
}
\caption{The values are the minimal average same-X loss values obtained from each estimator (see \eqref{lam.gr.FVS} and \eqref{lam.gr.Ridge}) within the set of $\lambda\in\{10^{-7+0.01k}:j=0,1,\ldots,1150\}$ from the simulation settings in Section \ref{lowsimulsection} when $(n,p)=(300,75)$.
The numbers in parentheses are the values of $\lambda$ that achieve the minimum for each setting and each estimator. Boldface indicates the best model for each $\tau$.}
\label{table:graph table 1}
\end{table}

\begin{table}[h!]
\centering
\resizebox{\columnwidth}{!}{%
\begin{tabular}{ |c||c|c|c|c|c|c|c|c|c|c|  }
 \hline
 \multicolumn{11}{|c|}{Performance comparison table $(n,p)=(300,150)$} \\
 \hline
 $p$ $\backslash$ $\tau$ & 0 & $10^{-4}$ & $10^{-2}$ & $10^{-1}$ & $10^{-0.5}$ & 1 & $10^{0.5}$ & $10$ & $10^{1.5}$ & $10^2$ \\
 \hline
FVS2 & {\bf 0.00305} & {\bf 0.00310} & {\bf 0.00470} & {\bf 0.00918} & {\bf 0.04589} & {\bf 0.27579} & {\bf 0.44741} & {\bf 0.50089} & {\bf 0.49639} & 0.50113 \\
 & (31622.78) & (31622.78) & (31622.78) & (109.648) & (10.23293) & (0.79433) & (0.11482) & (0.00708) & (0.00083) & (6.60693e-05) \\
\hline
Ridge & 0.00313 & 0.00320 & 0.00479 & 0.00922 & 0.04796 & 0.33086 & 0.47470 & 0.50273 & 0.49678 & {\bf 0.50103} \\
 & (31622.78) & (31622.78) & (31622.78) & (31622.78) & (5754.399) & (138.0384) & (12.30269) & (1.77828) & (0.18197) & (0.43652) \\
\hline
\end{tabular}
}
\caption{The values are the minimal average same-X loss values obtained from each estimator (see \eqref{lam.gr.FVS} and \eqref{lam.gr.Ridge}) within the set of $\lambda\in\{10^{-7+0.01k}:j=0,1,\ldots,1150\}$ from the simulation settings in Section \ref{lowsimulsection} when $(n,p)=(300,150)$.
The numbers in parentheses are the values of $\lambda$ that achieve the minimum for each setting and each estimator. Boldface indicates the best model for each $\tau$.}
\label{table:graph table 2}
\end{table}

\begin{table}[h!]
\centering
\resizebox{\columnwidth}{!}{%
\begin{tabular}{ |c||c|c|c|c|c|c|c|c|c|c|  }
 \hline
 \multicolumn{11}{|c|}{Performance comparison table $(n,p)=(300,250)$} \\
 \hline
 $p$ $\backslash$ $\tau$ & 0 & $10^{-4}$ & $10^{-2}$ & $10^{-1}$ & $10^{-0.5}$ & 1 & $10^{0.5}$ & $10$ & $10^{1.5}$ & $10^2$ \\
 \hline
FVS2 & {\bf 0.02187} & {\bf 0.02105} & {\bf 0.02609} & {\bf 0.01134} & {\bf 0.07150} & {\bf 0.41713} & {\bf 0.90093} & {\bf 0.90040} & 0.92051 & {\bf 0.83885} \\
 & (31622.78) & (31622.78) & (31622.78) & (323.5937) & (11.48154) & (1.07152) & (0.00047) & (1.07152e-05) & (3.38844e-07) & (2.63027e-05) \\
\hline
Ridge & 0.02235 & 0.02123 & 0.02626 & 0.01138 & 0.15322 & 0.62548 & 0.90505 & 0.90242 & {\bf 0.84360} & 0.84018 \\
 & (10000) & (31622.78) & (31622.78) & (31622.78) & (7762.471) & (245.4709) & (0.16982) & (0.00115) & (1.02329) & (1e-07) \\
\hline
\end{tabular}
}
\caption{The values are the minimal average same-X loss values obtained from each estimator (see \eqref{lam.gr.FVS} and \eqref{lam.gr.Ridge}) within the set of $\lambda\in\{10^{-7+0.02k}:j=0,1,\ldots,1150\}$ from the simulation settings in Section \ref{lowsimulsection} when $(n,p)=(300,250)$.
The numbers in parentheses are the values of $\lambda$ that achieve the minimum for each setting and each estimator. Boldface indicates the best model for each $\tau$.}
\label{table:graph table 3}
\end{table}

\clearpage

\subsection{Additional table for Section \ref{factorsimul}}\label{appd:factorsimul}

\begin{table}[h!]
\centering
\resizebox{\columnwidth}{!}{%
\begin{tabular}{ |c||c|c|c|c|c|c|c|c|c|c|c|  }
 \hline
 \multicolumn{12}{|c|}{Performance comparison table} \\
 \hline
 $(\tau_c,\tau_f)$,Coding &  OLS & 2n-Or & 2n-Es & 2n-Es95 & 2n-G & 2n-Rep & Ridge & LASSO & GL & MGL & SGL\\
 \hline
 ($10^{-0.5}$,$10^{-0.5}$) & 0.50186 & {\bf 0.11723} & \underline{0.14239} & 0.15445 & 0.12866 & 0.12664 & 0.12925 & 0.15168 & 0.16374 & 0.15945 & 0.15307 \\
 Coding-1 & (0.01530) & (0.00261) & (0.00467) & (0.00354) & (0.00309) & (0.00319) & (0.00540) & (0.00645) & (0.00812) & (0.00579) & (0.00547) \\
 \hline
 ($10^{-0.5}$,$10^{-0.5}$) & 0.50186 & {\bf 0.11723} & \underline{0.14239} & 0.15445 & 0.12866 & 0.12663 & 0.13380 & 0.15266 & 0.16712 & 0.15874 & 0.15489 \\
 Coding-2 & (0.01530) & (0.00261) & (0.00467) & (0.00354) & (0.00309) & (0.00317) & (0.00499) & (0.00708) & (0.00862) & (0.00693) & (0.00658) \\
 \hline
 (1,$10^{-0.5}$) & 0.50927 & {\bf 0.17630} & \underline{0.19666} & 0.22610 & 0.20258 & 0.18252 & 0.22405 & 0.24681 & 0.25424 & 0.24891 & 0.24874 \\
 Coding-1 & (0.01303) & (0.00430) & (0.00513) & (0.00656) & (0.00517) & (0.00443) & (0.00589) & (0.00651) & (0.00562) & (0.00774) & (0.00697) \\
 \hline
 (1,$10^{-0.5}$) & 0.50927 & {\bf 0.17630} & \underline{0.19666} & 0.22610 & 0.20258 & 0.18265 & 0.22813 & 0.24960 & 0.25910 & 0.24812 & 0.24918 \\
 Coding-2 & (0.01303) & (0.00430) & (0.00513) & (0.00656) & (0.00517) & (0.00443) & (0.00485) & (0.00722) & (0.00702) & (0.00725) & (0.00646) \\
 \hline
 ($10^{-0.5}$,1) & 0.49276 & 0.22100 & \underline{0.24620} & 0.27809 & 0.27201 & 0.23677 & {\bf 0.21132} & 0.22176 & 0.34283 & 0.27566 & 0.24439 \\
 Coding-1 & (0.01510) & (0.00419) & (0.00735) & (0.01182) & (0.00739) & (0.00554) & (0.00819) & (0.01228) & (0.01101) & (0.01400) & (0.01359) \\
 \hline
 ($10^{-0.5}$,1) & 0.49276 & {\bf 0.22100} & \underline{0.24620} & 0.27809 & 0.27201 & 0.23757 & 0.28954 & 0.29632 & 0.38464 & 0.34066 & 0.34019 \\
 Coding-2 & (0.01510) & (0.00419) & (0.00735) & (0.01182) & (0.00739) & (0.00555) & (0.00869) & (0.01005) & (0.00724) & (0.01048) & (0.01039) \\
 \hline
 (1,1) & 0.50877 & 0.38063 & \underline{0.39377} & 0.39377 & 0.49954 & 0.44260 & {\bf 0.34082} & 0.37666 & 0.40414 & 0.39063 & 0.38946 \\
 Coding-1 & (0.01491) & (0.01185) & (0.01306) & (0.01306) & (0.01955) & (0.01393) & (0.01153) & (0.01203) & (0.01342) & (0.01355) & (0.01303) \\
 \hline
 (1,1) & 0.50877 & {\bf 0.38063} & \underline{0.39377} & 0.39377 & 0.49954 & 0.46147 & 0.46169 & 0.47966 & 0.47640 & 0.44753 & 0.44525 \\
 Coding-2 & (0.01491) & (0.01185) & (0.01306) & (0.01306) & (0.01955) & (0.01426) & (0.02263) & (0.01306) & (0.03130) & (0.02385) & (0.01600) \\
 \hline
 ($10^{0.5}$,1) & 0.49117 & 0.43134 & \underline{0.43714} & 0.43714 & 0.54382 & 0.51857 & 0.42296 & 0.42204 &  {\bf 0.41071} & 0.41779 & 0.41682 \\
 Coding-1 & (0.01107) & (0.00918) & (0.00937) & (0.00937) & (0.01375) & (0.01127) & (0.00959) & (0.01019) & (0.01088) & (0.00985) & (0.00939) \\
 \hline
 ($10^{0.5}$,1) & 0.49117 & {\bf 0.43134} & \underline{0.43714} & 0.43714 & 0.54382 & 0.52616 & 0.46032 & 0.47328 & 0.47515 & 0.46817 & 0.46450 \\
 Coding-2 & (0.01107) & (0.00918) & (0.00937) & (0.00937) & (0.01375) & (0.01135) & (0.01040) & (0.01158) & (0.01182) & (0.01196) & (0.01110) \\
 \hline
 (1,$10^{0.5}$) & 0.49584 & 0.46276 & \underline{0.46526} & 0.46526 & 0.54069 & 0.62294 & 0.40551 & {\bf 0.32561} & 0.40242 & 0.37382 & 0.36664 \\
 Coding-1 & (0.01566) & (0.01462) & (0.01486) & (0.01486) & (0.01945) & (0.01823) & (0.01282) & (0.01191) & (0.01405) & (0.01260) & (0.01275) \\
 \hline
 (1,$10^{0.5}$) & 0.49584 & {\bf 0.46276} & \underline{0.46526} & 0.46526 & 0.54069 & 0.65369 & 0.49049 & 0.47572 & 0.47647 & 0.46569 & 0.47035 \\
 Coding-2 & (0.01566) & (0.01462) & (0.01486) & (0.01486) & (0.01945) & (0.01840) & (0.01506) & (0.01511) & (0.01465) & (0.01399) & (0.01488) \\
 \hline
 ($10^{0.5}$,$10^{0.5}$) & 0.48446 & 0.47790 & \underline{0.47878} & 0.47878 & 0.51832 & 0.56642 & 0.47364 & {\bf 0.44926} & 0.46486 & 0.46730 & 0.47053 \\
 Coding-1 & (0.01323) & (0.01278) & (0.01282) & (0.01282) & (0.01632) & (0.01528) & (0.01344) & (0.01239) & (0.01290) & (0.01303) & (0.01308) \\
 \hline
 ($10^{0.5}$,$10^{0.5}$) & 0.48446 & 0.47790 & \underline{0.47878} & 0.47878 & 0.51832 & 0.62073 & 0.47599 & 0.47412 & 0.47585 & {\bf 0.47288} & 0.47398 \\
 Coding-2 & (0.01323) & (0.01278) & (0.01282) & (0.01282) & (0.01632) & (0.01692) & (0.01341) & (0.01325) & (0.01318) & (0.01302) & (0.01295) \\
 \hline
\end{tabular}
}
 \caption{The performance comparison table for each simulation setting in Section \ref{factorsimul} characterized by a tuple of $(\tau_c,\tau_f)$ and a Coding number. The values are the mean squared same-X errors averaged over 50 independent replications.
 The numbers in parentheses are normalized sample standard deviations. The column labels are defined in Section \ref{factorsimul}. 
 Boldface indicates the best model. Underlined is our main proposed estimator.}
\label{table:low dim factor table}
\end{table}

\clearpage

\subsection{Additional table for Section \ref{fullrankmtx}}\label{appd:low.XT}

\begin{table}[h!]
\centering
\resizebox{\columnwidth}{!}{%
\begin{tabular}{ |c||c|c|c|c|c|c|c|c|  }
 \hline
 \multicolumn{9}{|c|}{Performance comparison table} \\
 \hline
 $(\psi,s)$,Coding &  OLS & 2n-Or & 2n-Es & 2n-Es95 & 2n-G & 2n-Rep & Ridge & LASSO \\
 \hline
 (0,0.025) & 0.50839  &  0.39358  &  \underline{0.39762}  &  0.39762  &  0.47646  &  0.40001  &  0.31489  &  {\bf 0.06404} \\
 Coding-1 & (0.00846) & (0.00662) & (0.00649) & (0.00649) & (0.01003) & (0.00664) & (0.00569) & (0.00375) \\
 \hline
 (0,0.025) & 0.50839  &  {\bf 0.39358} &  \underline{0.39762}  &  0.39762  &  0.47646  &  0.42485  &  0.60588  &  0.68829 \\
 Coding-2 & (0.00846) & (0.00662) & (0.00649) & (0.00649) & (0.01003) & (0.00736) & (0.01615) & (0.03787) \\
 \hline
 (0,0.05) & 0.49586  &  0.46544  &  \underline{0.46678}  &  0.46678  &  0.52128  &  0.47463  &  0.40975  &  {\bf 0.11051} \\
 Coding-1 & (0.00794) & (0.00769) & (0.00778) & (0.00778) & (0.00896) & (0.00809) & (0.00650) & (0.00441) \\
 \hline
 (0,0.05) & 0.49586  &  {\bf 0.46544}  &  \underline{0.46678}  &  0.46678  &  0.52128  &  0.47599  &  0.68096  &  0.69950 \\
 Coding-2 & (0.00794) & (0.00769) & (0.00778) & (0.00778) & (0.00896) & (0.00812) & (0.03665) & (0.04346) \\
 \hline
 (0,0.1) & 0.49187 & 0.47031 & \underline{0.47119} & 0.47119 & 0.51185 & 0.48061 & 0.44077 & {\bf 0.17147} \\
 Coding-1 & (0.00753) & (0.00713) & (0.00713) & (0.00713) & (0.00900) & (0.00733) & (0.00664) & (0.00481) \\
 \hline
 (0,0.1) & 0.49187 & {\bf 0.47031} & \underline{0.47119} & 0.47119 & 0.51185 & 0.48354 & 0.58174 & 0.63748 \\
 Coding-2 & (0.00753) & (0.00713) & (0.00713) & (0.00713) & (0.00900) & (0.00740) & (0.04227) & (0.04410) \\
 \hline
 (0,0.2) & 0.51313 & 0.50418 & \underline{0.50468} & 0.50468 & 0.52139 & 0.50593 & 0.48666 & {\bf 0.26380} \\
 Coding-1 & (0.00779) & (0.00776) & (0.00778) & (0.00778) & (0.00873) & (0.00791) & (0.00739) & (0.00769) \\
 \hline
 (0,0.2) & 0.51313 & \bf{0.50418} & \underline{0.50468} & 0.50468 & 0.52139 & 0.50602 & 1.08212 & 1.76470 \\
 Coding-2 & (0.00779) & (0.00776) & (0.00778) & (0.00778) & (0.00873) & (0.00791) & (0.01940) & (0.05839) \\
 \hline
 (0,0.3) & 0.49183 & 0.48747 & \underline{0.48770} & 0.48770 & 0.50323 & 0.49037 & 0.46678 & {\bf 0.27960} \\
 Coding-1 & (0.00657) & (0.00667) & (0.00667) & (0.00667) & (0.00713) & (0.00681) & (0.00653) & (0.00569) \\
 \hline
 (0,0.3) & 0.49183 & {\bf 0.48747} & \underline{0.48770} & 0.48770 & 0.50323 & 0.49046 & 4.19051 & 3.41197 \\
 Coding-2 & (0.00657) & (0.00667) & (0.00667) & (0.00667) & (0.00713) & (0.00681) & (0.05049) & (0.08670) \\
 \hline
 (1,0.025) & 0.50711  &  0.33304  &  0.33710  &  \underline{0.33710}  &  0.39516  &  0.33499  &  0.25947  &  {\bf 0.07922} \\
 Coding-1 & (0.00784) & (0.00516) & (0.00522) & (0.00522) & (0.00717) & (0.00519) & (0.00400) & (0.00418) \\
 \hline
 (1,0.025) & 0.50711  &  {\bf 0.33304}  &  \underline{0.33710}  &  0.33710  &  0.39516  &  0.33688  &  0.56775  &  0.62388 \\
 Coding-2 & (0.00784) & (0.00516) & (0.00522) & (0.00522) & (0.00717) & (0.00524) & (0.01124) & (0.01830) \\
 \hline
 (1,0.05) & 0.48570  &  0.31444  &  \underline{0.31995}  &  0.31995  &  0.38558  &  0.31870  &  0.25929  &  {\bf 0.09830} \\
 Coding-1 & (0.00845) & (0.00497) & (0.00527) & (0.00527) & (0.00736) & (0.00498) & (0.00591) & (0.00593) \\
 \hline
 (1,0.05) & 0.48570  &  {\bf 0.31444}  &  \underline{0.31995}  &  0.31995  &  0.38558  &  0.32063  &  0.46405  &  0.46705 \\
 Coding-2 & (0.00845) & (0.00497) & (0.00527) & (0.00527) & (0.00736) & (0.00505) & (0.00982) & (0.01272) \\
 \hline
 (1,0.1) & 0.51035  &  0.45242  &  \underline{0.45569}  &  0.45569  &  0.53442  &  0.45991  &  0.38374  &  {\bf 0.18685} \\
 Coding-1 & (0.00893) & (0.00731) & (0.00729) & (0.00729) & (0.01030) & (0.00723) & (0.00619) & (0.00620) \\
 \hline
 (1,0.1) & 0.51035  &  {\bf 0.45242}  &  \underline{0.45569}  &  0.45569  &  0.53442  &  0.46254  &  0.87261  &  1.01879 \\
 Coding-2 & (0.00893) & (0.00731) & (0.00729) & (0.00729) & (0.01030) & (0.00727) & (0.02527) & (0.05128) \\
 \hline
 (1,0.2) & 0.49424  &  0.45459  &  \underline{0.45662}  &  0.45662  &  0.50326  &  0.45741  &  0.38720  &  {\bf 0.23477} \\
 Coding-1 & (0.00814) & (0.00793) & (0.00791) & (0.00791) & (0.00995) & (0.00816) & (0.00664) & (0.00653) \\
 \hline
 (1,0.2) & 0.49424  &  {\bf 0.45459}  &  \underline{0.45662}  &  0.45662  &  0.50326  &  0.45783  &  0.63929  &  0.81334 \\
 Coding-2 & (0.00814) & (0.00793) & (0.00791) & (0.00791) & (0.00995) & (0.00819) & (0.01402) & (0.03587) \\
 \hline
 (1,0.3) & 0.50315  &  0.48846  &  \underline{0.48878}  &  0.48878  &  0.51896  &  0.49062  &  0.41947  &  {\bf 0.29917} \\
 Coding-1 & (0.00816) & (0.00810) & (0.00815) & (0.00815) & (0.00918) & (0.00826) & (0.00729) & (0.00752) \\
 \hline
 (1,0.3) & 0.50315  &  {\bf 0.48846}  &  \underline{0.48878}  &  0.48878  &  0.51896  &  0.49078  &  0.67030  &  0.82268 \\
 Coding-2 & (0.00816) & (0.00810) & (0.00815) & (0.00815) & (0.00918) & (0.00827) & (0.01589) & (0.04973) \\
 \hline
\end{tabular}
}
 \caption{The performance comparison table for each simulation setting in Section \ref{fullrankmtx} characterized by a tuple of $(\psi,s)$ and a Coding number. The values are the mean squared same-X errors averaged over 50 independent replications.
 The numbers in parentheses are normalized sample standard deviations. The column labels are defined in Section \ref{lowsimulsection}. 
 Boldface indicates the best model. Underlined is our main proposed estimator.}
\label{table:low dim XT table}
\end{table}

\clearpage

\subsection{Additional tables for Section \ref{highsimulsection}}\label{appd:added.high.tables}

Additional supplement tables are in Table \ref{table:highdim.samex.1}--\ref{table:graph table 4}.

\begin{table}[h!]
\centering
\resizebox{\columnwidth}{!}{%
\begin{tabular}{ |c||c|c|c|c|c|c|c|c|  }
 \hline
 \multicolumn{9}{|c|}{Performance comparison table when $(n,p,\sigma)=(200,300,1)$} \\
 \hline
 $\tau$ & Ridge & Lasso & 2n-G & 2n-Or & 2n-Rep1 & 2n-Rep2 & 2n-Rep3 & 2n-Rep4 \\
 \hline
0 & 0.01786  &  0.02242  &  0.01469  &  {\bf 0.00593} &  0.28577  &  0.00601  &  0.00616  &  \underline{0.00595} \\
& (0.00339) & (0.00542) & (0.00300) & (0.00129) & (0.00780) & (0.00128) & (0.00129) & (0.00129) \\
\hline
$10^{-4}$ & 0.01432  &  0.01468  &  0.01450  &  {\bf 0.00689}  &  0.26896  &  0.00696  &  0.00711  &  \underline{0.00692} \\
& (0.00266) & (0.00290) & (0.00275) & (0.00165) & (0.00749) & (0.00164) & (0.00165) & (0.00165) \\
\hline
$10^{-2}$ & 0.02166  &  0.02623  &  0.01219  &  {\bf 0.00513}  &  0.27957  &  0.00524  &  0.00536  &  \underline{0.00515} \\
& (0.00430) & (0.00763) & (0.00255) & (0.00084) & (0.00757) & (0.00083) & (0.00083) & (0.00084) \\
\hline
$10^{-1}$ & 0.02276  &  0.03052  &  0.01604  &  {\bf 0.01172}  &  0.26729  &  0.01180  &  0.01175  &  \underline{0.01174} \\
& (0.00293) & (0.00521) & (0.00136) & (0.00061) & (0.00624) & (0.00060) & (0.00061) & (0.00061) \\
\hline
$10^{-0.5}$ & 0.11405  &  0.12083  &  0.10753  &  {\bf 0.10100}  &  0.31219  &  0.11014  &  \underline{0.10811}  &  0.11011 \\
& (0.00300) & (0.00350) & (0.00146) & (0.00118) & (0.00684) & (0.00120) & (0.00115) & (0.00116) \\
\hline
$1$ & 0.91282  &  0.96185  &  0.97451  &  {\bf 0.49907}  &  \underline{0.53213}  &  0.97599  &  0.93679  &  0.97061 \\
& (0.00330) & (0.00612) & (0.00338) & (0.00546) & (0.00679) & (0.00175) & (0.00105) & (0.00060)\\
\hline
$10^{0.5}$ & 8.49899  &  9.75743  &  8.42798  &  {\bf 0.88561}  &  \underline{0.89875}  &  9.80771  &  8.24792  &  9.73334 \\
& (0.00707) & (0.02054) & (0.10022) & (0.01310) & (0.01264) & (0.01598) & (0.01343) & (0.00107) \\
\hline
$10$ & 73.96427  &  106.5722  &  106.5722  &  {\bf 0.97842}  &  \underline{1.01806}  &  106.5722  &  42.61824  &  104.4773 \\
& (0.04330) & (0.00082) & (0.00082) & (0.01307) & (0.01357) & (0.00082) & (0.10719) & (0.00446) \\
\hline
$10^{1.5}$ & 361.32690  &  791.61740  &  711.54760  &  {\bf 1.00270} &  \underline{1.00436}  &  681.97860  &  46.79847  &  728.59430\\
 & (0.17178) & (0.30844) & (0.79395) & (0.01140) & (0.01138) & (1.35906) & (0.15520) & (0.06574) \\
\hline
$10^2$ & 2377.538  &  10718.830  &  10718.830  &  \bf{0.99812}  &  \underline{0.99848}  &  10522.250  &  14.32732  &  5754.502 \\
& (0.57873) & (0.00150) & (0.00150) & (0.01184) & (0.01183) & (6.38851) & (0.07117) & (0.92518) \\
\hline
\end{tabular}
}
\caption{The performance comparison table for each simulation setting in Section \ref{highsimulsection} with $(n,p)=(200,300)$ and $\sigma=1$. The values are the mean squared same-X errors averaged over 50 independent replications.
The numbers in parentheses are normalized sample standard deviations. The column labels are defined in Section \ref{highsimulsection}. 
Boldface indicates the best model. Underlined is the best model within our proposed estimators \textbf{2n-Rep1} to \textbf{2n-Rep4}.
}
\label{table:highdim.samex.1}
\end{table}

\begin{table}[h!]
\centering
\resizebox{\columnwidth}{!}{%
\begin{tabular}{ |c||c|c|c|c|c|c|c|c|  }
 \hline
 \multicolumn{9}{|c|}{Performance comparison table when $(n,p,\sigma)=(200,300,2)$} \\
 \hline
 $\tau$ & Ridge & Lasso & 2n-G & 2n-Or & 2n-Rep1 & 2n-Rep2 & 2n-Rep3 & 2n-Rep4 \\
 \hline
0 & 0.07935  &  0.05637  &  0.06283  &  {\bf 0.03277}  &  2.34612  &  0.03383  &  0.04006  &  \underline{0.03288} \\
& (0.01238) & (0.01358) & (0.01605) & (0.00934) & (0.04869) & (0.00934) & (0.00936) & (0.00934) \\
\hline
$10^{-4}$ & 0.06238  &  0.03690  &  0.04443  &  {\bf 0.01775}  &  2.2782  &  0.01927  &  0.02505  &  \underline{0.01786} \\
& (0.01024) & (0.00682) & (0.01024) & (0.00341) & (0.05627) & (0.00347) & (0.00342) & (0.00341) \\
\hline
$10^{-2}$ & 0.07281  &  0.06067  &  0.04072  &  {\bf 0.02118}  &  2.30044  &  0.02258  &  0.02851  &  \underline{0.02129} \\
& (0.01027) & (0.01262) & (0.00533) & (0.00358) & (0.04880) & (0.00350) & (0.00351) & (0.00358) \\
\hline
$10^{-1}$ & 0.08896  &  0.08272  &  0.05622  &  {\bf 0.03965}  &  2.21296  &  0.04086  &  0.04557  &  \underline{0.03967}\\
& (0.01202) & (0.0142) & (0.00805) & (0.00547) & (0.0513) & (0.00549) & (0.00545) & (0.00547) \\
\hline
$10^{-0.5}$ & 0.14131  &  0.14596  &  0.12302  &  {\bf 0.11325}  &  2.22741  &  0.11519  &  \underline{0.11446}  &  0.11492 \\
& (0.00513) & (0.01019) & (0.00500) & (0.00316) & (0.04056) & (0.00311) & (0.00323) & (0.00312) \\
\hline
$1$ & 0.86595  &  0.96481  &  0.87980  &  {\bf 0.76578}  &  2.57110  &  0.91357  &  \underline{0.85542}  &  0.92523 \\
& (0.00976) & (0.02339) & (0.01148) & (0.00663) & (0.04396) & (0.00637) & (0.00545) & (0.00549) \\
\hline
$10^{0.5}$ & 7.65122  &  9.14910  &  9.12699  &  {\bf 2.73597}  &  \underline{3.20867}  &  8.98522  &  7.32603  &  9.15122 \\
& (0.04529) & (0.07148) & (0.04857) & (0.03793) & (0.04368) & (0.07203) & (0.03633) & (0.00588) \\
\hline
$10$ & 73.88077  &  109.19530  &  89.77449  & {\bf 3.88429}  &  \underline{3.92254}  &  88.15042  &  39.3269  &  107.6866 \\
& (0.08510) & (0.29530) & (0.64223) & (0.06352) & (0.06398) & (0.50894) & (0.18335) & (0.00819) \\
\hline
$10^{1.5}$ & 481.4546  &  1002.2620  &  1002.2620  &  {\bf 3.91530}  &  \underline{3.91610}  &  1002.2620  &  63.5840  &  911.6336 \\
 & (0.31973) & (0.00385) & (0.00385) & (0.05169) & (0.05206) & (0.00385) & (0.31341) & (0.16549) \\
\hline
$10^2$ & 2359.036  &  10597.59  &  10597.59  &  {\bf 3.99555}  &  \underline{3.99651}  &  10597.59  &  18.45940  &  5602.045 \\
& (1.05413) & (0.00368) & (0.00368) & (0.06336) & (0.06341) & (0.00368) & (0.16666) & (2.19424) \\
\hline
\end{tabular}
}
\caption{The performance comparison table for each simulation setting in Section \ref{highsimulsection} with $(n,p)=(200,300)$ and $\sigma=2$. The values are the mean squared same-X errors averaged over 50 independent replications.
 The numbers in parentheses are normalized sample standard deviations. The column labels are defined in Section \ref{highsimulsection}. 
 Boldface indicates the best model. Underlined is the best model within our proposed estimators \textbf{2n-Rep1} to \textbf{2n-Rep4}.
}
\label{table:highdim.samex.2}
\end{table}

\begin{table}[h!]
\centering
\resizebox{\columnwidth}{!}{%
\begin{tabular}{ |c||c|c|c|c|c|c|c|c|  }
 \hline
 \multicolumn{9}{|c|}{Performance comparison table when $(n,p,\sigma)=(200,300,3)$} \\
 \hline
 $\tau$ & Ridge & Lasso & 2n-G & 2n-Or & 2n-Rep1 & 2n-Rep2 & 2n-Rep3 & 2n-Rep4 \\
 \hline
0 & 0.22832  &  0.16547  &  0.08650  &  {\bf 0.05316}  &  6.61394  &  0.06214  &  0.11587  &  \underline{0.05343} \\
& (0.04175) & (0.04596) & (0.01490) & (0.00988) & (0.10294) & (0.01082) & (0.01077) & (0.00988) \\
\hline
$10^{-4}$ & 0.17094  &  0.13428  &  0.06448  &  {\bf 0.04548}  &  6.63878  &  0.04895  &  0.10504  &  \underline{0.04575} \\
& (0.02052) & (0.03632) & (0.00990) & (0.00954) & (0.11971) & (0.00950) & (0.00990) & (0.00954) \\
\hline
$10^{-2}$ & 0.18817  &  0.13054  &  0.09701  &  {\bf 0.04566}  &  6.43301  &  0.05337  &  0.10362  &  \underline{0.04592} \\
& (0.01896) & (0.03597) & (0.01584) & (0.00636) & (0.10094) & (0.00645) & (0.00648) & (0.00636) \\
\hline
$10^{-1}$ & 0.16236  &  0.16727  &  0.08413  &  {\bf 0.05843}  &  6.22776  &  0.06196  &  0.10845  &  \underline{0.05859} \\
& (0.01630) & (0.03973) & (0.01257) & (0.01052) & (0.12413) & (0.01047) & (0.01117) & (0.01052) \\
\hline
$10^{-0.5}$ & 0.28054  &  0.27263  &  0.16329  &  {\bf 0.14208}  &  6.45716  &  0.14896  &  0.18881  &  \underline{0.14232} \\
& (0.02537) & (0.06909) & (0.01635) & (0.01475) & (0.12854) & (0.01469) & (0.01497) & (0.01476) \\
\hline
$1$ & 1.05281  &  1.09944  &  1.01384  &  {\bf 0.95356}  &  6.77452  &  1.01920  &  \underline{0.95611}  &  1.04263 \\
& (0.02046) & (0.02191) & (0.00997) & (0.00905) & (0.10774) & (0.00877) & (0.00849) & (0.00715) \\
\hline
$10^{0.5}$ & 8.37061  &  9.48056  &  9.67063  &  {\bf 4.71207}  &  \underline{7.36306}  &  9.67729  &  7.55671  &  9.57846 \\
& (0.02177) & (0.06041) & (0.01659) & (0.05336) & (0.13437) & (0.01424) & (0.03817) & (0.00842) \\
\hline
$10$ & 72.86179  &  105.5150  &  99.00799  &  {\bf 8.31751}  &  \underline{8.80645}  &  97.58522  &  39.24465  &  104.3057 \\
& (0.12036) & (0.37105) & (0.82643) & (0.12924) & (0.13857) & (0.90360) & (0.25660) & (0.01449) \\
\hline
$10^{1.5}$ & 501.0324  &  995.3871  &  996.3600  &  {\bf 8.88848}  &  \underline{8.93730}  &  996.3600  &  71.83093  &  913.9395 \\
 & (0.40598) & (0.97373) & (0.00678) & (0.13554) & (0.13665) & (0.00678) & (0.42595) & (0.21334) \\
\hline
$10^2$ & 2015.788  &  8484.242  &  7050.887  &  {\bf 9.07261}  &  \underline{9.07689}  &  7228.744  &  21.54672  &  4691.845 \\
& (1.37168) & (1.37785) & (12.25646) & (0.12175) & (0.12132) & (13.23414) & (0.24329) & (2.66192) \\
\hline
\end{tabular}
}
\caption{The performance comparison table for each simulation setting in Section \ref{highsimulsection} with $(n,p)=(200,300)$ and $\sigma=3$. The values are the mean squared same-X errors averaged over 50 independent replications.
 The numbers in parentheses are normalized sample standard deviations. The column labels are defined in Section \ref{highsimulsection}. 
 Boldface indicates the best model. Underlined is the best model within our proposed estimators \textbf{2n-Rep1} to \textbf{2n-Rep4}.
}
\label{table:highdim.samex.3}
\end{table}

\begin{table}[h!]
\centering
\resizebox{\columnwidth}{!}{%
\begin{tabular}{ |c||c|c|c|c|c|c|c|c|c|c|  }
 \hline
 \multicolumn{11}{|c|}{Performance comparison table $(n,p)=(200,300)$} \\
 \hline
 $p$ $\backslash$ $\tau$ & 0 & $10^{-4}$ & $10^{-2}$ & $10^{-1}$ & $10^{-0.5}$ & 1 & $10^{0.5}$ & $10$ & $10^{1.5}$ & $10^2$ \\
 \hline
$\sigma=1$ & {\bf 0.00620} & {\bf 0.00237} & {\bf 0.00494} & {\bf 0.01363} & {\bf 0.09847} & {\bf 0.48945} & {\bf 0.91546} & {\bf 1.00967} & 1.15949 & {\bf 0.98135} \\
FVS2 & (31622.78) & (31622.78) & (5623.413) & (123.0269) & (9.33254) & (1.02329) & (0.10471) & (0.00977) & (0.00115) & (0.00019) \\
\hline
$\sigma=1$ & 0.00638 & 0.00255 & 0.00513 & 0.01373 & 0.10376 & 0.66386 & 0.96319 & 1.01786 & {\bf 1.14341} & 0.98173 \\
Ridge & (31622.78) & (31622.78) & (31622.78) & (31622.78) & (5623.413) & (239.883) & (13.18257) & (1.148154) & (2.13796) & (0.89125) \\
\hline
\hline
$\sigma=2$ & {\bf 0.01499} & {\bf 0.01633} & {\bf 0.01702} & {\bf 0.03516} & {\bf 0.11523} & {\bf 0.76342} & {\bf 2.79710} & {\bf 3.78894} & {\bf 4.01031} & {\bf 3.95205} \\
FVS2 & (31622.78) & (31622.78) & (27542.29) & (331.1311) & (37.15352) & (4.36516) & (0.42658) & (0.03715) & (0.00355) & (0.00028) \\
\hline
$\sigma=2$ & 0.01783 & 0.01892 & 0.01984 & 0.03759 & 0.11682 & 0.84901 & 3.43498 & 3.86065 & 4.02001 & 3.95274 \\
Ridge & (31622.78) & (31622.78) & (31622.78) & (31622.738) & (31622.78) & (5128.614) & (79.43282) & (14.12538) & (2.13796) & (0.02692) \\
\hline
\hline
$\sigma=3$ & {\bf 0.04586} & {\bf 0.04670} & {\bf 0.03703} & {\bf 0.06753} & {\bf 0.12454} & {\bf 1.11509} & {\bf 4.63284} & {\bf 8.29820} & {\bf 8.93890} & {\bf 8.80025} \\
FVS2 & (31622.78) & (31622.78) & (8709.636) & (851.1380) & (85.11380) & (9.12011) & (0.93325) & (0.08511) & (0.00977) & (0.00087) \\
\hline
$\sigma=3$ & 0.05893 & 0.05933 & 0.05001 & 0.08110 & 0.13283 & 1.15278 & 6.49951 & 8.72412 & 9.01420 & 8.80203 \\
Ridge & (31622.78) & (31622.78) & (31622.78) & (31622.78) & (31622.78) & (14125.38) & (602.5596) & (27.54229) & (5.01187) & (3.38844) \\
\hline
\end{tabular}
}
\caption{The values are the minimal average same-X loss values obtained from each estimator (see \eqref{lam.gr.FVS} and \eqref{lam.gr.Ridge}) within the set of $\lambda\in\{10^{-7+0.01k}:j=0,1,\ldots,1150\}$ from the simulation settings in Section \ref{highsimulsection} when $(n,p)=(200,300)$.
The numbers in parentheses are the values of $\lambda$ that achieve the minimum for each setting and each estimator. Boldface indicates the best model for each $\tau$ and $\sigma$.}
\label{table:graph table 4}
\end{table}

\clearpage

\subsection{Additional tables for Section \ref{fullrankmtx.high}}\label{appd:fullrank.high}

\begin{table}[h!]
\centering
\resizebox{\columnwidth}{!}{%
\begin{tabular}{ |c||c|c|c|c|c|  }
 \hline
 \multicolumn{6}{|c|}{Performance comparison table when $\sigma=1$} \\
 \hline
 $(s,{\rm Coding~\#})$ & 2n-Or & Ridge & LASSO & 2n-G & 2n-Rep1\\
 \hline
 $s=0.025$ & 0.41999  &  0.34635  &  {\bf 0.13778}  &  0.49072  &  \underline{0.55543} \\
 Coding-1 & (0.00502) & (0.00589) & (0.00727) & (0.00962) & (0.00936) \\
 \hline
 $s=0.025$ & {\bf 0.41999}  &  1.00610  &  1.14077  &  0.49072  & \underline{0.52888} \\
 Coding-2 & (0.00502) & (0.01917) & (0.03211) & (0.00962) & (0.00898) \\
 \hline
 $s=0.05$ & 0.55923  &  0.46571  &  \bf{0.16670}  &  0.71573  &  0.65167  \\
 Coding-1 & (0.00665) & (0.00816) & (0.00730) & (0.01901) & (0.00759) \\
 \hline
 $s=0.05$ & {\bf 0.55923}  &  0.64650  &  0.80326  &  0.71573  &  \underline{0.64631} \\
 Coding-2 & (0.00665) & (0.01711) & (0.03086) & (0.01901) & (0.00757) \\
 \hline
 $s=0.1$ & 0.88849  &  0.73842  &  {\bf 0.50612}  &  1.02179  &  \underline{0.92288} \\
 Coding-1 & (0.01241) & (0.01289) & (0.01148) & (0.03451) & (0.01333) \\
 \hline
 $s=0.1$ & {\bf 0.88849}  &  1.57656  &  1.62821  &  1.02179  &  \underline{0.92044} \\
 Coding-2 & (0.01241) & (0.03464) & (0.04985) & (0.03451) & (0.01331) \\
 \hline
 $s=0.2$ & 0.94443  &  0.86450  &  {\bf 0.68938}  &  1.00405  &  \underline{0.95970} \\
 Coding-1 & (0.01276) & (0.01517) & (0.01294) & (0.01984) & (0.01253) \\
 \hline
 $s=0.2$ & \bf{0.94443}  &  1.18902  &  1.31990  &  1.00405  &  \underline{0.95948} \\
 Coding-2 & (0.01276) & (0.02706) & (0.05416) & (0.01984) & (0.01253) \\
 \hline
 $s=0.3$ & 0.95895  &  {\bf 0.87565}  &  0.87934  &  0.98095  &  \underline{0.96919} \\
 Coding-1 & (0.01285) & (0.01319) & (0.01666) & (0.01405) & (0.01297) \\
 \hline
 $s=0.3$ & {\bf 0.95895}  &  1.24615  &  1.11530  &  0.98095  &  \underline{0.96907} \\
 Coding-2 & (0.01285) & (0.02910) & (0.02987) & (0.01405) & (0.01297) \\
 \hline
\end{tabular}
}
 \caption{The performance comparison table for each simulation setting in Section \ref{fullrankmtx.high} when $\sigma=1$. The values are the mean squared same-X errors averaged over 50 independent replications.
 The numbers in parentheses are normalized sample standard deviations. The column labels are defined in Section \ref{highsimulsection}. 
 Boldface indicates the best model. Underlined is our main proposed estimator: \textbf{2n-Rep1}.}
\label{table:high dim XT table1}
\end{table}

\begin{table}[h!]
\centering
\resizebox{\columnwidth}{!}{%
\begin{tabular}{ |c||c|c|c|c|c|c|  }
 \hline
 \multicolumn{7}{|c|}{Performance comparison table when $\sigma=3$} \\
 \hline
 $(s,{\rm Coding~\#})$ & 2n-Or & Ridge & LASSO & 2n-G & 2n-Rep2 & 2n-Rep3 \\
 \hline
 $s=0.01$ & {\bf 0.11933}  &  0.24273  &  0.16385  &  0.14141   &  0.12694  &  \underline{0.17240} \\
 Coding-1 & (0.01243) & (0.02976) & (0.02155) & (0.01606)  & (0.01302) & (0.01246) \\
 \hline
 $s=0.01$ & {\bf 0.11933}  &  0.39756  &  0.31803  &  0.14141  & 0.12621  &  \underline{0.16576}  \\
 Coding-2 & (0.01243) & (0.03321) & (0.03025) & (0.01606) & (0.01293) & (0.01199)  \\
 \hline
 $s=0.02$ & \bf{0.35239}  &  0.39643  &  0.42977  &  0.37503  &  0.35854  &  \underline{0.37686}  \\
 Coding-1 & (0.01035) & (0.01197) & (0.02694) & (0.01212) & (0.01053) & (0.01008) \\
 \hline
 $s=0.02$ & {\bf 0.35239}  &  0.50265  &  0.51940  &  0.37503  & 0.35827  &  \underline{0.37556}  \\
 Coding-2 & (0.01035) & (0.02094) & (0.02899) & (0.01212) & (0.01052) & (0.0098) \\
 \hline
 $s=0.03$ & {\bf 0.70534}  &  0.71936  &  0.88481  &  0.74575  &  0.72611  &  \underline{0.71032} \\
 Coding-1 & (0.01334) & (0.02562) & (0.05025) & (0.01639) & (0.01363) & (0.01358) \\
 \hline
 $s=0.03$ & {\bf 0.70534}  &  1.05775  &  1.14859  &  0.74575 & 0.72954  &  \underline{0.70970} \\
 Coding-2 & (0.01334) & (0.05330) & (0.05018) & (0.01639) & (0.01383) & (0.01400) \\
 \hline
 $s=0.04$ & {\bf 0.56631}  &  0.59650  &  0.68032  &  0.59733  &  0.57926  & \underline{0.57477} \\
 Coding-1 & (0.01107) & (0.02028) & (0.03621) & (0.01139) & (0.01099) & (0.01162) \\
 \hline
 $s=0.04$ & {\bf 0.56631}  &  0.74298  &  0.88815  &  0.59733 &  0.58074  &  \underline{0.57244}  \\
 Coding-2 & (0.01107) & (0.02746) & (0.06009) & (0.01139) & (0.01093) & (0.01130)  \\
 \hline
 $s=0.05$ & 1.78976  &  {\bf 1.60166}  &  1.83978  &  1.94710  &  1.91890  &  \underline{1.85974} \\
 Coding-1 & (0.01647) & (0.05570) & (0.04532) & (0.02877) & (0.02369) & (0.01577) \\
 \hline
 $s=0.05$ & 1.78976  &  {\bf 1.72413}  &  2.22285  &  1.94710 &  1.92666  &  \underline{1.86978}  \\
 Coding-2 & (0.01647) & (0.04972) & (0.05914) & (0.02877) & (0.02534) & (0.01907) \\
 \hline
\end{tabular}
}
 \caption{The performance comparison table for each simulation setting in Section \ref{fullrankmtx.high} when $\sigma=3$. The values are the mean squared same-X errors averaged over 50 independent replications.
 The numbers in parentheses are normalized sample standard deviations. The column labels are defined in Section \ref{highsimulsection}. 
 Boldface indicates the best model. Underlined is our main proposed estimator: \textbf{2n-Rep3}.}
\label{table:high dim XT table2}
\end{table}

\clearpage

\section{Shrinking toward the fitted values of a submodel} \label{submodelshrink}
In Section \ref{sec:main.method}--\ref{tuningselection}, we developed our new shrinkage method
for linear regression with fitted values that are 
invariant to invertible linear transformations of
the design matrix $X$.  Its shrinkage target was
the fitted values of the intercept-only model $P_1Y$.
We can generalize this so that the shrinkage target
is $P_{X_0} Y$, where $X_0$ is the design matrix for a submodel formed from a proper subset of the columns of $X$, 
e.g. $X_0 = 1_n$ (the first column of $X$) as it was previously.  We consider the following generalization of
\eqref{gencasesvec}:
\begin{align*}
\argmin_{  b \in \mathbb{R}^p}  \left\{ \gamma\|Y - Xb\|^2 + (1-\gamma) \|Xb - P_{X_0} Y\|^2 \right\},  
\end{align*}
where $\gamma\in[0,1]$.  
This generalization could 
be useful when $\|  \mu - P_{X_0}  \mu\|^2$ is substantially smaller than $\|  \mu - P_1  \mu\|^2$, where $ \mu=X \beta$.
For example, a practitioner may know that
some predictors have large effect sizes, so these predictors would be encoded in $X_0$.

All of the results obtained for the special
case that $X_0=1_n$ also hold in this general case with
$P_1Y$ replaced by $P_{X_0} Y$; and with ${\rm rank}(X)-1$
replaced by ${\rm rank}(X) - {\rm rank}(X_0)$.  The proofs follow by making these replacements in the proofs from 
the special case that $X_0=1_n$.  

We continue by 
stating these generalized results.
The generalized fitted-value shrinkage method's fitted values are 
$$
X\tilde{ \beta}^{(\gamma)} = \gamma P_XY + (1-\gamma) P_{X_0} Y.
$$
\begin{prop}\label{emsefvs2general}
For all $(n, p) \in \{1,2,\ldots\}\times \{1,2,\ldots\}$,
\begin{equation}\label{MSEbound-general}
\mathbb{E}\|X\tilde{ \beta}^{(\gamma)}- X  \beta\|^2  = \sigma^2 \left\{ \gamma^2 {\rm rank}(X) + (1-\gamma^2){\rm rank}(X_0) \right\} + (\gamma-1)^2 \|  \mu-P_{X_0}  \mu\|^2.
\end{equation}
\end{prop}
\noindent
The right side of \eqref{MSEbound-general} is minimized when $\gamma=\tilde\gamma_{\rm opt}$, 
where
\begin{align}\label{gamopt.submodel}
\tilde \gamma_{\rm opt}=\frac{\|  \mu - P_{X_0}  \mu\|^2}
{\sigma^2 ({\rm rank}(X) -{\rm rank}(X_0)) + \| \mu - P_{X_0}  \mu\|^2}.
\end{align}
We can estimate $\tilde\gamma_{\rm opt}$ in low dimensions with
\begin{align}\label{gamopt.in.submodel}
\tilde\gamma = (1-1/\tilde F) \cdot 1(\tilde F > 1),
\end{align}
where $\tilde F$ is the F-statistic for comparing the submodel $X_0$ 
to the full design matrix model $X$:
\begin{align}
\tilde F & = \frac{(\|Y - P_{X_0}Y\|^2 - \|Y - P_XY\|^2)/({\rm rank}(X)-{\rm rank}(X_0))}{\|Y - P_XY\|^2/(n-p)}\nonumber\\
 & = \frac{\|P_XY - P_{X_0}Y\|^2}{\hat\sigma^2({\rm rank}(X)-{\rm rank}(X_0))}.\label{fstatsub}
\end{align}
\noindent
Recall that $r={\rm rank}(X)$ and let $r_0 ={\rm rank}(X_0)$.
\begin{prop}\label{consistencysub}
Assume that the data-generating
model in \eqref{gencasesvec} is correct, that the errors
follows a distribution with finite fourth moment, and
$r-r_0 \geq 1$. 
If $p/n\rightarrow\tau\in [0,1)$ as $n\rightarrow \infty$ and 
either $(r-r_0)\rightarrow \infty$
or $\Tilde \delta^2\rightarrow \infty$,
then $\tilde\gamma - \tilde\gamma_{\rm opt} \rightarrow_P 0$.
\end{prop}

\begin{prop}\label{convergenceratesub}
Suppose that the assumptions of Proposition \ref{consistencysub} hold, that the errors
in \eqref{gencasesvec} are Gaussian, 
and that $r-r_0 \ge 5$ is nondecreasing as $n\rightarrow\infty$.  Then, 
\begin{equation}\label{equaitonrate}
  \tilde\gamma-\tilde\gamma_{\rm opt} =
    \begin{cases}
    O_P((r-r_0)^{-1/2}) & \text{if $\Tilde \delta^2=O(r-r_0)$}\\
O_P((\Tilde \delta^2/(r-r_0))^{-3/4})+O_P(n^{-1/2}(\Tilde \delta^2/(r-r_0))^{-1/2}) & \text{if $(r-r_0) \rightarrow \infty$ and $r-r_0=o(\Tilde \delta^2)$}\\
     O_P((\Tilde \delta^2)^{-3/4})+O_P(n^{-1/2}(\Tilde \delta^2)^{-1/2}) & \text{if $r-r_0=O(1)$.}\\
    \end{cases}       
\end{equation}
\end{prop}

\subsection{Additional simulations: submodel shrinkage}\label{appd:add.simulation}
We conduct additional simulations which is again the linear regression 
subjects model \eqref{gencasesvec} with $n=100,\,p=75$. 
The design matrix $X$ has ones in its first column and independent
draws from $N_{p-1}(0,  \Sigma)$ in the remaining entries on each row,
where $\Sigma_{jk}=0.5^{|j-k|}$.
However, in this section we explore, possible scenarios where the shrinkage towards known submodel (or submatrix) $X_0$ has an advantage over the shrinkage towards the intercept only model (see Section \ref{submodelshrink} for details).
Let $X_0$ denote the sub-matrix of $X$ which consists of the first $p_0$ columns of $X$.
We denote $X=(X_0,X_1)$.
We consider $p_0\in\{5,10,15,20,25\}$.
And, we update $X_1$ by $X_1^*=Q_{X_0}X_1=(I-P_{X_0})X_1$, and further replace $(X_0,X_1)$ with $(X_0,X_1^*)$ for $X$.
We conduct this orthogonalization procedure to decompose terms easier (see \eqref{decomposed} below).  
Then we create
\begin{align*}
&\beta_0=X_0^{-}(1_{p_0}+\tau_0 Z),\\
&\beta_1=\tau_1(X_1^*)^{-}Z',
\end{align*}
where $Z,\,Z'\sim N_n(0,I)$ and $Z,\,Z'$ are independent of each other.
In this setting, we have
\begin{align}
\|\mu-P_{X_0}\mu\|^2&=\|X_1^*\beta_1\|^2=\tau_1^2\|P_{X_1^*}Z'\|^2,\nonumber\\
\|\mu-P_1\mu\|^2&=\|(P_{X_0}-P_1)(X_0\beta_0)\|^2+\|X_1^*\beta_1\|^2\nonumber\\
&=\tau_0^2\|P_{X_0^*}Z\|^2+\tau_1^2\|P_{X_1^*}Z'\|^2,\label{decomposed}
\end{align}
where $X_0^*=(P_{X_0}-P_1)X_0$.
We recall that the oracle tuning parameters \eqref{gamopt}, \eqref{gamopt.submodel} are given by
\begin{align*}
&\gamma_{\rm opt}=\frac{\| \mu - P_1  \mu\|^2}
{\sigma^2 ({\rm rank}(X) -1) + \|  \mu - P_1  \mu\|^2},\\
&\tilde \gamma_{\rm opt}=\frac{\|  \mu - P_{X_0}  \mu\|^2}
{\sigma^2 ({\rm rank}(X) -{\rm rank}(X_0)) + \| \mu - P_{X_0}  \mu\|^2}.
\end{align*}
Thus, since $\sigma^2=1$, we have
\begin{align*}
\EE\left\|X\hat\beta^{(\gamma_{\rm opt})}-X\beta\right\|^2&=\sigma^2+\frac{\sigma^2({\rm rank}(X)-1)\|\mu-P_1\mu\|^2}{\sigma^2({\rm rank}(X)-1)+\|\mu-P_1\mu\|^2}\\ 
&=1+\frac{({\rm rank}(X)-1)\|\mu-P_1\mu\|^2}{({\rm rank}(X)-1)+\|\mu-P_1\mu\|^2},\\
\EE\left\|X\hat\beta^{(\tilde \gamma_{\rm opt})}-X\beta\right\|^2&=\sigma^2 {\rm rank}(X_0)+\frac{\sigma^2({\rm rank}(X)-{\rm rank}(X_0))\|\mu-P_{X_0}\mu\|^2}{\sigma^2({\rm rank}(X)-{\rm rank}(X_0))+\|\mu-P_{X_0}\mu\|^2}\\
&={\rm rank}(X_0)+\frac{({\rm rank}(X)-{\rm rank}(X_0))\|\mu-P_{X_0}\mu\|^2}{({\rm rank}(X)-{\rm rank}(X_0))+\|\mu-P_{X_0}\mu\|^2}.
\end{align*}
We control $\tau_1$ to satisfy the following equation:
\begin{align*}
\|\mu-P_{X_0}\mu\|^2&=\tau_1^2\|P_{X_1^*}Z'\|^2=R_1({\rm rank}(X)-{\rm rank}(X_0)),\\
\|\mu-P_1\mu\|^2&=\tau_0^2\|P_{X_0^*}Z\|^2+ \tau_1^2\|P_{X_1^*}Z'\|^2\\
&=\tau_0^2\|P_{X_0^*}Z\|^2+R_1({\rm rank}(X)-{\rm rank}(X_0))\\
&=R_2({\rm rank}(X)-1).
\end{align*}
Thus, $\EE\left\|X\hat\beta^{(\gamma_{\rm opt})}-X\beta\right\|^2>\EE\left\|X\hat\beta^{(\tilde \gamma_{\rm opt})}-X\beta\right\|^2$ holds when 
\begin{align*}
\frac{{\rm rank}(X)-1}{R_2+1}>{\rm rank}(X_0)+\frac{{\rm rank}(X)-{\rm rank}(X_0)}{R_1+1}.
\end{align*}
We vary $R_1\in\{10^{-1},10^{-1/2},1,10^{1/2},10\}$ and $R_2\in\{10^{-1/2},1,10^{1/2},10,10^{3/2}\}$ while keeping $R_2>R_1$.
This implies that there are 15 pairs of $(R_1,R_2)$.
Smaller $R_1$ implies smaller (expected) realization of $\|\mu-P_{X_0}\mu\|^2$.
Smaller $R_2$ indicates smaller (expected) realization of $\|\mu-P_{1}\mu\|^2-\|\mu-P_{X_0}\mu\|^2$.
We include \textbf{2n-Or}, \textbf{2n-Orsb}, \textbf{2n-Es}, \textbf{2n-Es95}, \textbf{2n-Essb}, \textbf{2n-Essb95}, \textbf{R}, \textbf{L}, \textbf{R-sb}, \textbf{L-sb} as the competitors.
In the set of competitors, \textbf{2n-Orsb} exploits the oracle tuning parameter for the submodel shrinkage towards $X_0$ (see \eqref{gamopt.submodel} for details). 
\textbf{2n-Essb95} is based on the estimator for $\tilde\gamma_{\rm opt}$ through 
\begin{align*}
\hat{\tilde\gamma}_{95}=(1-1/\tilde F) \cdot 1(\tilde F\geq \tilde f_{0.95}),    
\end{align*}
where $\tilde f_{0.95}$ is the 0.95 quantile of the central F-distribution with degrees of freedom ${\rm rank}(X)-{\rm rank}(X_0)$ and $n-{\rm rank}(X)$ (see \eqref{gamopt.in.submodel} for the definition of $\tilde F$).
And, \textbf{R-sb} and \textbf{L-sb} does not shrink the coefficients according to $X_0$, which are fitted by \texttt{glmnet} as for \textbf{R} and \textbf{L}.

The results are shown in Table \ref{table:submodel table}--\ref{table:submodel table5}.
In general, \textbf{2n-Orsb} was the best among the candidates except for few cases where \textbf{2n-Or} was the best.
When $R_2=10^{1/2}R_1$ for each $R_1$ and $p$, \textbf{R-sb} followed the next. And, \textbf{2n-Essb95} was much competitive than \textbf{2n-Essb}.
However, except for these cases, \textbf{2n-Essb} was generally the best estimator among the candidates set without the oracle submodel shrinkage \textbf{2n-Orsb}.

\begin{table}[h!]
\centering
\resizebox{\columnwidth}{!}{%
\begin{tabular}{ |c||c|c|c|c|c|c|c|c|c|c|c|  }
 \hline
 \multicolumn{12}{|c|}{Performance comparison table ($p_0=5$)} \\
 \hline
 $(R_1,R_2)$ &  OLS & 2n-Or & 2n-Orsb & 2n-Es & 2n-Essb & 2n-Es95 & 2n-Essb95 & R & R-sb & L & L-sb\\
 \hline
$(10^{-1},10^{-1/2})$ & 0.75315 & 0.18980 & {\bf 0.11566} & 0.22162 & \underline{0.13632} & 0.25021 & \underline{0.13246} & 0.23288 & 0.13411 & 0.20668 & 0.14579 \\
& (0.01806) & (0.00322) & (0.00450) & (0.00632) & (0.00749) & (0.00399) & (0.00708) & (0.00801) & (0.00820) & (0.01131) & (0.01161)\\
\hline
$(10^{-1},1)$ & 0.75391 & 0.37009 & {\bf 0.11780} & 0.40247 & \underline{0.14407} & 0.50593 & \underline{0.13321} & 0.49804 & 0.14304 & 0.28169 & 0.14074 \\
& (0.01409) & (0.00671) & (0.00400) & (0.00987) & (0.00858) & (0.02553) & (0.00837) & (0.01675) & (0.01016) & (0.01163) & (0.00779)\\
\hline
$(10^{-1},10^{1/2})$ & 0.74555 & 0.55969 & {\bf 0.11745} & 0.58596 & \underline{0.13973} & 0.58596 & \underline{0.12298} & 0.46691 & 0.13354 & 0.23834 & 0.13071 \\
& (0.01941) & (0.01321) & (0.00465) & (0.01470) & (0.00663) & (0.01470) & (0.00450) & (0.01297) & (0.00650) & (0.01203) & (0.00653)\\
\hline
$(10^{-1},10)$ & 0.74243 & 0.66124 & {\bf 0.11471} & 0.66965 & \underline{0.14174} & 0.66965 & \underline{0.13262} & 0.47052 & 0.12693 & 0.19610 & 0.13456 \\
& (0.01805) & (0.01599) & (0.00478) & (0.01557) & (0.00858) & (0.01557) & (0.00823) & (0.01317) & (0.00550) & (0.00965) & (0.00643) \\
\hline
$(10^{-1},10^{3/2})$ & 0.75591 & 0.73121 & {\bf 0.11318} & 0.73492 & \underline{0.15318} & 0.73492 & \underline{0.13518} & 0.62439 & 0.12744 & 0.23151 & 0.13621\\
& (0.01977) & (0.01879) & (0.00419) & (0.01881) & (0.00952) & (0.01881) & (0.00868) & (0.01493) & (0.00585) & (0.00887) & (0.00706)\\
\hline
$(10^{-1/2},1)$ & 0.74920 & 0.38666 & {\bf 0.22035} & 0.42746 & \underline{0.25591} & 0.53655 & \underline{0.28396} & 0.39515 & 0.26311 & 0.31882 & 0.28338 \\
& (0.01878) & (0.00856) & (0.00644) & (0.01359) & (0.01025) & (0.02627) & (0.00859) & (0.01157) & (0.00701) & (0.01086) & (0.00915)\\
\hline
$(10^{-1/2},10^{1/2})$ & 0.73399 & 0.55839 & {\bf 0.20824} & 0.57860 & \underline{0.23963} & 0.57860 & \underline{0.26789} & 0.52743 & 0.26231 & 0.36967 & 0.27402 \\
& (0.01736) & (0.01335) & (0.00431) & (0.01347) & (0.00601) & (0.01347) & (0.00423) & (0.01421) & (0.00519) & (0.01252) & (0.00577)\\
\hline
$(10^{-1/2},10)$ & 0.77899 & 0.70325 & {\bf 0.22648} & 0.71316 & \underline{0.26098} & 0.71316 & \underline{0.27987} & 0.66012 & 0.26736 & 0.41049 & 0.27625 \\
& (0.01652) & (0.01447) & (0.00505) & (0.01474) & (0.00769) & (0.01474) & (0.00663) & (0.01711) & (0.00494) & (0.01119) & (0.00509)\\
\hline
$(10^{-1/2},10^{3/2})$ & 0.72380 & 0.70648 & {\bf 0.21560} & 0.70942 & \underline{0.26006} & 0.70942 & \underline{0.28211} & 0.63569 & 0.25425 & 0.35122 & 0.27764 \\
& (0.01751) & (0.01671) & (0.00490) & (0.01687) & (0.00875) & (0.01687) & (0.00733) & (0.01505) & (0.00808) & (0.01075) & (0.00882)\\
\hline
$(1,10^{1/2})$ & 0.75791 & 0.56562 & {\bf 0.39208} & 0.58804 & \underline{0.42621} & 0.58804 & \underline{0.51043} & 0.72365 & 0.65020 & 0.69557 & 0.71052 \\
& (0.01772) & (0.01145) & (0.00722) & (0.01288) & (0.00904) & (0.01288) & (0.02242) & (0.05917) & (0.01349) & (0.01873) & (0.01158)\\
\hline
$(1,10)$ & 0.76410 & 0.69973 & {\bf 0.41118} & 0.71457 & \underline{0.45079} & 0.71457 & \underline{0.53034} & 0.71886 & 0.70649 & 0.78686 & 0.73697 \\
& (0.01491) & (0.01533) & (0.00927) & (0.01690) & (0.01357) & (0.01690) & (0.02386) & (0.01712) & (0.00936) & (0.01987) & (0.00834) \\
\hline
$(1,10^{3/2})$ & 0.71352 & 0.69691 & {\bf 0.38838} & 0.69815 & \underline{0.42070} & 0.69815 & \underline{0.52802} & 0.65320 & 0.65629 & 0.69988 & 0.71256 \\
& (0.01712) & (0.01669) & (0.00964) & (0.01677) & (0.01138) & (0.01677) & (0.02552) & (0.01473) & (0.01180) & (0.02147) & (0.01115) \\
\hline
$(10^{1/2},10)$ & 0.76490 & 0.68891 & {\bf 0.59384} & 0.69567 & \underline{0.61325} & 0.69567 & \underline{0.61325} & 0.75250 & 1.78387 & 1.40155 & 1.85192 \\
& (0.01509) & (0.01380) & (0.01295) & (0.01415) & (0.01343) & (0.01415) & (0.01343) & (0.01562) & (0.08072) & (0.09208) & (0.08544) \\
\hline
$(10^{1/2},10^{3/2})$ & 0.74194 & 0.72271 & {\bf 0.58001} & 0.72723 & \underline{0.60585} & 0.72723 & \underline{0.60585} & 0.71669 & 1.43120 & 1.10135 & 1.43520 \\
& (0.01683) & (0.01659) & (0.01398) & (0.01657) & (0.01454) & (0.01657) & (0.01454) & (0.01642) & (0.08376) & (0.05794) & (0.08676)\\
\hline
$(10,10^{3/2})$ & 0.73628 & 0.71097 & {\bf 0.67316} & 0.71279 & \underline{0.68197} & 0.71279 & \underline{0.68197} & 0.73636 & 0.74792 & 0.77375 & 0.78190 \\
& (0.01843) & (0.01785) & (0.01587) & (0.01779) & (0.01572) & (0.01779) & (0.01572) & (0.01670) & (0.01880) & (0.01965) & (0.02366) \\
\hline
\end{tabular}
}
\caption{The performance comparison table for each simulation setting in Section \ref{appd:add.simulation} labeled with $(R_1,R_2)$ when $p_0=5$. The values are the mean squared same-X errors averaged over 50 independent replications.
The numbers in parentheses are normalized sample standard deviations. The column labels are defined in Section \ref{appd:add.simulation}. 
Boldface indicates the best model. Underlined are our main proposed estimator \textbf{2n-Essb} and its variant \textbf{2n-Essb95}.}
\label{table:submodel table}
\end{table}

\begin{table}[h!]
\centering
\resizebox{\columnwidth}{!}{%
\begin{tabular}{ |c||c|c|c|c|c|c|c|c|c|c|c|  }
 \hline
 \multicolumn{12}{|c|}{Performance comparison table ($p_0=10$)} \\
 \hline
 $(R_1,R_2)$ &  OLS & 2n-Or & 2n-Orsb & 2n-Es & 2n-Essb & 2n-Es95 & 2n-Essb95 & R & R-sb & L & L-sb\\
 \hline
$(10^{-1},10^{-1/2})$ & 0.71448 & 0.18710 & \bf{0.15308} & 0.22542 & \underline{0.17806} & 0.25124 & \underline{0.17416} & 0.23652 & 0.17238 & 0.23617 & 0.18273\\
& (0.01415) & (0.00299) & (0.00623) & (0.00770) & (0.00967) & (0.00600) & (0.00973) & (0.00328) & (0.00727) & (0.00509) & (0.00910) \\
\hline
$(10^{-1},1)$ & 0.76523 & 0.38435 & {\bf 0.17032} & 0.44156 & \underline{0.20419} & 0.55711 & \underline{0.19342} & 0.46333 & 0.18641 & 0.31759 & 0.19389 \\
& (0.02069) & (0.00842) & (0.00736) & (0.01225) & (0.01114) & (0.02393) & (0.00996) & (0.01830) & (0.00937) & (0.01037) & (0.00948)\\
\hline
$(10^{-1},10^{1/2})$ & 0.74863 & 0.55887 & {\bf 0.15692} & 0.58561 & \underline{0.19190} & 0.58561 & \underline{0.18561} & 0.52945 & 0.17885 & 0.35795 & 0.18954 \\
& (0.01725) & (0.01512) & (0.00540) & (0.01618) & (0.01144) & (0.01618) & (0.01152) & (0.01508) & (0.01422) & (0.01589) & (0.01383)\\
\hline
$(10^{-1},10)$ & 0.73031 & 0.67055 & {\bf 0.15354} & 0.68075 & \underline{0.18249} & 0.68075 & \underline{0.17284} & 0.60857 & 0.17334 & 0.37673 & 0.17920 \\
& (0.01927) & (0.01676) & (0.00660) & (0.01698) & (0.01153) & (0.01698) & (0.01135) & (0.01724) & (0.00870) & (0.01413) & (0.00976)\\
\hline
$(10^{-1},10^{3/2})$ & 0.75546 & 0.72214 & {\bf 0.15356} & 0.72233 & \underline{0.17281} & 0.72233 & \underline{0.16116} & 0.68722 & 0.16839 & 0.35528 & 0.17926 \\
& (0.01857) & (0.01747) & (0.00601) & (0.01760) & (0.00777) & (0.01760) & (0.00594) & (0.01745) & (0.00736) & (0.01205) & (0.01133)\\
\hline
$(10^{-1/2},1)$ & 0.74042 & 0.37171 & {\bf 0.25738} & 0.40229 & \underline{0.28719} & 0.47629 & \underline{0.31099} & 0.48069 & 0.29383 & 0.36084 & 0.30896 \\
& (0.01479) & (0.00653) & (0.00739) & (0.00908) & (0.00888) & (0.02326) & (0.00779) & (0.01565) & (0.00831) & (0.01061) & (0.00792)\\
\hline
$(10^{-1/2},10^{1/2})$ & 0.75204 & 0.57483 & {\bf 0.25348} & 0.60166 & \underline{0.30645} & 0.60166 & \underline{0.32709} & 0.56319 & 0.31853 & 0.42377 & 0.32658 \\
& (0.01787) & (0.01230) & (0.00595) & (0.01392) & (0.01122) & (0.01392) & (0.00976) & (0.01406) & (0.01327) & (0.01374) & (0.01290)\\
\hline
$(10^{-1/2},10)$ & 0.72955 & 0.65795 & {\bf 0.25073} & 0.66969 & \underline{0.28660} & 0.66969 & \underline{0.31238} & 0.60863 & 0.29595 & 0.44410 & 0.31184 \\
& (0.01746) & (0.01538) & (0.00712) & (0.01577) & (0.00914) & (0.01577) & (0.00690) & (0.01409) & (0.00715) & (0.01277) & (0.00796)\\
\hline
$(10^{-1/2},10^{3/2})$ & 0.73923 & 0.72187 & {\bf 0.25615} & 0.72610 & \underline{0.29251} & 0.72610 & \underline{0.31651} & 0.68217 & 0.30146 & 0.45704 & 0.31165 \\
& (0.01746) & (0.01720) & (0.00744) & (0.01742) & (0.01076) & (0.01742) & (0.00919) & (0.01709) & (0.00851) & (0.01480) & (0.00745)\\
\hline
$(1,10^{1/2})$ & 0.72922 & 0.56326 & {\bf 0.40524} & 0.58561 & \underline{0.44096} & 0.58561 & \underline{0.51424} & 1.42669 & 0.67139 & 0.83852 & 0.71118 \\
& 0.01738 & 0.01349 & 0.01071 & 0.01426 & 0.01329 & 0.01426 & 0.02213 & 0.11245 & 0.01062 & 0.04332 & 0.00948\\
\hline
$(1,10)$ & 0.76305 & 0.69794 & {\bf 0.42835} & 0.70340 & \underline{0.45715} & 0.70340 & \underline{0.55749} & 0.67675 & 0.61699 & 0.63789 & 0.67237 \\
& (0.01763) & (0.01826) & (0.00774) & (0.01849) & (0.00858) & (0.01849) & (0.02114) & (0.01740) & (0.01371) & (0.01770) & (0.01541)\\
\hline
$(1,10^{3/2})$ & 0.74574 & 0.72950 & {\bf 0.42850} & 0.73108 & \underline{0.46795} & 0.73108 & \underline{0.54806} & 0.71366 & 0.66014 & 0.66518 & 0.71266 \\
& (0.01627) & (0.01608) & (0.01080) & (0.01632) & (0.01319) & (0.01632) & (0.02304) & (0.01601) & (0.01203) & (0.01679) & (0.01137)\\
\hline
$(10^{1/2},10)$ & 0.70736 & 0.65502 & {\bf 0.56871} & 0.65839 & \underline{0.58516} & 0.65839 & \underline{0.58516} & 0.69601 & 0.89156 & 0.86602 & 0.99405 \\
& (0.01675) & (0.01584) & (0.01302) & (0.01606) & (0.01342) & (0.01606) & (0.01342) & (0.01740) & (0.04046) & (0.03474) & (0.06104)\\
\hline
$(10^{1/2},10^{3/2})$ & 0.74739 & 0.72670 & {\bf 0.58911} & 0.72687 & \underline{0.60008} & 0.72687 & \underline{0.60008} & 0.72101 & 1.27025 & 0.90463 & 1.48933 \\
& (0.01764) & (0.01825) & (0.01436) & (0.01810) & (0.01469) & (0.01810) & (0.01469) & (0.01755) & (0.06683) & (0.04513) & (0.07778)\\
\hline
$(10,10^{3/2})$ & 0.77034 & 0.75114 & {\bf 0.71249} & 0.75322 & \underline{0.72214} & 0.75322 & \underline{0.72214} & 0.75868 & 0.76899 & 0.77443 & 0.78353 \\
& (0.01904) & (0.01788) & (0.01714) & (0.01789) & (0.01723) & (0.01789) & (0.01723) & (0.01823) & (0.01881) & (0.01965) & (0.02362)\\
\hline
\end{tabular}
}
\caption{The performance comparison table for each simulation setting in Section \ref{appd:add.simulation} labeled with $(R_1,R_2)$ when $p_0=10$. The values are the mean squared same-X errors averaged over 50 independent replications.
The numbers in parentheses are normalized sample standard deviations. The column labels are defined in Section \ref{appd:add.simulation}. 
Boldface indicates the best model. Underlined are our main proposed estimator \textbf{2n-Essb} and its variant \textbf{2n-Essb95}.}
\label{table:submodel table2}
\end{table}

\begin{table}[h!]
\centering
\resizebox{\columnwidth}{!}{%
\begin{tabular}{ |c||c|c|c|c|c|c|c|c|c|c|c|  }
 \hline
 \multicolumn{12}{|c|}{Performance comparison table ($p_0=15$)} \\
 \hline
 $(R_1,R_2)$ &  OLS & 2n-Or & 2n-Orsb & 2n-Es & 2n-Essb & 2n-Es95 & 2n-Essb95 & R & R-sb & L & L-sb\\
 \hline
$(10^{-1},10^{-1/2})$ & 0.73906 & {\bf 0.18705} & 0.20707 & 0.22212 & \underline{0.23212} & 0.25107 & \underline{0.22176} & 0.23587 & 0.21662 & 0.23868 & 0.22102\\ 
& (0.01651) & (0.00331) & (0.00708) & (0.00701) & (0.00951) & (0.00457) & (0.00851) & (0.00189) & (0.00686) & (0.00548) & (0.00785)\\
\hline
$(10^{-1},1)$ & 0.76501 & 0.38361 & {\bf 0.20750} & 0.42643 & \underline{0.23778} & 0.52284 & \underline{0.22275} & 0.52042 & 0.21732 & 0.37774 & 0.22867 \\
& (0.01998) & (0.00729) & (0.00807) & (0.01020) & (0.01032) & (0.02374) & (0.00992) & (0.01859) & (0.00833) & (0.01501) & (0.01129)\\
\hline
$(10^{-1},10^{1/2})$ & 0.74972 & 0.57755 & {\bf 0.20316} & 0.59343 & \underline{0.22415} & 0.59343 & \underline{0.21595} & 0.56342 & 0.21758 & 0.37132 & 0.22074 \\
& (0.01504) & (0.01216) & (0.00689) & (0.01247) & (0.00818) & (0.01247) & (0.00818) & (0.01266) & (0.00693) & (0.01101) & (0.00832)\\
\hline
$(10^{-1},10)$ & 0.73274 & 0.66642 & {\bf 0.20165} & 0.67168 & \underline{0.22111} & 0.67168 & \underline{0.21289} & 0.63203 & 0.20730 & 0.42650 & 0.22009 \\
& (0.01314) & (0.01337) & (0.00732) & (0.01343) & (0.00800) & (0.01343) & (0.00743) & (0.01213) & (0.00728) & (0.01274) & (0.00968)\\
\hline
$(10^{-1},10^{3/2})$ & 0.74472 & 0.72751 & {\bf 0.20604} & 0.73001 & \underline{0.23384} & 0.73001 & \underline{0.22267} & 0.69919 & 0.21841 & 0.43718 & 0.22180 \\
& (0.01615) & (0.01520) & (0.00972) & (0.01533) & (0.01126) & (0.01533) & (0.01121) & (0.01447) & (0.00999) & (0.01546) & (0.01020) \\
\hline
$(10^{-1/2},1)$ & 0.76213 & 0.38327 & {\bf 0.30416} & 0.41879 & \underline{0.32348} & 0.55251 & \underline{0.35431} & 0.55101 & 0.34674 & 0.50284 & 0.36479 \\
& (0.01394) & (0.00857) & (0.00933) & (0.01322) & (0.00985) & (0.02668) & (0.00830) & (0.01954) & (0.00906) & (0.01552) & (0.01029)\\
\hline
$(10^{-1/2},10^{1/2})$ & 0.75336 & 0.57640 & {\bf 0.29867} & 0.60240 & \underline{0.32607} & 0.60240 & \underline{0.34863} & 0.63712 & 0.33409 & 0.53705 & 0.34570 \\
& (0.01844) & (0.01211) & (0.00839) & (0.01360) & (0.00963) & (0.01360) & (0.00840) & (0.01731) & (0.00801) & (0.01687) & (0.00859)\\
\hline
$(10^{-1/2},10)$ & 0.73869 & 0.66927 & {\bf 0.29872} & 0.67690 & \underline{0.32887} & 0.67690 & \underline{0.35058} & 0.67962 & 0.34317 & 0.52114 & 0.35575 \\
& (0.01588) & (0.01276) & (0.00716) & (0.01290) & (0.00819) & (0.01290) & (0.00779) & (0.01444) & (0.00732) & (0.01146) & (0.00867)\\
\hline
$(10^{-1/2},10^{3/2})$ & 0.71810 & 0.70206 & {\bf 0.27774} & 0.70494 & \underline{0.30972} & 0.70494 & \underline{0.32982} & 0.64439 & 0.31927 & 0.48674 & 0.33309 \\
& (0.01477) & (0.01441) & (0.00616) & (0.01451) & (0.00834) & (0.01451) & (0.00700) & (0.01374) & (0.00626) & (0.01606) & (0.00760)\\
\hline
$(1,10^{1/2})$ & 0.75740 & 0.58188 & {\bf 0.45359} & 0.59729 & \underline{0.47434} & 0.59729 & \underline{0.58024} & 0.64914 & 0.58006 & 0.72244 & 0.63645 \\
& (0.01499) & (0.01246) & (0.00914) & (0.01310) & (0.00969) & (0.01310) & (0.02121) & (0.03592) & (0.01339) & (0.02507) & (0.01440)\\
\hline
$(1,10)$ & 0.75029 & 0.68005 & {\bf 0.45507} & 0.68792 & \underline{0.49030} & 0.68792 & \underline{0.59094} & 0.67875 & 0.59674 & 0.64093 & 0.64777 \\
& (0.01565) & (0.01421) & (0.01066) & (0.01489) & (0.01439) & (0.01489) & (0.02375) & (0.01438) & (0.01450) & (0.01704) & (0.01570) \\
\hline
$(1,10^{3/2})$ & 0.75292 & 0.72711 & {\bf 0.44943} & 0.73090 & \underline{0.48409} & 0.73090 & \underline{0.55935} & 0.71710 & 0.64883 & 0.69908 & 0.71974 \\
& (0.01757) & (0.01713) & (0.01001) & (0.01709) & (0.01297) & (0.01709) & (0.02233) & (0.01669) & (0.01323) & (0.02550) & (0.01221)\\
\hline
$(10^{1/2},10)$ & 0.74617 & 0.70012 & {\bf 0.62514} & 0.70391 & \underline{0.64041} & 0.70391 & \underline{0.64041} & 0.70021 & 1.01440 & 0.77926 & 1.22554 \\
& (0.01773) & (0.01574) & (0.01319) & (0.01555) & (0.01303) & (0.01555) & (0.01303) & (0.01543) & (0.06297) & (0.03676) & (0.08360)\\
\hline
$(10^{1/2},10^{3/2})$ & 0.74053 & 0.71630 & {\bf 0.59748} & 0.72139 & \underline{0.61858} & 0.72139 & \underline{0.61858} & 0.72021 & 0.98535 & 0.76380 & 1.31795 \\
& (0.01512) & (0.01438) & (0.01062) & (0.01421) & (0.01106) & (0.01421) & (0.01106) & (0.01373) & (0.05172) & (0.02666) & (0.08731)\\
\hline
$(10,10^{3/2})$ & 0.75094 & 0.73063 & {\bf 0.70217} & 0.73247 & \underline{0.70687} & 0.73247 & \underline{0.70687} & 0.75793 & 0.75685 & 0.76375 & 0.86866 \\
& (0.01790) & (0.01756) & (0.01765) & (0.01769) & (0.01766) & (0.01769) & (0.01766) & (0.01993) & (0.01957) & (0.02138) & (0.11094)\\
\hline
\end{tabular}
}
\caption{The performance comparison table for each simulation setting in Section \ref{appd:add.simulation} labeled with $(R_1,R_2)$ when $p_0=15$. The values are the mean squared same-X errors averaged over 50 independent replications.
The numbers in parentheses are normalized sample standard deviations. The column labels are defined in Section \ref{appd:add.simulation}. 
Boldface indicates the best model. Underlined are our main proposed estimator \textbf{2n-Essb} and its variant \textbf{2n-Essb95}.}
\label{table:submodel table3}
\end{table}

\begin{table}[h!]
\centering
\resizebox{\columnwidth}{!}{%
\begin{tabular}{ |c||c|c|c|c|c|c|c|c|c|c|c|  }
 \hline
 \multicolumn{12}{|c|}{Performance comparison table ($p_0=20$)} \\
 \hline
 $(R_1,R_2)$ &  OLS & 2n-Or & 2n-Orsb & 2n-Es & 2n-Essb & 2n-Es95 & 2n-Essb95 & R & R-sb & L & L-sb\\
 \hline
$(10^{-1},10^{-1/2})$ & 0.74928 & {\bf 0.19144} & 0.26339 & 0.22581 & \underline{0.28344} & 0.24767 & \underline{0.27525} & 0.22500 & 0.27586 & 0.24043 & 0.27492 \\
& (0.01737) & (0.00349) & (0.00901) & (0.00582) & (0.00962) & (0.00362) & (0.00849) & (0.00385) & (0.00902) & (0.00724) & (0.00886) \\
\hline
$(10^{-1},1)$ & 0.76855 & 0.38990 & {\bf 0.25949} & 0.42639 & \underline{0.27916} & 0.52016 & \underline{0.27192} & 0.50464 & 0.28395 & 0.47496 & 0.27980 \\
& (0.01646) & (0.00854) & (0.00973) & (0.01176) & (0.01166) & (0.02374) & (0.01069) & (0.01423) & (0.01201) & (0.02006) & (0.01060)\\
\hline
$(10^{-1},10^{1/2})$ & 0.72853 & 0.57548 & {\bf 0.24122} & 0.60937 & \underline{0.26693} & 0.63622 & \underline{0.26050} & 0.57106 & 0.26041 & 0.46402 & 0.26154 \\
& (0.01791) & (0.01182) & (0.00966) & (0.01684) & (0.01201) & (0.03798) & (0.01199) & (0.01565) & (0.01026) & (0.02221) & (0.01054)\\
\hline
$(10^{-1},10)$ & 0.76835 & 0.68745 & {\bf 0.26137} & 0.69543 & \underline{0.28189} & 0.69543 & \underline{0.27238} & 0.68579 & 0.27366 & 0.48972 & 0.27623 \\
& (0.01807) & (0.01552) & (0.00953) & (0.01571) & (0.01137) & (0.01571) & (0.01030) & (0.01693) & (0.01072) & (0.01633) & (0.01105)\\
\hline
$(10^{-1},10^{3/2})$ & 0.75399 & 0.73282 & {\bf 0.24057} & 0.73288 & \underline{0.27389} & 0.73288 & \underline{0.25773} & 0.71491 & 0.25589 & 0.50663 & 0.27301 \\
& (0.01855) & (0.01828) & (0.00868) & (0.01805) & (0.01116) & (0.01805) & (0.01096) & (0.01759) & (0.00931) & (0.01492) & (0.01107)\\
\hline
$(10^{-1/2},1)$ & 0.76210 & 0.37796 & {\bf 0.34408} & 0.42062 & \underline{0.37065} & 0.52704 & \underline{0.39400} & 0.65581 & 0.38498 & 0.52304 & 0.38963 \\
& (0.01533) & (0.00734) & (0.00865) & (0.01238) & (0.00961) & (0.02611) & (0.00901) & (0.01401) & (0.00855) & (0.01510) & (0.00844)\\
\hline
$(10^{-1/2},10^{1/2})$ & 0.75397 & 0.57185 & {\bf 0.33758} & 0.59039 & \underline{0.37094} & 0.59039 & \underline{0.38483} & 0.61261 & 0.37348 & 0.56290 & 0.38083 \\
& (0.02062) & (0.01429) & (0.00736) & (0.01408) & (0.00949) & (0.01408) & (0.00866) & (0.01382) & (0.00745) & (0.02240) & (0.00917)\\
\hline
$(10^{-1/2},10)$ & 0.76089 & 0.69668 & {\bf 0.33500} & 0.70945 & \underline{0.37056} & 0.70945 & \underline{0.39108} & 0.69243 & 0.37221 & 0.59079 & 0.38696 \\
& (0.01927) & (0.01724) & (0.01212) & (0.01820) & (0.01425) & (0.01820) & (0.01390) & (0.01773) & (0.01367) & (0.01729) & (0.01333)\\
\hline
$(10^{-1/2},10^{3/2})$ & 0.73794 & 0.71145 & \bf{0.31758} & 0.71333 & \underline{0.34944} & 0.71333 & \underline{0.37090} & 0.70549 & 0.35841 & 0.53875 & 0.37078 \\
& (0.01666) & (0.01537) & (0.00837) & (0.01535) & (0.01015) & (0.01535) & (0.00924) & (0.01588) & (0.00851) & (0.01207) & (0.01045)\\
\hline
$(1,10^{1/2})$ & 0.74388 & 0.57817 & {\bf 0.48466} & 0.60634 & \underline{0.51682} & 0.60634 & \underline{0.59302} & 0.71762 & 0.71339 & 0.82472 & 0.72248 \\
& (0.01709) & (0.01182) & (0.00990) & (0.01491) & (0.01285) & (0.01491) & (0.02135) & (0.02969) & (0.01086) & (0.03340) & (0.01215)\\
\hline
$(1,10)$ & 0.75184 & 0.67143 & {\bf 0.46777} & 0.67615 & \underline{0.49254} & 0.67615 & \underline{0.55636} & 0.69873 & 0.61929 & 0.71318 & 0.67527 \\
& (0.01602) & (0.01466) & (0.01253) & (0.01442) & (0.01400) & (0.01442) & (0.02270) & (0.01610) & (0.01460) & (0.02017) & (0.01639) \\
\hline
$(1,10^{3/2})$ & 0.73253 & 0.71228 & {\bf 0.47656} & 0.71667 & \underline{0.50780} & 0.71667 & \underline{0.60081} & 0.67553 & 0.58019 & 0.61532 & 0.65713 \\
& (0.01632) & (0.01557) & (0.01070) & (0.01568) & (0.01242) & (0.01568) & (0.02049) & (0.01528) & (0.01356) & (0.01579) & (0.01650) \\
\hline
$(10^{1/2},10)$ & 0.76473 & 0.68915 & {\bf 0.62244} & 0.69797 & \underline{0.63873} & 0.69797 & \underline{0.63873}& 0.70046 & 0.79316 & 0.75191 & 0.93870 \\
& (0.01788) & (0.01634) & (0.01466) & (0.01648) & (0.01519) & (0.01648) & (0.01519) & (0.01790) & (0.03791) & (0.02677) & (0.04186) \\
\hline
$(10^{1/2},10^{3/2})$ & 0.74738 & 0.72832 & {\bf 0.61258} & 0.73059 & \underline{0.62813} & 0.73059 & \underline{0.62813} & 0.71779 & 0.69134 & 0.71999 & 0.88120 \\
& (0.01477) & (0.01475) & (0.00981) & (0.01480) & (0.01095) & (0.01480) & (0.01095) & (0.01442) & (0.01732) & (0.01802) & (0.04352) \\
\hline
$(10,10^{3/2})$ & 0.77330 & 0.75047 & \bf{0.72958} & 0.75321 & \underline{0.73795} & 0.75321 & \underline{0.73795} & 0.75058 & 0.76238 & 0.76907 & 0.81230 \\
& (0.01758) & (0.01777) & (0.01634) & (0.01791) & (0.01708) & (0.01791) & (0.01708) & (0.01758) & (0.01914) & (0.02182) & (0.03946) \\
\hline
\end{tabular}
}
\caption{The performance comparison table for each simulation setting in Section \ref{appd:add.simulation} labeled with $(R_1,R_2)$ when $p_0=20$. The values are the mean squared same-X errors averaged over 50 independent replications.
The numbers in parentheses are normalized sample standard deviations. The column labels are defined in Section \ref{appd:add.simulation}. 
Boldface indicates the best model. Underlined are our main proposed estimator \textbf{2n-Essb} and its variant \textbf{2n-Essb95}.}
\label{table:submodel table4}
\end{table}

\begin{table}[h!]
\centering
\resizebox{\columnwidth}{!}{%
\begin{tabular}{ |c||c|c|c|c|c|c|c|c|c|c|c|  }
 \hline
 \multicolumn{12}{|c|}{Performance comparison table ($p_0=25$)} \\
 \hline
 $(R_1,R_2)$ &  OLS & 2n-Or & 2n-Orsb & 2n-Es & 2n-Essb & 2n-Es95 & 2n-Essb95 & R & R-sb & L & L-sb\\
 \hline
$(10^{-1},10^{-1/2})$ & 0.76497 & {\bf 0.19049} & 0.30159 & 0.22281 & \underline{0.32506} & 0.24591 & \underline{0.31473} & 0.24142 & 0.32010 & 0.26057 & 0.32071 \\
& (0.01648) & (0.00296) & (0.01082) & (0.00472) & (0.01019) & (0.00257) & (0.01034) & (0.00379) & (0.01266) & (0.00790) & (0.01203) \\
\hline
$(10^{-1},1)$ & 0.74045 & 0.38097 & {\bf 0.29835} & 0.41574 & \underline{0.31599} & 0.53243 & \underline{0.30714} & 0.56823 & 0.31653 & 0.55727 & 0.32161 \\
& (0.01600) & (0.00706) & (0.00818) & (0.00951) & (0.00835) & (0.02493) & (0.00842) & (0.02177) & (0.00854) & (0.01682) & (0.00971)\\
\hline
$(10^{-1},10^{1/2})$ & 0.77451 & 0.59889 & {\bf 0.31678} & 0.63007 & \underline{0.34244} & 0.63007 & \underline{0.33346} & 0.61851 & 0.32875 & 0.52932 & 0.33869 \\
& (0.01824) & (0.01308) & (0.01019) & (0.01591) & 0.01285 & (0.01591) & (0.01313) & (0.01514) & (0.01071) & (0.01625) & (0.01403)\\
\hline
$(10^{-1},10)$ & 0.74129 & 0.67982 & {\bf 0.30918} & 0.68829 & \underline{0.32619} & 0.68829 & \underline{0.32081} & 0.68705 & 0.32208 & 0.54571 & 0.32555\\
& (0.01492) & (0.01377) & (0.01066) & (0.01453) & (0.01047) & (0.01453) & (0.01082) & (0.01484) & (0.01091) & (0.01415) & (0.01124)\\
\hline
$(10^{-1},10^{3/2})$ & 0.77225 & 0.75080 & {\bf 0.30477} & 0.75409 & \underline{0.32567} & 0.75409 & \underline{0.31816} & 0.73210 & 0.32410 & 0.55371 & 0.33064 \\
& (0.01529) & (0.01496) & (0.00996) & (0.01497) & 0.00979 & (0.01497) & (0.00951) & (0.01443) & (0.00983) & (0.01393) & (0.01057)\\
\hline
$(10^{-1/2},1)$ & 0.75482 & 0.39287 & {\bf 0.37656} & 0.41816 & \underline{0.39349} & 0.53791 & \underline{0.40865} & 0.56735 & 0.38931 & 0.49390 & 0.40518 \\
& (0.01327) & (0.00805) & (0.00917) & (0.00830) & (0.00919) & (0.02512) & (0.00911) & (0.01512) & (0.00952) & (0.01287) & (0.00937)\\
\hline
$(10^{-1/2},10^{1/2})$ & 0.72091 & 0.55371 & \bf{0.36071} & 0.57103 & \underline{0.38912} & 0.57103 & \underline{0.40645} & 0.60223 & 0.38800 & 0.59133 & 0.40854 \\
& (0.01471) & (0.01097) & (0.00947) & (0.01057) & (0.00977) & (0.01057) & (0.00978) & (0.01422) & (0.01012) & (0.01467) & (0.01162) \\
\hline
$(10^{-1/2},10)$ & 0.76138 & 0.69117 & {\bf 0.37563} & 0.69642 & \underline{0.39616} & 0.69642 & \underline{0.41056} & 0.69354 & 0.40719 & 0.65465 & 0.41387 \\
& (0.01834) & (0.01604) & (0.00968) & (0.01593) & (0.01053) & (0.01593) & (0.01010) & (0.01638) & (0.01014) & (0.02152) & (0.00970)\\
\hline
$(10^{-1/2},10^{3/2})$ & 0.76176 & 0.73588 & {\bf 0.37828} & 0.73791 & \underline{0.40630} & 0.73791 & \underline{0.42234} & 0.72998 & 0.42599 & 0.62713 & 0.42156 \\
& (0.01563) & (0.01579) & (0.00972) & (0.01613) & (0.01100) & (0.01613) & (0.01035) & (0.01608) & (0.01180) & (0.01778) & (0.00971)\\
\hline
$(1,10^{1/2})$ & 0.77689 & 0.59444 & {\bf 0.50929} & 0.60615 & \underline{0.53893} & 0.60615 & \underline{0.61130} & 0.64751 & 0.62653 & 0.75841 & 0.65047 \\
& (0.01618) & (0.01291) & (0.01054) & (0.01339) & (0.01294) & (0.01339) & (0.02181) & (0.01518) & (0.01631) & (0.02510) & (0.01475) \\
\hline
$(1,10)$ & 0.73816 & 0.67218 & {\bf 0.48008} & 0.68745 & \underline{0.51483} & 0.68745 & \underline{0.58275} & 0.67501 & 0.57449 & 0.63969 & 0.65603 \\
& (0.01707) & (0.01616) & (0.01084) & (0.01671) & (0.01305) & (0.01671) & (0.01828) & (0.01694) & (0.01318) & (0.01695) & (0.01336) \\
\hline
$(1,10^{3/2})$ & 0.76912 & 0.75153 & \bf{0.51935} & 0.75422 & \underline{0.54752} & 0.75422 & \underline{0.62042} & 0.74042 & 0.66579 & 0.71353 & 0.72791 \\
& (0.01602) & (0.01617) & (0.01183) & (0.01608) & (0.01267) & (0.01608) & (0.02041) & (0.01532) & (0.01250) & (0.01787) & (0.01210) \\
\hline
$(10^{1/2},10)$ & 0.76857 & 0.70204 & {\bf 0.63661} & 0.70744 & \underline{0.64809} & 0.70744 & \underline{0.64809} & 0.73828 & 0.68589 & 0.73377 & 0.78250 \\
& (0.02200) & (0.01865) & (0.01592) & (0.01848) & (0.01610) & (0.01848) & (0.01610) & (0.01917) & (0.01751) & (0.02293) & (0.02508) \\
\hline
$(10^{1/2},10^{3/2})$ & 0.75449 & 0.72511 & {\bf 0.62411} & 0.72885 & \underline{0.63746} & 0.72885 & \underline{0.63746} & 0.73799 & 0.84932 & 0.72779 & 0.91551 \\
& (0.01453) & (0.01348) & (0.01268) & (0.01351) & (0.01253) & (0.01351) & (0.01253) & (0.01415) & (0.03238) & (0.01560) & (0.04969)\\
\hline
$(10,10^{3/2})$ & 0.73746 & 0.71892 & {\bf 0.69821} & 0.72377 & \underline{0.70643} & 0.72377 & \underline{0.70643} & 0.71619 & 0.71048 & 0.75170 & 0.73870 \\
& (0.01681) & (0.01643) & (0.01470) & (0.01637) & (0.01471) & (0.01637) & (0.01471) & (0.01641) & (0.01553) & (0.01772) & (0.01699)\\
\hline
\end{tabular}
}
\caption{The performance comparison table for each simulation setting in Section \ref{appd:add.simulation} labeled with $(R_1,R_2)$ when $p_0=25$. The values are the mean squared same-X errors averaged over 50 independent replications.
The numbers in parentheses are normalized sample standard deviations. The column labels are defined in Section \ref{appd:add.simulation}. 
Boldface indicates the best model. Underlined are our main proposed estimator \textbf{2n-Essb} and its variant \textbf{2n-Essb95}.}
\label{table:submodel table5}
\end{table}

\clearpage

\section{Codes}\label{codes}
\subsubsection{Codes for mtp data cleaning}\label{mtpcodes}
\begin{verbatim}
  hh2<-topo
  hh2<-na.omit(hh2)
  index<-c()
  for(i in 1:ncol(hh2)){
    if (length(unique(hh2[,i]))<30) index<-append(index,i)
  }
  set.seed(40000)
  cut_ind<-sample(nrow(hh2),size = 180,replace = FALSE)
  hh2<-hh2[cut_ind,]
  hh2<-sapply(hh2,as.numeric)
  #removing zero-valued columns
  hhx<-hh2[,-c(index,216,217,218,221,222,223,225:261,267)];hhy<-hh2[,267]
  hhx<-cbind(1,hhx) 
\end{verbatim}

\newpage
\singlespacing
\bibliography{allref} 

\begin{thebibliography}{}

\bibitem[\protect\citeauthoryear{Azriel}{Azriel}{2019}]{Azriel2019}
Azriel, D. (2019).
\newblock {The conditionality principle in high-dimensional regression}.
\newblock {\em Biometrika\/}~{\bf 106\/}(3), 702--707.

\bibitem[\protect\citeauthoryear{Azriel and Schwartzman}{Azriel and
  Schwartzman}{2020}]{azriel.linear.projection}
Azriel, D. and Schwartzman, A. (2020).
\newblock {Estimation of linear projections of non-sparse coefficients in
  high-dimensional regression}.
\newblock {\em Electronic Journal of Statistics\/}~{\bf 14\/}(1), 174 -- 206.

\bibitem[\protect\citeauthoryear{Baltagi}{Baltagi}{2002}]{bal2022}
Baltagi, B.~H. (2002).
\newblock {\em Econometrics}.
\newblock Heidelberg: Springer Berlin.

\bibitem[\protect\citeauthoryear{Barro and Lee}{Barro and
  Lee}{1994}]{barroorigin}
Barro, R. and Lee, J. (1994).
\newblock Data set for a panel of 138 countries.
\newblock {\em discussion paper, NBER\/}~{\bf 138}.

\bibitem[\protect\citeauthoryear{Cannon, Cobb, Hartlaub, Legler, Lock, Moore,
  Rossman, and Witmer}{Cannon et~al.}{2019}]{stat2datapackage}
Cannon, A., Cobb, G., Hartlaub, B., Legler, J., Lock, R., Moore, T., Rossman,
  A., and Witmer, J. (2019).
\newblock Stat2data: Datasets for stat2.
\newblock \url{https://CRAN.R-project.org/package=Stat2Data}.
\newblock R package version 2.0.0.

\bibitem[\protect\citeauthoryear{Cho, Yoo, Im, and Cha}{Cho
  et~al.}{2020}]{bcuref}
Cho, D., Yoo, C., Im, J., and Cha, D. (2020).
\newblock Comparative assessment of various machine learning-based bias
  correction methods for numerical weather prediction model forecasts of
  extreme air temperatures in urban areas.
\newblock {\em Earth and space science\/}~{\bf 7\/}(4).

\bibitem[\protect\citeauthoryear{Choi, Park, and Seo}{Choi
  et~al.}{2012}]{choi2012lasso}
Choi, Y., Park, R., and Seo, M. (2012).
\newblock Lasso on categorical data.

\bibitem[\protect\citeauthoryear{Cook, Forzani, and Rothman}{Cook
  et~al.}{2013}]{cfr13}
Cook, R.~D., Forzani, L., and Rothman, A.~J. (2013).
\newblock Prediction in abundant high-dimensional linear regression.
\newblock {\em Electronic Journal of Statistics\/}~{\bf 7}, 3059--3088.

\bibitem[\protect\citeauthoryear{Copas}{Copas}{1997}]{Copas1997}
Copas, J.~B. (1997).
\newblock Using regression models for prediction: shrinkage and regression to
  the mean.
\newblock {\em Statistical Methods in Medical Research\/}~{\bf 6\/}(2),
  167--183.
\newblock PMID: 9261914.

\bibitem[\protect\citeauthoryear{Cortez and Morais}{Cortez and
  Morais}{2007}]{ff}
Cortez, P. and Morais, A. (2007).
\newblock Efficient forest fire occurrence prediction for developing countries
  using two weather parameters.
\newblock {\em Environmental Science, Computer Science\/}.

\bibitem[\protect\citeauthoryear{Fan and Li}{Fan and Li}{2001}]{fan01}
Fan, J. and Li, R. (2001).
\newblock Variable selection via nonconcave penalized likelihood and its oracle
  properties.
\newblock {\em J. Amer. Statist. Assoc.\/}~{\bf 96\/}(456), 1348--1360.

\bibitem[\protect\citeauthoryear{Feng, Lurati, Ouyang, Robinson, Wang, Yuan,
  and Young}{Feng et~al.}{2003}]{topodata}
Feng, J., Lurati, L., Ouyang, H., Robinson, T., Wang, Y., Yuan, S., and Young,
  S.~S. (2003).
\newblock Predictive toxicology:  benchmarking molecular descriptors and
  statistical methods.
\newblock {\em Journal of Chemical Information and Computer Sciences\/}~{\bf
  43\/}(5), 1463--1470.
\newblock PMID: 14502479.

\bibitem[\protect\citeauthoryear{Frank and Friedman}{Frank and
  Friedman}{1993}]{frank93}
Frank, I.~E. and Friedman, J.~H. (1993).
\newblock A statistical view of some chemometrics regression tools.
\newblock {\em Technometrics\/}~{\bf 35\/}(2), 109--135.

\bibitem[\protect\citeauthoryear{Hastie and Tibshirani}{Hastie and
  Tibshirani}{2004}]{Hastie2004RidgeSVD}
Hastie, T. and Tibshirani, R. (2004, 07).
\newblock {Efficient quadratic regularization for expression arrays}.
\newblock {\em Biostatistics\/}~{\bf 5\/}(3), 329--340.

\bibitem[\protect\citeauthoryear{Hoerl and Kennard}{Hoerl and
  Kennard}{1970}]{hoerl70}
Hoerl, A.~E. and Kennard, R.~W. (1970).
\newblock Ridge regression: {B}iased estimation for nonorthogonal problems.
\newblock {\em Technometrics\/}~{\bf 12}, 55--67.

\bibitem[\protect\citeauthoryear{Hotelling}{Hotelling}{1933}]{hotelling33}
Hotelling, H. (1933).
\newblock Analysis of a complex of statistical variables into principal
  components.
\newblock {\em Journal of Educational Psychology\/}~{\bf 24}, 417--441.

\bibitem[\protect\citeauthoryear{Karthikeyan, Glen, and Bender}{Karthikeyan
  et~al.}{2005}]{mtpdata}
Karthikeyan, M., Glen, R.~C., and Bender, A. (2005).
\newblock General melting point prediction based on a diverse compound data set
  and artificial neural networks.
\newblock {\em Journal of Chemical Information and Modeling\/}~{\bf 45\/}(3),
  581--590.
\newblock PMID: 15921448.

\bibitem[\protect\citeauthoryear{Kleiber and Zeileis}{Kleiber and
  Zeileis}{2008}]{aerpackage}
Kleiber, C. and Zeileis, A. (2008).
\newblock Applied econometrics with {R}.
\newblock \url{https://CRAN.R-project.org/package=AER}.
\newblock {ISBN} 978-0-387-77316-2.

\bibitem[\protect\citeauthoryear{Koenker}{Koenker}{2022}]{quantregpackage}
Koenker, R. (2022).
\newblock quantreg: Quantile regression.
\newblock \url{https://CRAN.R-project.org/package=quantreg}.
\newblock R package version 5.94.

\bibitem[\protect\citeauthoryear{Laurent and Massart}{Laurent and
  Massart}{2000}]{laurent2000adaptive}
Laurent, B. and Massart, P. (2000).
\newblock Adaptive estimation of a quadratic functional by model selection.
\newblock {\em Annals of Statistics\/}, 1302--1338.

\bibitem[\protect\citeauthoryear{Liu, rong Zheng, and Feng}{Liu
  et~al.}{2020}]{Liu2020EstimationOE}
Liu, X., rong Zheng, S., and Feng, X. (2020).
\newblock Estimation of error variance via ridge regression.
\newblock {\em Biometrika\/}~{\bf 107}, 481--488.

\bibitem[\protect\citeauthoryear{Rosset and Tibshirani}{Rosset and
  Tibshirani}{2020}]{rosset20}
Rosset, S. and Tibshirani, R. (2020).
\newblock From fixed-x to random-x regression: Bias-variance decompositions,
  covariance penalties, and prediction error estimation.
\newblock {\em Journal of the American Statistical Association\/}~{\bf
  115\/}(529), 138--151.

\bibitem[\protect\citeauthoryear{Simon and Tibshirani}{Simon and
  Tibshirani}{2012}]{simon2012standardization}
Simon, N. and Tibshirani, R. (2012).
\newblock Standardization and the group lasso penalty.
\newblock {\em Statistica Sinica\/}~{\bf 22\/}(3), 983.

\bibitem[\protect\citeauthoryear{Thodberg}{Thodberg}{2015}]{Tecatordataset}
Thodberg, H.~H. (2015).
\newblock Tecator meat sample dataset.
\newblock \url{http://lib.stat.cmu.edu/datasets/tecator}.

\bibitem[\protect\citeauthoryear{Tibshirani}{Tibshirani}{1996}]{tibs96}
Tibshirani, R. (1996).
\newblock Regression shrinkage and selection via the lasso.
\newblock {\em J. Roy. Statist. Soc., Ser. B\/}~{\bf 58}, 267--288.

\bibitem[\protect\citeauthoryear{Wold}{Wold}{1966}]{wold66}
Wold, H. (1966).
\newblock Estimation of principal components and related models by iterative
  least squares.
\newblock In P.~R. Krishnajah (Ed.), {\em Multivariate Analysis}, pp.\
  391--420. New York: Academic Press.

\bibitem[\protect\citeauthoryear{Xie}{Xie}{1988}]{xie1988simple}
Xie, W.~Z. (1988).
\newblock A simple way of computing the inverse moments of a non-central
  chi-square random variable.
\newblock {\em Journal of econometrics\/}~{\bf 37\/}(3), 389--393.

\bibitem[\protect\citeauthoryear{Yuan and Lin}{Yuan and
  Lin}{2006}]{Grouplasso.YuanandLin}
Yuan, M. and Lin, Y. (2006).
\newblock Model selection and estimation in regression with grouped variables.
\newblock {\em Journal of the Royal Statistical Society: Series B (Statistical
  Methodology)\/}~{\bf 68\/}(1), 49--67.

\bibitem[\protect\citeauthoryear{Zhang and Zhou}{Zhang and
  Zhou}{2020}]{zhang2020non}
Zhang, A.~R. and Zhou, Y. (2020).
\newblock On the non-asymptotic and sharp lower tail bounds of random
  variables.
\newblock {\em Stat\/}~{\bf 9\/}(1), e314.

\bibitem[\protect\citeauthoryear{Zhang}{Zhang}{2010}]{zhang10}
Zhang, C.-H. (2010).
\newblock Nearly unbiased variable selection under minimax concave penalty.
\newblock {\em Annals of Statistics\/}~{\bf 38\/}(2), 894--942.

\bibitem[\protect\citeauthoryear{Zhao, Zhou, and Liu}{Zhao
  et~al.}{2023}]{zhao.linearfunctional}
Zhao, J., Zhou, Y., and Liu, Y. (2023).
\newblock Estimation of linear functionals in high-dimensional linear models:
  From sparsity to nonsparsity.
\newblock {\em Journal of the American Statistical Association\/}~{\bf 0\/}(0),
  1--13.

\bibitem[\protect\citeauthoryear{Zhu}{Zhu}{2020}]{zhuconvex2020}
Zhu, Y. (2020).
\newblock A convex optimization formulation for multivariate regression.
\newblock In H.~Larochelle, M.~Ranzato, R.~Hadsell, M.~Balcan, and H.~Lin
  (Eds.), {\em Advances in Neural Information Processing Systems}, Volume~33,
  pp.\  17652--17661. Curran Associates, Inc.

\bibitem[\protect\citeauthoryear{Zou}{Zou}{2006}]{zou06_adapt}
Zou, H. (2006).
\newblock The adaptive lasso and its oracle properties.
\newblock {\em J. Amer. Statist. Assoc.\/}~{\bf 101\/}(476), 1418--1429.

\end{thebibliography}

\end{alphasection}

\end{document}